\pdfoutput=1
\documentclass[10pt,english]{article}
\usepackage[unicode=true,pdfusetitle,bookmarks=true,bookmarksnumbered=false,bookmarksopen=false,breaklinks=false,pdfborder={0 0 1},colorlinks=true]{hyperref}
\usepackage[T1]{fontenc}
\usepackage[utf8]{inputenc}
\usepackage{geometry}
\usepackage{xcolor}

\usepackage{babel}
\usepackage{float}
\usepackage{amsmath}
\usepackage{amsthm}
\usepackage{amssymb}
\usepackage[noend]{algorithmic}
\usepackage[boxed]{algorithm}

\usepackage{xspace}
\usepackage[numbers,sort&compress]{natbib}

\colorlet{DarkRed}{red!75!black}
\colorlet{DarkGreen}{green!75!black}
\colorlet{DarkBlue}{blue!75!black}

\hypersetup{
	linkcolor = DarkRed,
	citecolor = DarkGreen,
	urlcolor = DarkBlue,
	bookmarksnumbered = true,
	linktocpage = true
}

\makeatletter

\floatstyle{ruled}
\newfloat{algorithm}{tbp}{loa}
\providecommand{\algorithmname}{Algorithm}
\floatname{algorithm}{\protect\algorithmname}

\theoremstyle{plain}
\newtheorem{thm}{\protect\theoremname}
  \theoremstyle{definition}
  
  \theoremstyle{plain}
  \newtheorem{lem}{\protect\lemmaname}
  \newtheorem{claim}{\protect\claimname}
  \newtheorem{fact}{Fact}

  \newtheorem{corollary}{Corollary}

\newcommand{\minpq}[3]{\mathrm{min}_{#1\rightarrow #2}(#3)}

\makeatletter
\def\norm#1{
  \@ifnextchar\bgroup%
   {\normalpnorm{#1}}%
   {\defaultnorm{#1}}%
}
\def\defaultnorm#1{%
    \renderNorm{#1}{}
}
\def\normalpnorm#1#2{%
   \@ifnextchar\bgroup%
   {\banaxnorm{#1}{#2}}%
   {\renderNorm{#2}{#1}}
}

\def\banaxnorm#1#2#3{%
    \renderNorm{#3}{#1\to#2 }
}

\def\renderNorm#1#2{%
    \@ifnextchar^%
    {\fixExponent{#1}{#2}}%
    {\ensuremath{\mathchoice%
        {\lVert #1 \rVert_{#2}}%
        {\lVert #1 \rVert_{#2}}%
        {\lVert #1 \rVert_{#2}}%
        {\lVert #1 \rVert_{#2}}}%
    }%
}
\def\fixExponent#1#2^#3{%
    \ensuremath{{\lVert #1 \rVert^{#3}_{#2}}}%
}

\def\enorm#1#2{%
   \@ifnextchar\bgroup%
   {\norm{L_{#1}}{L_{#2}}}%
   {\norm{L_{#1}}{#2}}%
}

\def\cnorm#1#2{%
   \@ifnextchar\bgroup%
   {\norm{\ell_{#1}}{\ell_{#2}}}%
   {\norm{\ell_{#1}}{#2}}%
}

\makeatother

\newcommand{\mysmalldot}[2]{\ensuremath{\langle #1, #2 \rangle}}

\newcommand{\rowspace}{\ensuremath{\mathrm{rowspace}}}
\newcommand{\cp}{\ensuremath{\mathrm{cp}}}

\newcommand{\hatf}{\widehat{f}}
\newcommand{\hatg}{\widehat{g}}

\newcommand{\bfa}{{\mathbf a}}

\newcommand{\bff}{{\mathbf f}}

\newcommand{\bfA}{{\mathbf A}}

\newcommand{\bfhatg}{{\bf \widehat{g}}}

\newcommand{\calL}{\mathcal{L}}
\newcommand{\calI}{\ensuremath{\mathcal{I}}}

\newcommand{\R}{{\mathbb R}}
\newcommand{\F}{{\mathbb F}}
\newcommand{\N}{{\mathbb N}}

\newcommand{\Holder}{H\"{o}lder\xspace}

\newcommand{\sgn}{\mathrm{sgn}}

\newtheorem*{rep@theorem}{\rep@title}
\newcommand{\newreptheorem}[2]{%
	\newenvironment{rep#1}[1]{%
		\def\rep@title{#2 \ref{##1}}%
		\begin{rep@theorem}[restated]}%
		{\end{rep@theorem}}}

\usepackage{enumerate}
\usepackage{paralist}

\usepackage{enumitem}

\setdescription{leftmargin=30pt}

\usepackage{lmodern}
\usepackage[T1]{fontenc}
\usepackage{multicol}

\makeatother

  \providecommand{\claimname}{Claim}
  \providecommand{\definitionname}{Definition}
  \providecommand{\lemmaname}{Lemma}
  
\providecommand{\theoremname}{Theorem}

\newcommand{\tC}{{\widetilde C}}
\newcommand{\tU}{{\widetilde U}}
\newcommand{\tV}{{\widetilde V}}

\newcommand{\tZ}{{\widetilde Z}}

\newcommand{\tE}{{\widetilde E}}

\newcommand{\ah}[1]{\alpha_{#1}}

\newcommand{\tD}{{\widetilde D}}

\newcommand{\fN}{\mathcal{N}}

\newcommand{\dM}{d_{M}}
\newcommand{\nuN}{\nu_{|\fN|}}

\newcommand{\CVy}{C^V_{y}}

\newcommand{\D}[1]{\mathcal{D}_{#1}}

\newcommand{\Cy}{C_{y}}
\newcommand{\tCy}{\widetilde{C}_{y}}
\newcommand{\tsy}{\widetilde{s}^{(y)}}

\newcommand{\CW}{C^W}
\newcommand{\CWx}{C^W_{x}}
\newcommand{\CWy}{C^W_{y}}
\newcommand{\CWz}{C^W_{z}}

\newcommand{\tR}{\widetilde{R}}

\newcommand{\colsupp}{\textup{ColSupp}}

\newcommand{\Wix}{\mathcal{W}_{i,x}}

\newcommand{\Yix}{Y_{i,x}}
	\newcommand{\hYix}{\hat{Y}_{i,x}}

\newcommand{\pr}[1]{\text{\bf Pr}\normalfont\lbrack #1 \rbrack} 
\newcommand{\ex}[1]{\mathbb{E}\normalfont\lbrack #1 \rbrack}
\newcommand{\exx}[1]{\mathbb{E}\normalfont\left[ #1 \right] }
\newcommand{\bpr}[1]{\text{\bf Pr}\normalfont \Big[#1 \Big]} 
\newcommand{\var}[1]{\text{Var}\normalfont\lbrack #1 \rbrack}

\newcommand{\LV}[1]{}
\newcommand{\SV}[1]{#1}
\newcommand{\eqdef}{\mathrel{\mathop:}=}
\newcommand{\Ex}{\mathbb{E}}
\newcommand{\Var}{\mathrm{Var}}
\newcommand{\E}{\Ex}
\newcommand{\opt}{\mathrm{OPT}}
\newcommand{\optk}{\mathrm{OPT}_{k}}
\newcommand{\optV}{\mathrm{OPT}_{k}^{V}}
\newcommand{\optW}{\mathrm{OPT}_{k}^{W}}
\newcommand{\poly}{\mathrm{poly}}
\newcommand{\nnz}[1]{\lVert #1 \rVert_0}
\newcommand{\nnzs}[1]{\Vert #1 \Vert_0}
\newcommand{\Def}{\overset{\mathrm{def}}{=}}
\newcommand{\eps}{\varepsilon}

\newcommand{\prob}{Binary $\ell_0$-Rank-$k$ Approximation\xspace}

\newcommand{\probbf}{Binary \boldmath$\ell_0$-Rank-\boldmath$k$ Approximation\xspace}

\newcommand{\lprobbf}{Entrywise $\ell_p$-Rank-$k$ Approximation\xspace}
\newcommand{\fieldprobbf}{Entrywise $\ell_0$-Rank-$k$ Approximation over $\mathbb{F}_q$\xspace}

\newcommand{\probk}{Binary $\ell_0$-Rank-$k$\xspace}
\newcommand{\genprobk}{Generalized Binary $\ell_0$-Rank-$k$\xspace}

\newtheorem{theorem}{Theorem}
\newtheorem{definition}[theorem]{Definition}
\newtheorem{hypothesis}[theorem]{Hypothesis}
\newtheorem{observation}[theorem]{Observation}

\newcommand{\classfont}[1]{\textsf{#1}}
\newcommand{\classP}{\classfont{P}\xspace}
\newcommand{\NP}{\classfont{NP}\xspace}

\newcommand{\med}{\text{med}}
\newcommand{\NNZ}{\text{nnz}}
\newcommand{\est}{\text{est}}

\DeclareMathOperator*{\argmin}{arg\,min}

\title{A PTAS for \texorpdfstring{$\ell_p$}{lp}-Low Rank Approximation}
\date{}

\author{Frank Ban\thanks{University of California, Berkeley, \texttt{fban@math.berkeley.edu}}\\
	\and
	Vijay Bhattiprolu\thanks{Carnegie Mellon University, \texttt{vpb@cs.cmu.edu}. Supported by NSF CCF-1422045 and CCF-1526092. Work partly done while visiting UC Berkeley and CMSA, Harvard.}\\
	\and
	Karl Bringmann\thanks{Max Planck Institute for Informatics, \texttt{kbringma@mpi-inf.mpg.de}}\\
	\and
	Pavel Kolev\thanks{Max Planck Institute for Informatics, \texttt{pkolev@mpi-inf.mpg.de}. Funded by the Cluster of Excellence ``Multimodal Computing and Interaction'' within the Excellence Initiative of the German Federal Government.}\\
	\and
	Euiwoong Lee\thanks{NYU, \texttt{luw0315@gmail.com}. Supported in part by the Simons Collaboration on Algorithms and Geometry. Work partly done while a research fellow at the Simons Institute.}\\
	\and
	David P. Woodruff\thanks{Carnegie Mellon University, \texttt{dwoodruf@andrew.cmu.edu}. Supported partly by Office of Naval Research (ONR) grant N00014-18-1-2562. Work partly done while visiting the Simons Institute.}}

\begin{document}

	\maketitle

	\begin{abstract}
	A number of recent works have studied  algorithms for entrywise $\ell_p$-low rank approximation,
    namely algorithms which given an $n \times d$ matrix $A$ (with $n \geq d$),
	output a rank-$k$ matrix $B$ minimizing $\|A-B\|_p^p=\sum_{i,j} |A_{i,j} - B_{i,j}|^p$ when $p > 0$; and $\|A-B\|_0 = \sum_{i,j} [A_{i,j} \neq B_{i,j}]$ for $p=0$, where $[\cdot]$ is the Iverson bracket, that is, $\|A-B\|_0$ denotes the number of entries $(i,j)$ for which $A_{i,j} \neq B_{i,j}$. For $p = 1$, this is often considered more robust than the SVD, while for $p = 0$ this corresponds to minimizing the number of disagreements, or robust PCA. This problem is known to be NP-hard for $p \in \{0,1\}$, already for $k = 1$, and while there are polynomial time approximation algorithms, their approximation factor is at best $\poly(k)$. It was left open if there was a polynomial-time approximation scheme (PTAS) for $\ell_p$-approximation for any $p \geq 0$. We show the following:

	\begin{enumerate}
		\item On the algorithmic side, for $p \in (0,2)$,
		we give the first $n^{\poly(k/\eps)}$ time $(1+\eps)$-approximation algorithm.
		For $p = 0$, there are various problem formulations,
		a common one being the binary setting in which $A\in\{0,1\}^{n\times d}$
		and $B = U\cdot V$, where $U\in\{0,1\}^{n \times k}$ and $V\in\{0,1\}^{k \times d}$.
		There are also various notions of multiplication $U \cdot V$,
		such as a matrix product over the reals, over a finite field, or over a Boolean semiring.
		We give the first almost-linear time approximation scheme
		for what we call the Generalized Binary $\ell_0$-Rank-$k$ problem,
		for which these variants are special cases.
		Our algorithm computes $(1+\eps)$-approximation in time
		$(1/\eps)^{2^{O(k)}/\eps^{2}} \cdot nd^{1+o(1)}$,
		where $o(1)$ hides a factor $\left(\log\log d\right)^{1.1}/\log d$.
		In addition, for the case of finite fields of constant size,
		we obtain an alternate PTAS running in time $n \cdot d^{\poly(k/\eps)}$.

		\item On the hardness front, for $p \in (1,2)$, we show under
		the Small Set Expansion Hypothesis and Exponential Time Hypothesis (ETH),
		there is no constant factor approximation algorithm running
		in time $2^{k^{\delta}}$ for a constant $\delta > 0$,
		showing an exponential dependence on $k$ is necessary.
		For $p = 0$, we observe that there is no approximation algorithm for the
		Generalized Binary $\ell_0$-Rank-$k$ problem running in time
		$2^{2^{\delta k}}$ for a constant $\delta > 0$.
		We also show for finite fields of constant size, under the ETH,
		that any fixed constant factor approximation algorithm requires
		$2^{k^{\delta}}$ time for a constant $\delta > 0$.
	\end{enumerate}
\end{abstract}
	\thispagestyle{empty}
	\newpage

	\thispagestyle{empty}
	\tableofcontents
    \newpage
	\setcounter{page}{1}

	\section{Introduction}
	\label{sec:intro}
	Low rank approximation is a common way of compressing a matrix via dimensionality reduction. The goal is to replace a given $n \times d$ matrix $A$ by a rank-$k$ matrix $A'$ that approximates $A$ well, in the sense that $\|A - A'\|$ is small for some measure $\|.\|$.
	Since we can write the rank-$k$ matrix $A'$ as $U \cdot V$, where $U$ is $n \times k$ and $V$ is $k \times d$, it suffices to store the $k(n+d)$ entries of $U$ and $V$, which is a significant reduction compared to the $nd$ entries of $A$. Furthermore, computing $A' x = U(V x)$ takes time $O(k(n+d))$, which is much less than the time $O(nd)$ for computing $A x$.

	Low rank approximation is extremely well studied, see the surveys~\cite{kv09, m11, w14}
	and the many references therein.
	In this paper, we study the following two variants of entrywise $\ell_p$-low rank approximation.
	Given a matrix $A$ and an integer $k$, one seeks to find a rank-$k$ matrix $A'$, minimizing
	$\|A-A'\|_p^p = \sum_{i,j} |A_{i,j} - A'_{i,j}|^p$ when $p > 0$;
	and $\|A-A'\|_0 = \sum_{i,j} [A_{i,j} \neq A'_{i,j}]$ for $p=0$,
	where $[\cdot]$ is the Iverson bracket, that is,
	$\|A-A'\|_0$ denotes the number of entries $(i,j)$ for which $A_{i,j} \neq A'_{i,j}$.

	When $p = 2$,
	this coincides with the Frobenius norm error measure, which can be
	solved in polynomial time using the singular value decomposition (SVD);
	see also \cite{w14} for a survey of more efficient algorithms based on
	the technique of linear sketching.

	Recently there has been considerable interest in obtaining algorithms
	for $p \neq 2$. For $0 \leq p < 2$, this error measure is often considered
	more robust than the SVD, since one pays less attention to noisy entries
	as one does not square the differences, but instead raises the difference
	to a smaller power. Conversely, for $p > 2$, this error measure pays
	more attention to outliers, and $p = \infty$ corresponds to a guarantee
	on each entry. This problem was shown to be NP-hard for $p \in \{0,1\}$
	\cite{GV15,dajw15,Pauli09}.

	{\it $\ell_p$-Low Rank Approximation for $p > 0$.}
	A number of initial algorithms for $\ell_1$-low rank approximation were given in
	\cite{kk03,kk05,kimefficient,kwak08,zlsyO12,bj12,bd13,bdb13,meng2013cyclic,mkp13,mkp14,mkcp16,park2016iteratively}. There is also related work on robust PCA \cite{wgrpm09,clmw11,nnasj14,nyh14,chd16,zzl15} and measures which
	minimize the sum of Euclidean norms of rows \cite{dv07,fmsw10,fl11,sv12,cw15b}, though neither directly
	gives an algorithm for $\ell_1$-low rank approximation.
	Song et al. \cite{swz17} gave the first approximation
	algorithms with provable guarantees
	for entrywise $\ell_p$-low rank approximation for $p \in [1,2)$.
	Their algorithm provides a $\poly(k \log n)$ approximation and runs in
	polynomial time, that is, the algorithm outputs a matrix $B$ for which
	$\|A-B\|_p \leq \poly(k \log n) \min_{\textrm{rank-}k \ A'}\|A-A'\|_p$.
	This was generalized by Chierichetti et al. \cite{CGOL17}
	to $\ell_p$-low rank approximation, for every $p \geq 1$, where the authors
	also obtained a $\poly(k \log n)$ approximation in polynomial time.

	In Song et al. \cite{swz17} it is also shown that if $A$ has entries bounded by $\text{poly}(n)$ then an $O(1)$ approximation can be achieved, albeit in $n^{\poly(k)}$ time. This algorithm depends inherently on the triangle inequality and as a result the constant factor of approximation is greater than $3$. Improving this constant of approximation requires techniques that break this triangle inequality barrier. This is a real barrier, since the algorithm of \cite{swz17} is based on a row subset selection algorithm, and
	there exist matrices for which any subset of rows contains at best
	a $2(1-\Theta(1/n))$-approximation (Theorem G.8 of \cite{swz17}),
	which we discuss more below.

	{\it $\ell_0$-Low Rank Approximation.}
	When $p = 0$, one seeks a rank-$k$ matrix $A'$ for which
	$\|A-A'\|_0$ is as small as possible, where for a matrix $C$, $\|C\|_0$
	denotes the number of non-zero entries of $C$. Thus, in this case,
	we are trying to minimize the number of disagreements between $A$ and $A'$.
	Since $A'$ has rank $k$,
	we can write it as $U \cdot V$ and we seek to minimize $\|A-U\cdot V\|_0$.
	This was studied by Bringmann et al. \cite{BKW17} when $A, U,$ and $V$ are matrices over the reals
	and $U \cdot V$
	denotes the standard matrix product, and the work of \cite{BKW17} provides a
	$\poly(k \log n)$ bicriteria approximation algorithm. See also earlier work for $k = 1$
	giving a 2-approximation \cite{sjy09,j14}.
	$\ell_0$-low rank approximation is also well-studied when $A$, $U$, and $V$
	are each required to be binary matrices. In this case, there are a number of choices for
	the ground field (or, more generally, semiring). Specifically, for $A' = U \cdot V$ we can
	write the entry $A'_{i,j}$ as the inner product of the $i$-th row of $U$ with the $j$-th column of $V$
	-- and the specific inner product function $\langle .,. \rangle$ depends on the ground field.
	We consider both (1) the ground field is $\mathbb{F}_2$ with inner product $\langle x,y \rangle = \bigoplus_{i=1}^k x_i \cdot y_i \in \{0,1\}$ \cite{y11,ggyt12,dajw15,prf15}, and (2) the Boolean semiring $\{0,1,\wedge,\vee\}$ in which the inner product becomes $\langle x,y \rangle = \bigvee_{i=1}^k x_i \wedge y_i = 1 - \prod_{i=1}^k (1-x_i \cdot y_i) \in \{0,1\}$ \cite{bv10,dajw15,mmgdm08,sbm03,sh06,vag07}.
	Besides the abovementioned upper bounds, which coincide with all of these models when $k = 1$, the only other
	algorithm we are aware of is by Dan et al.~\cite{dajw15}, who for arbitrary $k$ presented an $n^{O(k)}$-time $O(k)$-approximation over $\mathbb{F}_2$, and an $n^{O(k)}$-time $O(2^k)$-approximation over the Boolean semiring.

	Although $\ell_p$-low rank approximation is NP-hard for $p \in \{ 0, 1 \}$, a central open question is if $(1+\eps)$-approximation is possible, namely:
	\emph{Does $\ell_p$-low rank approximation have a polynomial time approximation scheme (PTAS) for any constant~$k$ and $\eps$?}

	\subsection{Our Results}
	We give the first PTAS for $\ell_p$-low rank approximation for $0 \leq p < 2$ in the unit cost RAM model of computation. For $p = 0$ our algorithms work for both finite fields and the Boolean semiring models. We also give time lower bounds, assuming the Exponential Time Hypothesis (ETH)~\cite{ImpagliazzoP01}
	and in some cases the Small Set Expansion Hypothesis~\cite{RS10}, providing evidence that an exponential dependence on $k$, for $p > 0$, and a doubly-exponential dependence on $k$, for $p = 0$, may be necessary.

	\subsubsection{Algorithms}
	We first formally define the problem we consider for $0 < p < 2$. We may assume w.l.o.g. that $n\geq d$, and thus the input size is $O(n)$.

	\begin{definition}(\textbf{\lprobbf})
		Given an $n \times d$ matrix $A$ with integer entries bounded
		in absolute value by $\poly(n)$, and a positive integer $k$,
		output matrices $U\in\mathbb{R}^{n \times k}$ and $V\in\mathbb{R}^{k \times d}$ minimizing $\|A-UV\|_p^p \eqdef \sum_{i=1, \ldots, n, j = 1, \ldots, d} |A_{i,j} - (U \cdot V)_{i,j}|^p$.
		An algorithm for \lprobbf is an $\alpha$-approximation if it outputs $U$ and $V$ for which
		$$\|A-UV\|_p^p \leq \alpha\cdot \min_{U' \in \mathbb{R}^{n \times k}, V' \in \mathbb{R}^{k \times d}} \|A-U'V'\|_p^p.$$
	\end{definition}

	Our main result for $0 < p < 2$ is as follows.
	\begin{thm}[PTAS for $0 < p < 2$] \label{thm:main1}
		For any $p \in (0,2)$ and $\eps \in (0,1)$,
		there is a $(1+\eps)$-approximation algorithm to \lprobbf running
		in time $n^{\poly(k/\eps)}$.
	\end{thm}

	For any constants $k\in\N$ and $\eps>0$, Theorem~\ref{thm:main1} computes a $(1+\eps)$-approximate solution to \lprobbf in polynomial time. This significantly strengthens the approximation guarantees in~\cite{swz17,CGOL17}.

	We next consider the case $p = 0$.
	In order to study the $\mathbb{F}_2$ and Boolean semiring settings
	in a unified way, we introduce the following more general problem.

	\begin{definition}(\textbf{\genprobk})
		Given a matrix $A\in\{0,1\}^{n\times d}$ with $n\geq d$, an integer $k$,
		and an inner product function
		$\langle .,. \rangle \colon \{0,1\}^k \times \{0,1\}^k \to \mathbb{R}$,
		compute matrices $U\in\{0,1\}^{n \times k}$ and $V\in\{0,1\}^{k \times d}$
		minimizing $\|A - U V\|_0$,
		where the product $U V$ uses $\langle .,. \rangle$.
		An algorithm for the \genprobk problem is an $\alpha$-approximation,
		if it outputs matrices $U\in\{0,1\}^{n \times k}$ and $V\in\{0,1\}^{k \times d}$
		such that
		$$\|A-U V\|_0 \leq \alpha\cdot
		\min_{U' \in \{0,1\}^{n \times k}, V' \in \{0,1\}^{k \times d}}
		\|A-U' V'\|_0.$$
	\end{definition}

	Our first result for $p = 0$ is as follows.
	\begin{thm}[PTAS for $p = 0$] \label{thm:main2}
		For any $\eps \in (0,\tfrac{1}{2})$, there is a $(1+\eps)$-approximation algorithm
		for the \genprobk problem running in time
		$(1/\eps)^{2^{O(k)}/\eps^2} \cdot nd^{1+o(1)}$
		and succeeds with constant probability~\footnote{
			The success probability can be further amplified to $1-\delta$
			for any $\delta>0$ by running $O(\log(1/\delta))$ independent
			trials of the preceding algorithm.},
		where $o(1)$ hides a factor $\left(\log\log d\right)^{1.1}/\log d$.
	\end{thm}

	Hence, we obtain the first almost-linear time approximation scheme for the \genprobk problem,
	for any constant $k$. In particular, this yields the first polynomial time
	$(1+\eps)$-approximation for constant $k$ for $\ell_0$-low rank approximation
	of binary matrices when the underlying field is $\mathbb{F}_2$ or the Boolean semiring.
	Even for $k=1$, no PTAS was known before.

	Theorem \ref{thm:main2} is doubly-exponential in $k$, and we show below that this is necessary
	for any approximation algorithm for \genprobk. However, in the special case when the
	base field is $\mathbb{F}_2$, or more generally $\mathbb{F}_q$ and $A, U,$ and $V$ have
	entries belonging to $\mathbb{F}_q$, it is possible to obtain an algorithm running in time
	$n \cdot d^{\poly(k/\eps)}$, which is an improvement for certain super-constant
	values of $k$ and $\eps$. We formally define the problem and state our result next.

	\begin{definition}(\textbf{\fieldprobbf})
		Given an $n \times d$ matrix $A$ with entries that are in $\mathbb{F}_q$ for any constant $q$,
		and a positive integer $k$, output matrices
		$U\in\mathbb{F}_q^{n \times k}$ and $V\in\mathbb{F}_q^{k \times d}$ minimizing $\|A-UV\|_0$.
		An algorithm for \fieldprobbf is an $\alpha$-approximation if it outputs matrices
		$U$ and $V$ such that
		$$\|A-UV\|_0 \leq \alpha\cdot
		\min_{U' \in \mathbb{F}_q^{n \times k}, V' \in \mathbb{F}_q^{k \times d}} \|A-U'V'\|_0.$$
	\end{definition}

	Our main result for \fieldprobbf is the following:
	\begin{thm}[Alternate $\mathbb{F}_q$ PTAS for $p = 0$] \label{thm:main3}
		For $\eps \in (0,1)$ there is
		a $(1+\eps)$-approximation algorithm to \fieldprobbf running
		in time $n \cdot d^{\poly(k/\eps)}$.
	\end{thm}

	\subsubsection{Hardness}
	We first obtain conditional time lower bounds for \lprobbf for $p \in (1,2)$.
	Our results assume the Small Set Expansion Hypothesis (SSEH).
	Originally conjectured by Raghavendra and Stuerer~\cite{RS10}, it is still the only assumption that implies strong hardness results for various graph problems such as
	Uniform Sparsest Cut~\cite{RST12} and Bipartite Clique~\cite{Manurangsi18}.
	Assuming this hypothesis, we rule out {\em any} constant factor approximation $\alpha$.

	\begin{thm}[Hardness for \lprobbf] \label{thm:hardp}
		Fix $p \in (1, 2)$ and $\alpha > 1$.
		Assuming the Small Set Expansion Hypothesis, there is no $\alpha$-approximation algorithm
		for \lprobbf that runs in time $\poly(n)$.

		Consequently, additionally assuming the Exponential Time Hypothesis,
		there exists $\delta := \delta(p, \alpha) > 0$ such that
		there is no $\alpha$-approximation algorithm for \lprobbf
		that runs in time $2^{k^{\delta}}$.
	\end{thm}

	This shows that assuming the SSEH and the ETH, any constant factor approximation algorithm needs at least a subexponential dependence on $k$.
	We also prove hardness of approximation results for $p \in (2, \infty)$ (see Theorem~\ref{thm:hardness_big_p}) without the SSEH.
	They are the first hardness results for Entrywise $\ell_p$-Rank-$k$ Approximation other than $p = 0, 1$.

	We next show that our running time for \genprobk is close to optimal,
	in the sense that the running time of any PTAS for \genprobk
	must depend exponentially on $1/\eps$ and doubly exponentially on $k$,
	assuming the Exponential Time Hypothesis.

	\begin{thm}[Hardness for \genprobk] \label{thm:hard0}
		Assuming the Exponential Time Hypothesis, \genprobk has no $(1+\eps)$-approximation algorithm
		in time $2^{1/\eps^{o(1)}} \cdot 2^{n^{o(1)}}$. Further, for any $\eps \geq 0$, \genprobk has no
		$(1+\eps)$-approximation algorithm in time $2^{2^{o(k)}} \cdot 2^{n^{o(1)}}$.
	\end{thm}

	Next we obtain conditional lower bounds for \fieldprobbf for any fixed $q$:
	\begin{thm}[Hardness for \fieldprobbf] \label{thm:hard0field}
		Let $\mathbb{F}_q$ be a finite field and $\alpha > 1$.
		Assuming $\classP \neq \NP$,
		there is no $\alpha$-approximation algorithm for \fieldprobbf
		that runs in time $poly(n)$.

		Consequently, assuming the Exponential Time Hypothesis, there exists $\delta := \delta(\alpha) > 0$ such that
		there is no $\alpha$-approximation algorithm for \fieldprobbf
		that runs in time $2^{k^{\delta}}$.
	\end{thm}

	This shows that assuming the ETH, any constant factor approximation algorithm needs at least a subexponential dependence on $k$.

	\subsubsection{Additional Results}
	We obtain several additional results on $\ell_p$-low rank approximation.
	We summarize our results below and defer the details to Section~\ref{sec:misc}.

	\paragraph{$\ell_p$-low rank approximation for $p > 2$}
	Let $g$ be a standard Gaussian random variable and let $\gamma_p := \E_{g} [|g|^p]^{1/p}$.
	We note that $\gamma_p>1$, for any $p > 2$. Then, under ETH
	no $(\gamma_p^p-\eps)$-approximation algorithm runs in time $O(2^{k^{\delta}})$.
	On the algorithmic side, we give a simple $(3+\eps)$-approximation algorithm running
	in time $n^{\poly(k/\eps)}$.

	\paragraph{Weighted $\ell_p$-low rank approximation for $0< p < 2$}
	We also generalize Theorem~\ref{thm:main1} to the following weighted setting.
	Given a matrix $A$, an integer $k$ and a rank-$r$ matrix $W$, we seek to find
	a rank-$k$ matrix $A'$ such that
	\[
	\|W \circ (A-A')\|_p^p \leq (1+\eps) \min_{\textrm{rank-}k \ A_k}\|W \circ (A -A_k)\|_p^p.
	\]
	Our algorithm runs in time $n^{r \cdot \poly(k / \eps)}$. We defer the details to Theorem~\ref{thm:Wp02}.

	\subsubsection*{Related Work}
	Our results, in particular Theorem~\ref{thm:main2} and Theorem~\ref{thm:hard0},
	had been in submission as of April~2018. Shortly after posting this 	manuscript~\cite{BBBKLW18_arXiv} to arXiv on 16~July~2018,
	we became aware that in an unpublished work Fomin et al.
	have independently obtained a very similar PTAS for \probk.
	Their manuscript~\cite{FGLPS18} was posted to arXiv on 18~July~2018.
	Interestingly, \cite{BBBKLW18_arXiv,FGLPS18} have independently discovered
	i) a reduction between the \probk problem and a clustering problem with constrained centers;
	ii) a structural sampling theorem extending~\cite{AlonS99} which yields a simple
	but inefficient deterministic PTAS; and
	iii) an efficient sampling procedure,
	building on ideas from~\cite{kumar2004simple,ABHKS05,ACMN2010},
	which gives an efficient randomized PTAS.
	Notably, by establishing an additional structural result, Fomin et al.~\cite{FGLPS18}
	design a faster sampling procedure which yields a randomized PTAS for the \probk problem
	that runs in linear time $(1/\eps)^{2^{O(k)}/\eps^{2}} \cdot nd$.

	\subsection{Our Techniques}
	We give an overview of our techniques, separating them into those for
	our algorithms for $0 < p < 2$, those for our algorithms for $p = 0$,
	and those for our hardness proofs.

	\subsubsection{Algorithms for \texorpdfstring{$0 < p < 2$}{0 < p < 2}}

	We illustrate the techniques for $p = 1$; the algorithms for other $p \in (0,2)$ follow similarly. Consider a target rank $k$. One of the
	surprising aspects of our $(1 + \eps)$-approximation result is that for $p = 1$, it breaks a potential lower bound from \cite{swz17}. Indeed, in Theorem G.8, they construct $(n-1) \times n$ matrices $A$ such that the closest rank-$k$ matrix $B$ in the row span of $A$ provides at best a $2(1-\Theta(1/n))$-approximation to $A$!

	This should be contrasted with $p = 2$, for which it is
	well-known that for any $A$ there exists a subset of $k/\eps$ rows of $A$ containing a $k$-dimensional
	subspace in its span which is a $(1+\eps)$-approximation (these are called column subset selection
	algorithms; see \cite{w14} for a survey).
	In fact, for $p = 1$, all known algorithms
	\cite{swz17,CGOL17}
	find a best $k$-dimensional subspace in either the span of the rows or of the columns of $A$, and thus
	provably cannot give better than a $2$-approximation. To bypass this, we therefore critically need to leave the row space and column space of $A$.

	Our starting point is the ``guess a sketch'' technique of \cite{rsw16}, which was used in the context of weighted low
	rank approximation. Let us consider the optimization problem $\min_V \|U^*V-A\|_1$, where $U^*$
	is a left factor of an optimal $\ell_1$-low rank approximation for $A$. Suppose we could choose a {\it sketching matrix}
	$S$ with a small number $r$ of rows for which $\|SU^*V-SA\|_1 = (1 \pm \eps)\|U^*V -A\|_1$ for all
	$V$. Then, if we somehow knew $U^*$, we could optimize for $V$ in the sketched space to find a good
	right factor $V$.

	Of course we do not know $U^*$, but if $S$ had a small number $r$ of rows,
	then we could consider instead the $\|\cdot\|_{1,2}$-norm optimization problem $\min_V \|SU^*V-SA\|_{1,2}$,
	where for a matrix $C$, $\|C\|_{1,2}$ is defined as $\sum_{i=1}^d \|C_{:,i}\|_2$, the sum of the $\| \cdot \|_2$-norms of its columns.
	The solution $V$ to $\min_V \|SU^*V-SA\|_{1,2}$ is a $\sqrt{r}$-approximation to the original problem
	$\min_V \|SU^*V-SA\|_1$.

	In the $\|\cdot\|_{1,2}$ norm, the solution $V$ can be written in terms of the so-called normal
	equations for regression, namely, $V = (SU^*)^{\dagger} SA$, where $C^\dagger$ denotes the Moore-Penrose
	pseudoinverse of $C$.
	The key property exploited in \cite{swz17} is then that although we do not know $U^*$,
	$(SU^*)^\dagger SA$ is a $k$-dimensional subspace in the row span of $SA$ providing a $\sqrt{r}$-approximation,
	and one does know $SA$. This line of reasoning ultimately leads to a $\poly(k)$-approximation.

	The approach above fails to give a $(1+\eps)$-approximation for multiple reasons: (1) we may not be able to find a
	$(1+\eps)$-approximation from the row span of $A$, and (2) we lose a $\sqrt{r}$
	factor when we switch to the $\|\cdot\|_{1,2}$ norm.

	Instead, suppose we were instead just to guess all the values of $SU^*$. These values might be arbitrary real numbers, but observe that we can assume
	there is an optimal solution $U^* V^*$ for which $V^*$ is a so-called $\ell_1$-well conditioned basis, which
	loosely speaking means that $\|yV^*\|_1 \approx \|y\|_1$ for any row vector $y$. Also,
	we can show that if $U^*V^* \neq A$, then $\|U^*V^*-A\|_1 \geq n^{-\Theta(k)}$. Furthermore, we can assume that the entries of $A$ are bounded by $\text{poly}(n)$. These three facts allow us to round the entries of $U^*$ to an integer multiple of $n^{-\Theta(k)}$ of absolute value at most $n^{O(k)}$. Now
	suppose we could also discretize the entries of $S$ to multiples of $n^{-\Theta(k)}$ and of absolute value
	at most $n^{O(k)}$. Then we would actually be able to guess the
	correct $SU^*$ after $n^{\Theta(k^2 r)}$ tries, where recall $r$ is the number of rows of $S$.
	We will show below that $r$ can be $\poly(k/\eps)$, so this will be within our desired running time.

	In general, if $\mathcal{A}(Ux) = (1 \pm \eps) \| Ux \|$ for all $x$, then we say that $\mathcal{A}$ defines a subspace embedding. At this point, we can use the triangle inequality to get a constant factor approximation. If $S$ is a subspace embedding, then $$\| U^* V - A \|_1 \leq \| U^* (V - V^*) \|_1 + \| U^* V^* - A \|_1 \leq (1 + O(\eps)) \| SU^* (V - V^*) \|_1 + \| U^* V^* - A \|_1$$ and $$\| SU^* (V - V^*) \|_1 \leq \| SU^* V - SA \|_1 + \| SU^* V^* - SA \|_1$$ so by taking $V$ to be a minimizer for $\| SU^* V - SA \|_1$ we can get an approximation factor close to $3$. The triangle inequality was useful here because $S$ had a small distortion on the subspace defined by $U^*$. To improve this result, we would need a mapping that has small distortion on the {\it{affine space}} defined by $U^* V - A$, as $V$ varies.

	Given $SU^*$ and $SA$, if in fact $S$ has the property that $\|SU^*V-SA\|_1 = (1 \pm \eps)\|U^*V-A\|_1$ for all $V$, then we will be in good shape. At this point
	we can solve for the optimal $V$ to $\min_V \| SU^*V-SA \|_1$ by solving an $\ell_1$-regression problem using
	linear programming. Notice that unlike \cite{rsw16}, the approach described above does not create
	``unknowns'' to represent the entries of $SU^*$ and set up a polynomial system of inequalities.
	For Frobenius norm error, this approach is feasible because $ \| SU^*V-SA \|_F^2 = \sum_{i = 1}^n  \| SU^*V_{:,i}-SA_{:,i} \|_F^2$ can be minimized over each column $V_{:,i}$ using the normal equations for regression. However, we do not know how to set up a polynomial system of inequalities for $\ell_1$-error
	(which define $V$ in terms of the $SU^*$ variables).

	Unfortunately the approach above is fatally flawed; there is no known sketching matrix $S$ with a small number
	$r$ of rows for which $\| SU^*V-SA \|_1 = (1 \pm \eps)\|U^*V-A\|_1$ for all $V$. Instead, we adapt a
	``median-based'' embedding with a non-standard subspace embedding analysis
	that appeared in the context of sparse recovery \cite{birw16}.
	In Lemma F.1
	of that paper, it is shown that if $L$ is a $d$-dimensional subspace of $\mathbb{R}^n$, and
	$S$ is an $r \times n$ matrix of i.i.d. standard Cauchy random variables for $r = O(d \eps^{-2} \log(d/\eps))$,
	then with constant probability,
	$(1-\eps)\|x\|_1 \leq \med (Sx) \leq (1+\eps)\|x\|_1$
	simultaneously for all $x \in L$.
	Here for a vector $y$, $\med (y)$ denotes the
	median of absolute values of its entries.
	For a matrix $M$, $\med (M)$ denotes the sum of the medians of its columns $\sum_i \med(M_{:,i})$.

	In our context, this gives us that for a fixed column $A_{:,i}$ of $A$
	and $i$-th column $V_{:,i}$ of $V$, if
	$S$ is an i.i.d. Cauchy matrix with $O(k \eps^{-2} \log(k/\eps))$ rows, then
	with constant probability $\med(SU^*V_{:,i} - SA_{:,i}) = (1 \pm \eps)\| U^*V_{:,i} - A_{:,i} \|_1$ for all vectors $V_{:,i}$.
	Since $V_{:,i}$ is only $k$-dimensional, and one can show that its entries can be taken to be integer multiples
	of $n^{-\poly(k)}$ bounded in absolute value by $n^{\poly(k)}$, we can enumerate over all $V_{:,i}$ and find the best
	solution. We need, however, to adapt the argument in \cite{birw16} to argue that if rather than taking the
	median, we take a $(1/2 \pm \eps)$-quantile, we still obtain a subspace embedding.

	Unfortunately, this still does not work. The issue is that $S$ succeeds only with constant probability
	in achieving $\med (SU^*V_{:,i} - SA_{:,i}) = (1 \pm \eps)\|U^*V_{:,i} - A_{:,i} \|_1$ for all vectors $V_{:,i}$. Call this property,
	of an index $i \in [n] \eqdef \{1, 2, \ldots, n\}$, {\it good}. A na\"ive amplification
	of the probability to $1-1/n$ would allow us to union bound over all $i$, but this would require $S$ to have
	$\Omega(\log n)$ rows. At this point though, we would not obtain a PTAS since enumerating the entries of $SU^*$
	would take $n^{\Omega(\log n)}$ time. Nor can we use different $S$ for different columns of $A$, since we may guess
	different $SU^*$ for different $i$ and not obtain a consistent solution $V$.

	Before proceeding, we first relax the requirement
	that $\med (SU^*V-SA) = (1 \pm \eps)\|U^*V-A\|_1$ for all $V$.
	We only need $\med (SU^*V-SA) \geq (1-\eps)\|U^*V-A\|_1$ for all $V$,
	and $\med (SU^*V^* - SA) \leq (1+\eps)\|U^*V^* - A\|_1$
	for the fixed optimum $U^*V^*$. We can prove  $\med (SU^*V^*-SA) \leq (1+\eps)\min_V \|U^*V -A\|_1$ by using tail
	bounds for a Cauchy random variable; we do so in Lemma \ref{fixedbound}.

	Moreover, we next argue that it suffices to have the properties:
	i) a $(1-\poly(\eps/k))$-fraction of columns are good, and
	ii) the error introduced by bad columns is small.
	We can achieve (i) by increasing the number of rows of $S$ by a $\log(k/\eps)$ factor,
	which still allows for an enumeration in time $n^{\poly(k/\eps)}$.
	The main issue is to control the error from bad columns.
	In particular, it is possible to have a matrix $V$ and a column $A_{:,i}$
	such that $\|U^*V_{:,i} - A_{:,i}\|_1$ is large and yet
	$\med (SU^*V_{:,i}-SA_{:,i})$ is small, which results in accepting a bad solution $V$.
	While for an average matrix $V$, the expected value of $\sum_{i \textrm{ is bad}}\|U^*V_{:,i}-A_{:,i}\|_1$ is small,
	we need to argue that this holds for every matrix $V$.

	In order to control the error from bad columns, we first show that
	$\med (SU^*V^*-SA) = (1 \pm \eps)\|U^*V^*-A\|_1$ for the fixed matrix $U^*V^*-A$,
	and then we demonstrate that the total contribution to $\|U^*V^*-A\|_1$ from bad columns, is small.
	We show the latter using Markov's bound for the fixed matrix $U^*V^*-A$.
	Combining this with the former, yields that the total contribution of
	$\med (SU^*V_{:,i}^*-SA_{:,i})$ to $\|SU^*V^*-SA\|_1$
	from bad columns (in the original, unsketched space) is small.

	We convert the preceding argument for bad columns of the fixed matrix $U^*V^*-A$,
	into an argument for bad columns of a general matrix $U^*V-A$.
	Inspired by ideas for $\|\cdot\|_{1,2}$ norm, established in \cite{cw15b},
	we partition the bad columns of a given matrix $V$ into classes,
	using the following measurement, which differs substantially from~\cite{cw15b}.
	We look at \emph{quantiles} to handle the median operator, and we say
	that a bad column $A_{:,i}$ is \emph{large} if
	\begin{equation}\label{eq:BadGuys}
		\lVert U^{*}V_{:,i}-A_{:,i}\rVert_{1}\geq\frac{1}{\eps}\left(\lVert U^{*}V_{:,i}^{*}-A_{:,i}\rVert_{1}+\frac{1}{1-O(\eps)}q_{1-\eps/2}(S(U^{*}V^{*}-A)_{:,i})\right),
	\end{equation}
	where $q_{1-\eps/2}$ is the $(1-\eps/2)$-th quantile of coordinates of column
	$S(U^*V^*-A)_{:,i}$ arranged in order of non-increasing absolute values.
	Otherwise, a bad column $A_{:,i}$ is \emph{small}.

	We show that small bad columns can be handled by applying the preceding argument
	for the fixed matrix $U^*V^*-A$, since intuitively, the error they introduce is dominated by
	the contribution of the corresponding columns of matrix $U^*V^*-A$,
	and we can control this contribution.

	Our analysis for the large bad columns uses a different approach, which we summarize in
	Claim \ref{largesketch}. The key insight is to use the additivity of a sketch matrix $S$,
	and to write
	\begin{eqnarray}\label{eqn:exp1}
		S(U^*V-A)_{:,i} = S(U^*V-U^*V^*)_{:,i} + S(U^* V^*-A)_{:,i}.
	\end{eqnarray}
	Then, by applying our ``robust'' version (Lemma~\ref{embedding}) of median-based
	subspace embedding~\cite{birw16}, it follows that at least a $(1/2+\eps)$-fraction
	of the entries of column vector $S(U^*V-U^*V^*)_{:,i}$ have absolute value at least
	\begin{eqnarray*}
		&  & (1-O(\eps))\cdot\|U^* (V-V^*)_{:,i}\|_1\\
		& \overset{(a)}{\geq} & (1-O(\eps))\cdot\Big(\|(U^* V-A)_{:,i}\|_1 - \|(U^* V^*-A)_{:,i}\|_1\Big)\\
		& \overset{(b)}{\geq} & \left(1-O(\eps)\right)\cdot\lVert(U^{*}V-A)_{:,i}\rVert_{1}+ 	q_{1-\eps/2}(S(U^{*}V^{*}-A)_{:,i}),
	\end{eqnarray*}
	where (a) follows by triangle inequality, and (b) by \eqref{eq:BadGuys} since the bad column
	$A_{:,i}$ is \emph{large}.
	Thus, at least a $(1/2+\eps)$-fraction of entries of $S(U^*V-U^*V^*)_{:,i}$
	have absolute value at least
	\begin{equation}\label{eq:AvrBadEntries}
		(1-O(\eps))\cdot\|(U^* V-A)_{:,i}\|_1 + q_{1-\eps/2}(S(U^* V^*-A)_{:,i}).
	\end{equation}

	Since at most an $\eps/2$ fraction of entries of $S(U^* V^*-A)_{:,i}$
	have an absolute value of at least $q_{1-\eps/2}(S(U^* V^*-A)_{:,i})$,
	by definition of quantile,
	it follows by~\eqref{eq:AvrBadEntries} that in equation~\eqref{eqn:exp1}
	at most an $\eps/2$-fraction of entries of $S(U^*V-A)_{:,i}$
	can have their absolute value reduced to less than $(1-O(\eps))\cdot\|(U^* V-A)_{:,i}\|_1$.
	Furthermore, by~\eqref{eq:AvrBadEntries} at least $(1/2+\eps/2)$-fraction of entries
	of $S(U^*V-U^*V^*)$ have absolute value at least $(1-O(\eps))\|(U^* V-A)_{:,i}\|_1$.
	Therefore, the median of absolute value of the entries of $S(U^*V-A)_{:,i})$
	is at least $(1-O(\eps))\|(U^* V-A)_{:,i}\|_1$, as desired.

	Our analysis for $0 < p < 2$ uses similar arguments, but in contrast relies on
	$p$-stable random variables. In the case when $0 < p < 1$, special care is needed
	since the triangle inequality does not hold.

	\newpage
	\subsubsection{Algorithms for \texorpdfstring{$p = 0$}{p = 0}}

	In the case when $p = 0$ and the entries of matrix $A$ belong
	to a finite field $\mathbb{F}_q$ for constant $q$,
	we use similar arguments as in the case for $p=1$.
	Here, instead of $p$-stable random variables
	we apply a linear sketch for estimating the number of distinct elements,
	established in \cite{KNW10}.
	We show that it suffice to set the number of rows of
	the sketching matrix $S$ to $\poly(k/\eps) \cdot \log d$.
	Further, since each entry of $S$ has only $q$ possible values,
	it is possible to guess matrix $S$ by enumeration in time
	$q^{\poly(k/\eps) \cdot \log d} = d^{\poly(k/\eps)}$,
	which will lead to a total running time of $n \cdot d^{\poly(k/\eps)}$.
	This yields a PTAS for constant $q$. We defer the details to Section~\ref{sec:upperbound}.

	We now consider the binary setting, where both the input matrix $A$ has entries in $\{0,1\}$ and
	any solution $U,V$ is restricted to have entries in $\{0,1\}$.
	In this case, the \genprobk problem can be rephrased as a clustering problem with constrained centers,
	whose goal is to choose a set of centers satisfying a certain system of linear equations,
	in order to minimize the total $\ell_0$-distance of all columns of $A$ to their closest center.
	The main difference to usual clustering problems is that the centers cannot be chosen independently.

	We view the choice of matrix $U$ as picking a set of ``cluster centers''
	$S_U \Def  \{ U \cdot y \mid y \in \{0,1\}^k \}$.
	Observe that any column of $U \cdot V$ is in $S_U$, and thus
	we view the choice of column $V_{:,j}$ as picking one of the constrained centers in $S_U$.
	Formally, we rephrase the \genprobk problem as
	\begin{align}
		\min_{U \in \{0,1\}^{n \times k}, V \in \{0,1\}^{k \times d}} \|A - U \cdot V\|_0 &
		= \min_{U \in \{0,1\}^{n \times k}} \sum_{j=1}^d \min_{V_{:,j} \in \{0,1\}^k}
		\|A_{:,j} - U \cdot V_{:,j} \|_0  \nonumber \\
		&= \min_{U \in \{0,1\}^{n \times k}} \sum_{j=1}^d
		\min_{s \in S_U} \| A_{:,j} - s \|_0.\label{eq:defClusteringFormulation}
	\end{align}

	Any matrix $V$ gives rise to a ``clustering'' as partitioning
	$C_V = (\Cy )_{y \in \{0,1\}^k}$ of the columns of $V$ with
	$\Cy  = \{ j \in [d] \mid V_{:,j} = y \}$.
	If we knew an optimal clustering $C=C_V$, for some optimal matrix $V$,
	we could compute an optimal matrix $U$ as the best response to $V$.
	Note that
	\[
	\min_{U\in\{0,1\}^{n \times k}} \sum_{y \in \{0,1\}^k} \sum_{j \in \Cy}\
	\nnz{A_{:,j}-U\cdot y} =
	\sum_{i=1}^{n} \min_{U_{i,:}\in\{0,1\}^k} \sum_{y \in \{0,1\}^k} \sum_{j \in \Cy}
	\nnz{A_{i,j}-U_{i,:}\cdot y}.
	\]
	Therefore, given $C$ we can compute independently for each $i\in[n]$
	the optimal row $U_{i,:}\in\{0,1\}^k$, by enumerating over all possible
	binary vectors of dimension $k$ and selecting the one that minimizes the
	summation
	\[
	\sum_{y \in \{0,1\}^k} \sum_{j \in \Cy}\nnz{A_{i,j}-U_{i,:}\cdot y}.
	\]

	What if instead we could only \emph{sample from} $C$?
	That is, suppose that we are allowed to draw a constant number $t = \poly(2^k/\eps)$ of
	samples from each of the optimal clusters $\Cy $ uniformly at random. Denote by $\tCy $
	the samples drawn from $\Cy $. A natural approach is to replace the exact cost above by
	the following unbiased estimator:
	\begin{align*}
		\widetilde{E} \Def  \sum_{y \in \{0,1\}^k} \frac{|\Cy |}{|\tCy |}
		\cdot \sum_{j \in \tCy } \|A_{:,j} - U \cdot y \|_0.
	\end{align*}

	We show that with good probability any matrix $U = U(\tC)$
	minimizing the estimated cost $\widetilde{E}$ is close to an optimal solution.
	In particular, we prove for any matrix $V\in\{0,1\}^{k\times d}$ that
	\begin{equation}\label{eq:samplingtheorem}
		\Ex_{\tC}[\| A - U(\tC) \cdot V \|_0] \le (1+\eps)\cdot
		\min_{U \in \{0,1\}^{n \times k}} \|A - U \cdot V\|_0.
	\end{equation}
	The biggest issue in proving statement~\eqref{eq:samplingtheorem} is that
	the number of samples $t = \poly(2^k/\eps)$
	is independent of the ambient space dimension $d$.
	A key prior probabilistic result, established by Alon and Sudakov~\cite{AlonS99},
	gives an additive $\pm \eps nd$ approximation for the maximization version of
	a clustering problem with unconstrained centers, known as Hypercube Segmentation.
	Since the optimum value of this maximization problem is always at least $nd/2$,
	a multiplicative factor $(1+\eps)$-approximation is obtained.
	Our contribution is twofold. First, we generalize their analysis to clustering problems
	with constrained centers, and second we prove a multiplicative factor $(1+\eps)$-approximation
	for the minimization version.
	The proof of~\eqref{eq:samplingtheorem} takes a significant fraction of this chapter.

	We combine the sampling result (\ref{eq:samplingtheorem}) with the following observations to obtain
	a deterministic polynomial time approximation scheme (PTAS) in time $n\cdot d^{\poly(2^k/\eps)}$.
	We later discuss how to further improve this running time.
	Let $U,V$ be an optimal solution to the \genprobk problem.

	\smallskip
	\begin{compactenum}[\mbox{}\hspace{\parindent}(i)]
		\item[(1)] To evaluate the estimated cost $\widetilde{E}$, we need the sizes $|\Cy |$ of
		an optimal clustering $C$. We can guess these sizes with an $d^{2^k}$ overhead in the running time.
		In fact, it suffices to know these cardinalities approximately, see Lemma~\ref{lem:condition},
		and thus this overhead~\footnote{
			In Section~\ref{sec:sampling}, we establish an efficient sampling procedure, see Algorithm~\ref{alg:Sampling}, that further reduces the total overhead for guessing
			the sizes $|\Cy |$ of an optimal clustering to
			$(2^k/\eps)^{2^{O(k)}}\cdot (\log d)^{(\log\log d)^{0.1}}$.}
		can be reduced to $(t+ \eps^{-1} \cdot \log d)^{2^k}$.
		\smallskip

		\item[(2)] Using the (approximate) size $|\Cy |$ and
		the samples $\tCy$ drawn u.a.r. from $\Cy$, for all $y\in\{0,1\}^{k}$,
		we can compute in time $2^{O(k)} nd$ a matrix $U(\tC)$
		minimizing the estimated cost $\widetilde{E}$,
		since the estimator $\widetilde{E}$ can be split
		into a sum over the rows of $U(\tC)$ and
		each row is chosen independently as a minimizer among all possible
		binary vectors of dimension $k$.
		\smallskip

		\item[(3)] Given $U(\tC)$, we can compute a best response matrix $V(\tC)$
		which has cost $\|A - U(\tC) \cdot V(\tC)\|_0 \le \|A - U(\tC) \cdot V\|_0$,
		and thus by \eqref{eq:samplingtheorem} the expected cost at most $(1+\eps) \opt$.
		\smallskip

		\item[(4)] The only remaining step is to draw samples $\tCy $ from the optimal clustering.
		However, in time $O(d^{2^k t}) = d^{\poly(2^k/\eps)}$ we can enumerate all possible families
		$(\tCy )_{y \in \{0,1\}^k}$, and the best such family yields a solution that is at least
		as good as a random sample. In total, we obtain a PTAS in time
		$n\cdot d^{\poly(2^k/\eps)}$.
	\end{compactenum}
	\smallskip\smallskip

	The largest part of this chapter is devoted to make the above PTAS efficient,
	i.e., to reduce the running time
	from $n\cdot d^{\poly(2^k/\eps)}$
	to $(2/\eps)^{2^{O(k)}/\eps^2} \cdot nd^{1+o(1)}$,
	where $o(1)$ hides a factor $\left(\log\log d\right)^{1.1}/\log d$.
	By the preceding outline, it suffices to speed up Steps (1) and (4), i.e.,
	to design a fast algorithm that guesses approximate cluster sizes and
	samples from the optimal clusters.

	The standard sampling approach for clustering problems such as $k$-means~\cite{kumar2004simple}
	is as follows. At least one of the clusters of the optimal solution is ``large'',
	say $|\Cy | \ge d/2^k$.
	Sample $t$ columns uniformly at random from the set $[d]$ of all columns.
	Then with probability at least $(1/2^k)^t$ all samples lie in $\Cy $,
	and in this case they form a uniform sample from  this cluster.
	In the usual situation without restrictions on the cluster centers,
	the samples from $\Cy $ allow us to determine an approximate cluster center~$\tsy$.
	Do this as long as large clusters exist (recall that we have guessed approximate
	cluster sizes in Step (1), so we know which clusters are large).
	When all remaining clusters are small, remove the $d/2$ columns that are closest to
	the approximate cluster centers $\tsy$ determined so far,
	and estimate the cost of these columns using the centers $\tsy$.
	As there are no restrictions on the cluster centers, this yields a good cost
	estimation of the removed columns, and since the $\ell_0$-distance is additive
	the algorithm recurses on the remaining columns, i.e. on an instance of twice smaller size.
	We continue this process until each cluster is sampled.
	This approach has been used to obtain linear time approximation schemes
	for $k$-means and $k$-median in a variety of ambient spaces~\cite{kumar2004simple,kumar2005linear,ACMN2010}.

	The issue in our situation is that we cannot fix a cluster center $\tsy$
	by looking only at the samples $\tCy$, since we have dependencies among cluster centers.
	We nevertheless make this approach work, by showing that a uniformly random column
	$r^{(y)}\in[d]$ is a good ``representative'' of the cluster $\Cy$ with not-too-small probability.
	In the case when all remaining clusters are small, we then simply remove the $d/2$ columns
	that are closest to the representatives $r^{(y)}$ of the clusters that we already sampled from.
	Although these representatives can be far from the optimal cluster centers due to
	the linear restrictions on the latter, we show in Section~\ref{sec:sampling}
	that nevertheless this algorithm yields samples from the optimal clusters.

	We prove that the preceding algorithm succeeds with probability at least
	$(\eps/t)^{2^{O(k)\cdot t}}$.
	Further, we show that the approximate cluster sizes $|\tCy|$ of an optimal clustering
	can be guessed with an overhead of
	$(2^k/\eps)^{2^{O(k)}}\cdot (\log d)^{(\log\log d)^{0.1}}$.
	In contrast to the standard clustering approach, the representatives $r^{(y)}$
	do not yield a good cost estimation of the removed columns.
	We overcome this issue by first collecting all samples $\tC$ from the optimal clusters,
	and then computing approximate cluster centers that satisfy certain linear constraints,
	i.e. a matrix $U(\tC)$ and its best response matrix $V(\tC)$.
	The latter computation runs in linear time $2^{O(k)}\cdot nd$ in the size of
	the \emph{original} instance, and this in combination with the guessing overhead,
	yields the total running time of $(2/\eps)^{2^{O(k)}/\eps^2} \cdot nd^{1+o(1)}$.
	For further details, we refer the reader to Algorithm~\ref{alg:Sampling}
	in Subsection~\ref{subsec:Pseudocode}.

	Our algorithm achieves a substantial generalization of the standard clustering approach
	and applies to the situation with constrained centers.
	This yields the first randomized almost-linear time approximation scheme
	for the \genprobk problem.

	\subsubsection{Hardness}
	Our hardness results for the $\ell_p$ norm for $p \in (1, 2)$ in Theorem~\ref{thm:hardp} and $p \in (2, \infty)$ in Theorem~\ref{thm:hardness_big_p} are established via a connection to the matrix $p \to q$ norm problem and its variants.
	Given a matrix $A \in \R^{n \times d}$, $\norm{p}{q}{A}$ is defined to be
	$\norm{p}{q}{A} := \max_{x \in \R^d, \norm{p}{x} = 1} \norm{q}{Ax}$.

	Approximately computing this quantity for various values of $p$ and $q$ has been known to have applications to the Small Set Expansion Hypothesis~\cite{BBHKSZ12}, quantum information theory~\cite{HM13}, robust optimization~\cite{Steinberg05}, and the Grothendieck problem~\cite{Grothendieck56}.
	After active research~\cite{HO10, BV11, BBHKSZ12, BGGLT18}, it is now known that computing the $p \to q$ norm of a matrix is NP-hard to approximate within some constant $c(p, q) > 1$ except when $p = q = 2$, $p = 1$, or $ q = \infty$.
	(Hardness of the case $p < q$ with $2 \in [p, q]$ is only known under stronger assumptions such as the Small Set Expansion Hypothesis or the Exponential Time Hypothesis.)
	See~\cite{BGGLT18} for a survey of recent results on the approximability of these problems.

	We also introduce the problem of computing the following quantity
	$$\minpq{p}{q}{A} := \min_{x \in \R^d, \norm{p}{x} = 1} \norm{q}{Ax}$$
	as an intermediate problem.
	Recall that $p^* = p / (p - 1)$ is the \Holder conjugate of $p$ for which $1/p + 1/p^* = 1$.
	The following lemma shows that computing $\ell_p$-low rank approximation when $k = d-1$ is equivalent to computing $\minpq{p^*}{p}{\cdot}$.

	\begin{lem}
		\label{lem:lp_to_minpq}
		Let $p \in (1, \infty)$.
		Let $A \in \R^{n \times d}$ with $n \geq d$ and $k = d - 1$.
		Then \[
		\min_{U \in \R^{n \times k}, V \in \R^{k \times d}} \norm{p}{U V - A} = \min_{x \in \R^d, \norm{p^*}{x} = 1} \norm{p}{Ax} = \minpq{p^*}{p}{A}.
		\]
	\end{lem}

	A simple but crucial observation for the above lemma is that if we let $a_1, \dots, a_n \in \R^d$ be the rows of $A$, computing the best $(d - 1)$-rank approximation of $A$ in the entrywise $\ell_p$ norm is equivalent to computing the $(d - 1)$-dimensional subspace $S \subseteq \R^d$ (i.e., $\rowspace(V) = S$) that minimizes $\norm{p}{(\rho_1, \dots, \rho_n)}$, where
	$\rho_i := \min_{y \in S} \norm{p}{y - a_i}$ denotes the $\ell_p$-distance between $S$ and $a_i$.

	If $x \in \R^d$ is a vector orthogonal to $S$, \Holder's inequality shows that \[
	\rho_i = \min_{y \in S} \norm{p}{y - a_i} = \min_{\langle x, z + a_i \rangle = 0} \norm{p}{z}
	\geq \frac{|\langle x, z \rangle|}{\norm{p^*}{x}} = \frac{|\langle x, a_i \rangle|}{\norm{p^*}{x}}.
	\]

	Taking $z$ to be the {\em \Holder dual} of $x$, we can show that indeed
	$\rho_i = |\langle x, a_i \rangle| / \norm{p^*}{x}$. Then $\norm{p}{(\rho_1, \dots, \rho_n)}$ equals $\norm{p}{Ax} / \norm{p^*}{x}$, finishing the lemma.

	This new connection allows us to prove a number of new hardness results for low rank approximation problems.
	Previously, even exact hardness results were known only for $p = 0, 1$ and there was no APX-hardness result.

	\paragraph{$\ell_p$ norm with $1 < p < 2$.}
	For $p \in (1, 2)$, we reduce computing $\norm{2}{p^*}{\cdot}$ to computing $\minpq{p^*}{p}{\cdot}$.

	If $A$ is an invertible matrix, then
	\[
	\minpq{p}{p^{*}}{A^{-1}}=\min_{x\neq0}\frac{\norm{p}{A^{-1}x}}{\norm{p^{*}}{x}}=
	\left(\max_{x\neq0}\frac{\norm{p^{*}}{x}}{\norm{p}{A^{-1}x}}\right)^{-1}\\
	=\left(\max_{y\neq0}\frac{\norm{p^{*}}{Ay}}{\norm{p}{y}}\right)^{-1}=
	\frac{1}{\norm{p}{p^{*}}{A}},
	\]
	and thus computing $\minpq{p^*}{p}{\cdot}$ is equivalent to computing $\norm{p}{p^*}{\cdot}$.

	By appropriately perturbing and padding $0$'s, we can show that computing the latter can be reduced to computing the former modulo arbitrarily small error.
	Standard facts from Banach spaces additionally show that $\norm{p \to p^*}{AA^T} = \norm{2 \to p^*}{A}^2$, proving the following lemma.

	\begin{lem} \label{lem:minpq_to_maxpq}
		For any $\eps >0, p\in (1,\infty)$, there is an algorithm that runs in $poly(n, \log (1/\eps))$ time and on a non-zero input matrix $A$,
		computes a matrix $B$ satisfying \[
		(1-\eps)\norm{2}{p^*}{A}^{-2} \leq \minpq{p^*}{p}{B} \leq (1+\eps)\norm{2}{p^*}{A}^{-2}. \]
	\end{lem}
	To finish Theorem~\ref{thm:hardp} for $\ell_p$-low rank approximation for $p \in (1, 2)$, we use the hardness of approximating the $2 \to q$ norm of a matrix proved by Barak et al.~\cite{BBHKSZ12} assuming the Small Set Expansion Hypothesis when $q = p^* > 2$.
	Given a $d$-regular graph $G = (V, E)$ and size bound $\delta \in (0, 1/2)$,
	the Small Set Expansion problem asks to find a subset $U \subseteq V$ with $|U| / |V| \leq \delta$ that minimizes
	$\Phi(U) = \frac{|E(U, V \setminus U)|}{d|U|} = 1 - (1_U)^T A (1_U)$, where $A$ and $1_U$ are the normalized adjacency matrix of $G$ and the normalized indicator vector of $U$, respectively.
	Consequently, the problem is equivalent to finding a sparse indicator vector $v$ with high Rayleigh quotient $v^T A v$, and one natural approach is to find a sparse vector in a subspace corresponding to large eigenvalues of $A$.
	For $q > 2$, since $\norm{q}{v} / \norm{2}{v}$ is maximized when $v$ is supported on only one coordinate and minimized when all entries of $v$ are equal in magnitude, $\norm{q}{v} / \norm{2}{v}$ is a natural analytic notion of sparsity, so if we let $P$ be the orthogonal projection on to the subspace corresponding to large eigenvalues, a high $\norm{2}{q}{P}$ seems to indicate that $G$ has a non-expanding small set.
	Barak et al. formalized this and proved the following theorem when $q \geq 4$ is an even integer, but the same proof essentially works for $q \in (2, \infty)$. For completeness, we present the proof in Section~\ref{sec:hardness}.

	\begin{thm}[\cite{BBHKSZ12}]
		\label{thm:hardness:SSE}
		Assuming the Small Set Expansion Hypothesis,
		for any $q \in (2, \infty)$ and $r > 1$, it is NP-hard to approximate the $\norm{2}{q}{\cdot}$ norm within a factor $r$.
	\end{thm}

	\paragraph{$\ell_p$ norm with $2 < p$.}
	Our hardness results for $p \in (2, \infty)$ are proved directly from the above intermediate problem.
	The following hardness result for $\minpq{p^*}{p}{\cdot}$ implies our hardness result for $p \in (2, \infty)$.
	It follows from a similar result by Guruswami et al.~\cite{GRSW16}, which proves the same hardness for the $\minpq{2}{p}{\cdot}$ norm,
	with some modifications that connect the $2$ norm and the $p^*$ norm. Recall that $\gamma_p := \E_{g} [|g|^p]^{1/p}$ where $g$ is a standard Gaussian, which is strictly greater than $1$ for $p > 2$.

	\begin{thm}\label{thm:hardness:NP}
		For any $p \in (2, \infty)$ and $\eps > 0$, it is NP-hard to approximate the $\minpq{p^*}{p}{\cdot}$ norm within a factor $\gamma_p - \eps$. 
	\end{thm}

	\paragraph{Finite Fields.}
	Our hardness results for finite fields rely on the following lemma.
	\begin{lem}
		\label{lem:hardness_fields}
		Let $\F$ be a finite field and $A \in \F^{n \times d}$ with $n \geq d$ and $k = d - 1$.
		Then, we have
		\[
		\min_{U \in \F^{n \times k}, V \in \F^{k \times d}} \norm{0}{U V - A} = \min_{x \in \F^d, x \neq 0} \norm{0}{Ax}.
		\]
	\end{lem}
	The proof has a similar structure to Lemma~\ref{lem:lp_to_minpq} for the $\ell_p$ norm in $\R$.
	We can still identify a subspace $S \subseteq \F^d$ with codimension $1$ with a vector $x$ with $\langle v, x \rangle = 0$ for every $v \in S$.
	In finite fields, $x$ can be possibly in $S$, but it does not affect the proof.
	Then for each row $a_i$ of $A$, if $\langle a_i, x \rangle = 0$, then $a_i \in S$ and we incur no error on the $i$th row. If $\langle a_i, x \rangle \neq 0$, changing one entry of $a_i$ will ensure that it will be contained in $S$, so the total number of errors  given $S$ is  exactly $\norm{0}{Ax}$.

	The quantity in the right-hand side, $\min_{x \in \F^d, x \neq 0} \norm{0}{Ax}$, is exactly the minimum Hamming weight of any non-zero codeword of the code that has $A^T$ as a generator matrix, or the minimum distance of the code. Then Theorem~\ref{thm:hard0field} above immediately follows from the following theorem by Austrin and Khot~\cite{AK11}.

	\begin{thm} [\cite{AK11}]
		For any finite field $\F$ and $r > 1$, unless $\classP = \NP$, there is no $r$-approximation algorithm for computing the minimum distance of a given linear code in polynomial time.
		\label{thm:dms}
	\end{thm}

	\paragraph{Paper Outline:} In Section \ref{sec:prelim} we give preliminaries. In Section \ref{sec:upperbound} we give our algorithms for $\ell_p$-low rank approximation, $0 < p < 2$, and since it is technically similar, our algorithm for $p = 0$ over finite fields. In Section \ref{sec:combinatorial} we give our algorithm for \genprobk. In Section \ref{sec:hardness} we give all of our hardness results. In Section \ref{sec:misc} we mention various additional results.

	\section{Preliminaries} \label{sec:prelim}

	For a matrix $A$ we write $A_{i,j}$ for its entry at position $(i,j)$, $A_{i,:}$ for its $i$-th row, and $A_{:,i}$ for its $i$-th column.

	For $0 \leq p \leq \infty$, we will let $\| A \|_p$ denote the entrywise $\ell_p$-norm of $A$. That is, $\| A \|_0$ equals the number of non-zero entries of $A$, $\| A \|_\infty = \max_{i, j} \lVert A_{i, j} \rVert$, and $\| A \|_p = (\sum_{i, j} A_{i, j}^p)^{1/p}$.

	For two matrices $A,B$ the value $\nnz{A-B}$ is a measure of similarity that is sometimes called their Hamming distance.

	We will typically give the dimensions of a matrix $A$ as $n \times d$ when $A$ has entries from a field such as $\mathbb{R}$ or $\mathbb{F}_q$. When the entries of $A$ are binary, we will typically give its dimensions as $m \times n$.

	We first recall some basic results about Cauchy variables. These have the property that if $x \in \R^n$ and $Z, C_i$ are i.i.d standard Cauchy variables (for $i = 1, \ldots, n$) then
	it holds that $\sum_{i=1}^n x_i C_i \sim \lVert x \rVert_1 Z$.

	\begin{fact} \label{cauchyfact}
		If $C$ is a Cauchy variable with scale $\gamma$, then \begin{enumerate}[label=(\roman*)]
			\item For $\tau > 1$, $\pr{\left| C \right| > \tau \gamma} \leq \frac{1}{\tau}$
			\item For small $\eps > 0$, $\pr{\left| C \right| > (1 + \eps) \gamma} < \frac{1}{2} - \Theta (\eps)$
			\item For small $\eps > 0$, $\pr{\left| C \right| < (1 - \eps) \gamma} < \frac{1}{2} - \Theta (\eps)$
		\end{enumerate}
	\end{fact}

	The following results are adapted from \cite{birw16}. We want to analyze the quantiles of the entries of a vector after a dense Cauchy sketch is applied to it.

	\begin{definition}
		Let $0 < \alpha < 1$. Let $v \in \R^m$. We let $q_{\alpha} (v)$ denote the $\frac{1}{\alpha}$-quantile of $\lvert v_1 \rvert$, $\lvert v_2 \rvert$, $\ldots$, $\lvert v_m \rvert$, or the minimum value greater than $\lceil \alpha n \rceil$ of the values $\lvert v_1 \rvert$, $\lvert v_2 \rvert$, $\ldots$, $\lvert v_m \rvert$. For $M \in \R^{m \times n}$,
		we let $$q_{\alpha} (M) \Def \sum_{i = 1}^n q_{\alpha} (M_{:,i}).$$
	\end{definition}

	We will be particularly interested in the median of the entries of a sketched vector.
	\begin{definition}
		For $v \in \R^n$, we write $\med(v)$ as shorthand for $q_{\frac{1}{2}} (v)$.
		Further, for $M \in \R^{m \times n}$, we let $$\med (M) \Def \sum_{i = 1}^n \med (M_{:,i}).$$
	\end{definition}

	\begin{lem} \label{quantiles}
		Let $S \in \R^{m \times n}$ have entries that are i.i.d. standard Cauchy variables and let $x \in R^n$. Then \begin{enumerate}[label=(\roman*)]
			\item $\pr{q_{\frac{1}{2} - \Theta(\eps)} (Sx) < (1 - \eps) \lVert x \rVert_1} < \exp (-\Theta(\eps^2) m)$
			\item $\pr{q_{\frac{1}{2} + O (\eps)} (Sx) > (1 + \eps) \lVert x \rVert_1} < \exp (-\Theta(\eps^2) m)$
			\item For $M > 2$, $\pr{q_{1 - \frac{\eps}{2}} (Sx) > \frac{M}{\eps} \lVert x \rVert_1} < \exp (- \Theta (\eps) M m)$
			\item For $M > 2$, $\pr{\med (Sx) > M \lVert x \rVert_1} < \exp (- \Theta(m) M)$
		\end{enumerate}
	\end{lem}

	\begin{proof}
		Note that for each $1 \leq i \leq m$, $(Sx)_i$ is distributed as a Cauchy variable with scale $\lVert x \rVert_1$. By Fact \ref{cauchyfact}, $\pr{\left| (Sx)_i \right| < (1 - \eps) \lVert x \rVert_1} < \frac{1}{2} - \Theta (\eps)$. We want to bound the probability that more than a $\frac{1}{2} - \Theta(\eps)$ fraction of the $(Sx)_i$'s are smaller than $(1 - \eps) \lVert x \rVert_1$. The desired upper bound follows from Chernoff's bound as $\exp (-\Theta(m) (\frac{1}{2} - \Theta(\eps) - (\frac{1}{2} - \Theta(\eps)))^2 )$, from which (i) follows. We can prove (ii) using a similar argument.

		For (iii), we know from Fact \ref{cauchyfact} that $\pr{(Sx)_i > \frac{M}{\eps}} < \frac{\eps}{M}$. Thus a Chernoff bound gives
		$$\pr{q_{1 - \frac{\eps}{2}} (Sx) > \tfrac{M}{\eps} \lVert x \rVert_1} < \exp (- \Theta (m) (\tfrac{\eps}{2} - \tfrac{\eps}{M})^2 (\tfrac{\eps}{M})^{-1})$$ and the result follows.
		For (iv), a similar proof holds using $\pr{(Sx)_i > M \lVert x \rVert_1} < \frac{1}{M}$.

	\end{proof}

	\begin{lem} \label{embedding}
		Let $X \subset \R^n$ be a $k$-dimensional space and $\eps > 0$. Let $S$ have $\Theta(\frac{1}{\delta} \cdot \text{poly}(\frac{k}{\eps}))$ rows, $n$ columns, and i.i.d. Cauchy entries with scale parameter $\gamma = 1$. Then with probability at least $1 - \Theta(\delta)$, for all $x \in X$, $$(1 - \Theta(\eps)) \lVert x \rVert_1 \leq q_{\frac{1}{2} - \eps} (Sx) \leq q_{\frac{1}{2} + \eps} (Sx) \leq (1 + O(\eps)) \lVert x \rVert_1  $$
	\end{lem}

	\begin{proof}
	    Let $u$ and $v$ be positive integers such that $S$ has a number of rows between $\frac{1}{\delta} \cdot \Omega(k^u / \eps^v)$ and $\frac{1}{\delta} \cdot O(k^{u+1} / \eps^{v+1})$.

	    Let $\eps' = \frac{\delta^2 \cdot \eps^{v+2}}{k^{u+3}}$ and let $N$ be an $\eps'$-net for the intersection of $X$ and the unit $\ell_1$ ball. Then $\lvert N \rvert = \exp(O(k \log \frac{k}{\eps \cdot \delta}))$.

		Let $y \in N$. By Lemma \ref{quantiles}, $\pr{q_{\frac{1}{2} - \Theta(\eps)} (Sy) < (1 - \eps) \lVert y \rVert_1} < \exp (-\Theta (\frac{k^u}{\delta \cdot \eps^{v-2}}))$. Thus, for all $y \in N$, $q_{\frac{1}{2} - \Theta(\eps)} (Sy) \geq 1 - \eps$ holds with probability $1 - \Theta(\delta)$ by a union bound.

		Let $X'$ be a matrix whose columns form an Auerbach basis (\cite{MM13}) for the subspace $X$. That is, each column of $X'$ has $\ell_1$ norm 1 and $\lVert z' \rVert_\infty \leq \lVert X' z' \rVert_1$ for all $z'$. Let $\tau = \frac{k^{u+2}}{\delta^2 \cdot \eps^{v+1}}$.

		By Fact \ref{cauchyfact}, each entry of $SX'$ is greater than $\tau$ with probability at most $\tau^{-1}$ because each column of $X'$ has $\ell_1$ norm $1$. A union bound tells us that $\lVert SX' \rVert_\infty \leq \tau$ with probability at least $1 - \Theta(\frac{k^{u+2}}{\eps^{v+1} \cdot \delta \cdot \tau})$ or $1 - \Theta(\delta)$.

		For arbitrary $z \in X$, we can write $z = X'z'$. Thus
		\begin{eqnarray*}
			\lVert Sz\rVert_{\infty}&=&\lVert SX'z'\rVert_{\infty}\leq\lVert SX'\rVert_{\infty}\cdot\lVert z'\rVert_{1}\leq \tau \cdot k\lVert z'\rVert_{\infty}\\&\leq&\frac{k^{u+3}}{\delta^2 \cdot \eps^{v+1}}\cdot\lVert X'z'\rVert_{1}= \frac{k^{u+3}}{\delta^2 \cdot \eps^{v+1}}\cdot\lVert z\rVert_{1}.
		\end{eqnarray*}

		Given any $x$ in the intersection of the unit $\ell_1$ ball and $X$, we can write $x = y + z$ where $y \in N$, $z \in X$, and $\lVert z \rVert_1 \leq \eps'$. By the above argument, we know $\lVert Sz \rVert_\infty \leq \frac{k^{u+3}}{\delta^2 \cdot \eps^{v+1}} \cdot \lVert z \rVert_1 \leq \eps$. Since $Sx = Sy + Sz$, then $(1 - \Theta (\eps)) \leq q_{\frac{1}{2} - \Theta(\eps)} (Sx)$ for any unit $x$. We can scale $x$ and $\eps$ by the appropriate constants to get the desired statement.

		The RHS inequality follows from a similar argument.
	\end{proof}

	We immediately have the following corollary about medians of Cauchy sketches over subspaces.
	\begin{corollary} \label{medembed}
		Let $X \subset \R^n$ be a $k$-dimensional space and $\eps, \delta > 0$. Let $S$ have $\Theta(\frac{1}{\delta} \cdot \text{poly}(\frac{k}{\eps}))$ rows, $n$ columns, and i.i.d. Cauchy entries with scale parameter $\gamma = 1$. With probability at least $1 - \Theta(\delta)$, for all $x \in X$, $$(1 - \eps) \lVert x \rVert_1 \leq \med(Sx) \leq (1 + \eps) \lVert x \rVert_1$$
	\end{corollary}

	We can also bound the median and the $(1-\eps/2)$-quantile of a Cauchy sketch of a fixed matrix.

	\begin{lem} \label{fixedbound}
		Let $S$ be an $m \times n$ matrix ($m = \Theta(1 / \text{poly}(\eps))$) with i.i.d. standard Cauchy entries and let $M$ be an $n \times d$ matrix. For $\eps > 0$, with probability $1 - 1/\Omega(1)$, $$(1 - \eps) \lVert M \rVert_1 \leq \med(SM) \leq (1 + \eps) \lVert M \rVert_1 $$
	\end{lem}

	\begin{proof}
		Lemma \ref{quantiles} tells us that we can choose $m$ so that $\pr{\med (SM_{:,i}) = (1 \pm \eps) \lVert M_{:,i} \rVert_1} \geq 1 - \Theta(\eps)$ for each $i$. Say $i$ is good if $\med (SM_{:,i}) \geq (1 - \eps) \lVert M_{:,i} \rVert_1$ and bad otherwise. Then $\ex{\sum_{\text{bad } i} \lVert M_{:,i} \rVert_1} \leq \eps \lVert M \rVert_1$ so Markov's inequality tells us $\sum_{\text{bad } i} \lVert M_{:,i} \rVert_1 \leq O(\eps) \lVert M \rVert_1$ with probability $1 - 1/\Omega(1)$ and also $\sum_{\text{good } i} \lVert M_{:,i} \rVert_1 \geq (1 - \Theta(\eps)) \lVert M \rVert_1 $.

		This implies that $$\med(SM) \geq \sum_{\text{good } i} \med(SM_{:,i}) \geq (1 - \eps) \sum_{\text{good } i} \lVert M_{:,i} \rVert_1 \geq (1 - \eps)(1 - \Theta(\eps)) \lVert M \rVert_1$$ which gives our first desired inequality.

		Now say that column $i$ is small if $\med (SM_{:,i}) < (1 + \eps) \lVert M_{:,i} \rVert_1$ and (for $k \geq 1$) $k$-large if $$ (k + 1 + \eps) \lVert M_{:,i} \rVert_1 > \med (SM_{:,i}) \geq (k + \eps) \lVert M_{:,i} \rVert_1.$$

		For $k \geq 3$, we can bound
		\begin{align*}
		    \pr{i \text{ is } k\text{-large}} &\leq \pr{\med (SM_{:,i} \geq (k + \eps) \cdot  \lVert M_{:,i} \rVert_1} < \exp (-\Theta(m) \cdot (k + \eps)) \\
		    &= \exp (-\Theta(m) \cdot \eps) \cdot \exp (-\Theta(m) \cdot k) < \eps \cdot \exp (-\Theta(m) \cdot k)
		\end{align*} where the second inequality comes from Lemma \ref{quantiles} and the last inequality comes from choosing $m = \Theta(1/\text{poly}(\eps))$.

		For $k = 1$ or $k = 2$, note that if $i$ is $k$-large, then $\med (SM_{:,i}) \geq (1 + \eps) \lVert M_{:,i} \rVert_1$ which occurs with probability at most $\Theta(\eps)$ as mentioned earlier.

		This lets us bound
		\begin{align*}
			\exx{ \sum_{k\geq1}k\sum_{k\text{-large }i}\lVert M_{:,i}\rVert_{1} }
			&\leq\Theta(\eps)\lVert M\rVert_{1}+2\Theta(\eps)\lVert M\rVert_{1}+\sum_{k\geq3}k\eps\exp(-\Theta(m)k)\lVert M\rVert_{1}&\\&\leq O(\eps)\lVert M\rVert_{1}\sum_{k\geq3}\frac{k}{\exp(\Theta(m)k)}\leq O(\eps)\lVert M\rVert_{1}&
		\end{align*} where the last inequality occurs because the given infinite series converges by the ratio test.

		Therefore \begin{align*} \med(SM) &= \sum_{\text{small }i} \med(SM_{:,i}) + \sum_{k \geq 1} \sum_{k\text{-large } i} \med(SM_{:,i})\\
			&\leq (1 + \eps) \lVert M \rVert_1 + \sum_{k \geq 1} (k + 1 + \eps) \sum_{k\text{-large } i}  \lVert M_{:,i} \rVert_1 \\
			&\leq (1 + \eps) \lVert M \rVert_1 + \sum_{k \geq 1} 3k \sum_{k\text{-large } i}  \lVert M_{:,i} \rVert_1 \\
			&\leq (1 + O(\eps))\cdot \lVert M \rVert_1
		\end{align*} where the first inequality holds by the definition of $k$-large and the third inequality holds with probability $1 - 1/\Omega(1)$ by Markov's inequality.
	\end{proof}

	\paragraph{Chebyshev's inequality.} We record some basic facts.
	Let $Z_1$,$\ldots$,$Z_n$ be independent Bernoulli random variables, with $Z_i \sim \textup{Ber}(p_i)$.
	Let $Z := Z_1 + \ldots + Z_n$ and $\mu := \Ex[Z]$.
	\begin{lem} \label{lem:cheb1}
		For any $\Delta > 0$, we have $\Pr[|Z - \mu| > \Delta] \le \mu/\Delta^2$.
	\end{lem}
	\begin{proof}
		By independence, we have
		$$\Var(Z) = \sum_{i=1}^n \Var(Z_i) = \sum_{i=1}^n p_i(1-p_i) \le \sum_{i=1}^n p_i = \mu.$$
		By Chebyshev's inequality, for any $\Delta > 0$ we have
		$$\Pr[|Z - \mu| > \Delta] \le \Var(Z) / \Delta^2.$$ With $\Var(Z) \le \mu$
		we thus obtain the claim.
	\end{proof}
	\begin{lem} \label{lem:cheb2}
		For any $\Delta > 0$, we have $\Pr[|Z - \mu| > \Delta] \le \sqrt{n}/\Delta$.
	\end{lem}
	\begin{proof}
		As in the previous lemma's proof, we have
		$$\Pr[|Z - \mu| > \Delta] \le \Var(Z) / \Delta^2,$$
		where $\Var(Z) \le \mu \le n$, and thus
		$$\Pr[|Z - \mu| > \Delta] \le n / \Delta^2.$$
		The statement follows since if $\sqrt{n}/\Delta < 1$
		we have $n/\Delta^2 \le \sqrt{n}/\Delta$,
		and otherwise the inequality is trivial.
	\end{proof}

	\newpage
	\section{\texorpdfstring{$\ell_p$}{lp}-Approximation Algorithms} \label{sec:upperbound}

	Recall that in the \lprobbf problem (for $0 < p < 2$) we are given an $n \times d$ matrix $A$ with integer entries bounded in absolute value by $\poly(n)$, a positive integer $k$, and we want to output matrices $U\in\mathbb{R}^{n \times k}$ and $V\in\mathbb{R}^{k \times d}$ minimizing $\|A-UV\|_p^p \eqdef \sum_{i=1, \ldots, n, j = 1, \ldots, d} |A_{i,j} - (U \cdot V)_{i,j}|^p$.
	In this section, we prove Theorem~\ref{thm:main1}, restated here for convenience.

	\begingroup
	\addtocounter{thm}{-1}
	\def\thethm{\ref{thm:main1}}
	\begin{thm}[PTAS for $0 < p < 2$]
		Let $p \in (0,2)$ and $\eps \in (0,1)$.
		There is a $(1+\eps)$-approximation algorithm to \lprobbf running
		in $n^{\poly(k/\eps)}$ time.
	\end{thm}
	\endgroup

	In Subsections~\ref{subsec:p0}, we prove in Corollary~\ref{PTAS:p1} our core algorithm result
	which solves the \lprobbf problem for $p=1$.
	In Subsection~\ref{subsec:bigp}, we give an algorithm for the case when $1 < p < 2$,
	and we prove its correctness in Corollary~\ref{PTAS:bigp}.
	In Subsection~\ref{subsec:smallp}, we prove in Corollary~\ref{PTAS:smallp} the correctness
	of our algorithm for $0 < p < 1$.
	Then, we conclude the proof of Theorem~\ref{thm:main1} by combining Corollary~\ref{PTAS:p1}, Corollary~\ref{PTAS:bigp} and Corollary~\ref{PTAS:smallp}.

	In Subsection~\ref{subsec:pover2}, we give a $(3 + \eps)$-approximation algorithm
	for the \lprobbf problem in the case when $p > 2$.

	Recall that in the \fieldprobbf problem we are given an $n \times d$ matrix $A$ with entries in $\mathbb{F}_q$, a positive integer $k$, and we want to output matrices
	$U\in\mathbb{F}_{q}^{n \times k}$ and $V\in\mathbb{F}_{q}^{k \times d}$ minimizing $\|A-UV\|_0$.
	In Subsection~\ref{subsec:finitefield}, we prove Theorem~\ref{thm:main3}
	and for reader's convenience we restate here our result.

	\begingroup
	\addtocounter{thm}{-1}
	\def\thethm{\ref{thm:main3}}
	\begin{thm}[Alternate $\mathbb{F}_q$ PTAS for $p = 0$]
		For $\eps \in (0,1)$ there is
		a $(1+\eps)$-approximation algorithm to \fieldprobbf running
		in $n \cdot d^{\poly(k/\eps)}$ time.
	\end{thm}
	\endgroup

	\subsection{\texorpdfstring{$\ell_1$}{l1}-Approximation Algorithm} \label{subsec:p0}


	In this subsection, $A_k$ will denote the rank $k$ matrix closest to $A$ in the entrywise $\ell_1$-norm.
	\noindent We will need a claim adapted from \cite{CW09}.

	\begin{claim} \label{lowerbound}
		If $A$ is $n$ by $d$ and has integer entries bounded by $\gamma = \text{poly}(n)$ and rank $r > k$, then we have $$\min_{\text{rank k } A_k} \lVert A - A_k \rVert_1 \geq \frac{1}{\text{poly}(n)^k}$$
	\end{claim}

	\begin{proof}
		Note that it suffices to lower bound $\sigma_{k+1}$, the $k$th singular value of $A$, because $\lVert A - A_k \rVert_1 \geq \lVert A - A_k \rVert_F \geq \sigma_{k+1}$.

		Since $A$ has integer entries, then so does $A^T A$ and its characteristic polynomial has integer coefficients. Now $A^T A$ has eigenvalues $\sigma_i^2$ so its characteristic polynomial's last term is $\prod_{i = 1}^r \sigma_i^2$ which is at least $1$ because it is a positive integer. For any $j$, $\sigma_j^2 \leq \lVert A \rVert^2_F \leq nd \gamma^2$.

		We have $$\sigma_{k+1}^{r-k} \geq \prod_{k < i \leq r} \sigma_i^2 \geq \frac{\prod_{1 \leq i \leq r} \sigma_i^2}{(nd \gamma^2)^{k}} \geq \frac{1}{(nd \gamma^2)^{k}}$$ so $\sigma_{k+1} \geq \frac{1}{(nd\gamma^2)^k}$ because $r - k \geq 1$.
	\end{proof}

	We can now describe our $(1 + \eps)$-approximation algorithm. For the rest of this section, let $U^*$ and $V^*$ be minimizers for $\lVert UV - A \rVert_1$ with $OPT = \lVert U^* V^{*} - A \rVert_1$. The quantities $\theta, \psi$ will be bounded above by $\text{poly}(n)$. The quantity $q$ will be bounded above by $\text{poly}(k)$. The specifics of how these values are chosen will be described in the algorithm's proof of correctness. The validity of the specific sampling described in Step 1 of Algorithm~\ref{alg:p1} will be proved in Corollary~\ref{PTAS:p1}.

	\begin{algorithm}[H]
		\caption{$(1+\eps)$-$\ell_1$ low rank approximation} \label{alg:p1}
		\textbf{Input:} A $n \times d$ matrix $A$ with integer entries bounded by $\gamma = \poly(n)$. An integer $k \in [d]$ and a real $\eps \in (0, 1)$.

		\textbf{Output:} Matrices $U \in \R^{n \times k}$ and $V \in \R^{k \times d}$ satisfying $\| UV - A \|_1 \leq (1 + \eps) OPT$.

		1. \textbf{If} $A$ has rank at most $k$, then \textbf{return} a rank $k$ decomposition $U, V$ of $A$.

		2. \textbf{Sample} an $m \times n$ matrix $S$ satisfying the conditions of Theorem~\ref{approxalg} (e.g. by taking $m = \poly(k / \eps)$ and sampling each entry of $S$ from a standard Cauchy distribution).

		3. \textbf{Round} the entries of $S$ to the nearest multiple of $\frac{\eps^2}{(\theta \psi)^k}$ where $\theta, \psi \leq \poly(n)$ are chosen as described in the proof of Theorem~\ref{approxalg}.

		4. \textbf{Set} $U$ and $V$ to be zero matrices as a default.

		5. Exhaustively \textbf{guess} all possible values of $SU^*$ with entries \textbf{rounded} to the nearest multiple of $\frac{\eps}{qnk \theta^k} \frac{\eps^2}{(\theta \psi)^k}$, where $q \leq \poly(k)$ is chosen as described in the proof of Theorem~\ref{approxalg}.

		$\quad$ 6. For each guessed $SU^*$, \textbf{set} $\tilde{V} = \argmin_V \med(SU^*V - SA)$ s.t. \indent $\lVert V \rVert_\infty \leq \frac{2nd\gamma qnk \theta^k}{\eps}$.

		$\quad$ 7. For each $\tilde{V}$, \textbf{set} $\tilde{U} = \argmin_U \lVert U \tilde{V} - A \rVert_1$.

		$\quad$ 8. \textbf{If} $\|\tilde{U} \tilde{V} - A \|_1 < \|UV - A \|_1$, then \textbf{set} $U = \tilde{U}, V = \tilde{V}$.

		9. \textbf{Return} $U, V$.
	\end{algorithm}

	\begin{thm} \label{approxalg}
		Let $A$ be an $n \times d$ matrix with integer entries such that $\lVert A \rVert_\infty$ is bounded by $\gamma = \text{poly}(n)$. Suppose $S$ is an $m \times n$ random matrix such that with probability $1 - 1/\Omega(1)$, $\med(SU^* V - SA) \geq (1 - \eps) \lVert U^* V - A \rVert_1$ for all $V$ and for a fixed $V^*$, $\med(SU^* V^* - SA) \leq (1 + \eps) \lVert U^* V^* - A \rVert_1$ with probability $1 - 1/\Omega(1)$. Suppose further that $\lVert S \rVert_\infty \leq \text{poly}(n)$. Then Algorithm \ref{alg:p1} is a $(1 + \eps)$-approximation algorithm for rank $k$ low rank approximation in the entrywise $\ell_1$ norm and runs in time $\text{poly}(n)^{mk}$.
	\end{thm}

	\begin{proof}
		First, if $A$ has rank at most $k$, then we can just use Gaussian elimination to deduce that its optimal low rank approximation has value 0. We will assume its rank is greater than $k$.

		We can assume $V^*$ is an $\ell_1$ well-conditioned basis since we can replace $U^*$ and $V^*$ with $U^*R$ and $V^*R^{-1}$ respectively for an invertible $R$. Thus for all $x$ we have $\frac{\lVert x \rVert_1}{q'} \leq \lVert x^T V^* \rVert_1 \leq q \lVert x \rVert_1$ where $q', q = \text{poly}(k)$. Using this well-conditioned basis property we see that each entry of $U^*$ is at most $2nd \gamma q' \leq \text{poly}(n)$ because otherwise $\lVert U^* V^* - A \rVert_1 \geq 2nd \gamma - \lVert A \rVert_1 \geq \lVert A \rVert_1$ and we could improve the $\ell_1$ error by taking $U^* = 0$.

		Claim \ref{lowerbound} says that there exists $\theta \leq \text{poly}(n)$ such that $OPT \geq \frac{1}{\theta^k}$. By using the well-conditioned basis property of $V^*$ and Claim \ref{lowerbound}, we can also assume that each entry of $U^*$ is rounded to nearest multiple of $\frac{\eps}{qnk \theta^k}$ as this will incur an additive error of at most $\eps OPT$. Thus $U^*$ has discretized and bounded entries. Note that there are at most $\eps^{-1} \text{poly}(n)^k$ possible values for each entry of $U^*$.

		Since the entries of $U^*$ are discretized by $\frac{\eps}{qnk \theta^k}$, then the entries of $V^*$ can be bounded above by $\frac{2nd\gamma qnk \theta^k}{\eps}$ because otherwise $\lVert U^* V^* - A \rVert_1 \geq 2nd\gamma - \lVert A \rVert_1 \geq \lVert A \rVert_1$ and we might as well have set $V^* = 0$.

		Let $V$ be arbitrary with $\lVert V \rVert_\infty \leq \frac{2nd\gamma qnk \theta^k}{\eps}$. Then $\lVert U^* V - A \rVert_1 \leq \eps^{-1} \psi^k$ where $\psi \leq \text{poly}(n)$. We will round each entry of $S$ to the nearest multiple of $\frac{\eps^2}{(\theta \psi)^k}$, so we can write $S = \tilde{S} + \Delta$ where $\tilde{S}$ is discretized and $\lVert \Delta \rVert_\infty \leq \frac{\eps^2}{(\theta \psi)^k}$. Note that $\lVert \Delta (U^* V - A) \rVert_1 \leq \frac{\eps}{\theta^k} \leq \eps OPT$.

		Now we will prove the correctness of our algorithm. We can sample $S = \tilde{S} + \Delta$. Note that $\tilde{S}U^*$ will have entries that are multiples of $\frac{\eps}{qnk \theta^k} \frac{\eps^2}{(\theta \psi)^k} \geq \text{poly}(\frac{\eps}{n})^k$ and bounded by $\text{poly}(\frac{n}{\eps})^k$ because $\tilde{S}$ is discretized and bounded. Since $\tilde{S}U^*$ is $m \times n$, then in $\text{poly}(\frac{n}{\eps})^k$ time we can exhaustively search through all possible values of $\tilde{S}U^*$ and one of them will be correct.

		For each guess of $\tilde{S}U^*$ and each $i$ we minimize $\med (\tilde{S}U^* V_{:,i} - \tilde{S}A_{:,i})$ over $\lVert V_{:,i} \rVert_\infty \leq \frac{2nd\gamma qnk \theta^k}{\eps}$ \footnote{Observe that there are at most $m!$ orderings of the entries of $\tilde{S}U^* V_{:,i} - \tilde{S}A_{:,i}$ and we are minimizing a linear function over $V_{:,i}$ subject to a linear constraint. This can be solved with linear programming, so it will be done within the $\poly(n)^{mk}$ runtime.} to get $\tilde{V_{:,i}}$. We have $\med(\tilde{S}U^* \tilde{V} - \tilde{S}A) \leq \med(\tilde{S}U^* V^* - \tilde{S}A)$.

		Now \begin{align*} \med(\tilde{S}U^* V^* - \tilde{S}A) &= \med(S(U^*V^* - A) - \Delta(U^* V^* - A)) \\
			&\leq \med(S(U^* V^* - A)) + \eps OPT \\
			&\leq (1 + \eps) \lVert U^* V^* - A \rVert_1 + \eps OPT \\
			&\leq (1 + O(\eps)) OPT.
		\end{align*}

		We choose $\tilde{U}$ to minimize $\lVert \tilde{U}\tilde{V} - A \rVert_1$, so \begin{align*} \med(\tilde{S}U^* \tilde{V} - \tilde{S}A) &= \med(S(U^*\tilde{V} - A) - \Delta(U^* \tilde{V} - A)) \\
			&\geq \med(S(U^*\tilde{V} - A)) - \eps OPT \\
			&\geq (1 - \eps) \lVert U^* \tilde{V} - A \rVert_1 - \eps OPT \\
			&\geq (1 - \eps) \lVert \tilde{U} \tilde{V} - A \rVert_1 - \eps OPT
		\end{align*}

		It follows that the best $\tilde{U}$ and $\tilde{V}$ will satisfy $\lVert \tilde{U} \tilde{V} - A \rVert_1 \leq (1 + O(\eps)) OPT$.
	\end{proof}

	Note that if $m = \Theta(\text{poly}(\frac{k \log d}{\eps}))$, then the above algorithm is a quasipolynomial time $(1 + \eps)$-approximation scheme (treating $k$ like a constant). This is because we can use Corollary \ref{medembed} (with $\delta = \text{poly}(1/d)$) to see that $$\med(S\left[U^* \ A_{:,i}] \left[V_{:,i} \ 1 \right]^T \right]) = (1 \pm \eps) \lVert \left[U^* \ A_{:,i}\right] \left[V_{:,i} \ 1 \right]^T \rVert_1$$ (when $V_{:,i}$ is arbitrary) with probability at least $1 - \text{poly}(1/d)$ for each $i$. By a union bound, $\med (S(U^*V - A)) = (1 \pm \eps) \lVert U^* V - A \rVert_1$ for arbitrary $V$ with probability $1 - 1/\Omega(1)$. Furthermore, Fact \ref{cauchyfact} tells us that $\pr{S_{i,j} \geq \text{poly}(n)} \leq \text{poly}(n)^{-1}$ so by a union bound, all entries of $S$ are bounded by $\text{poly}(n)$ with probability $1 - 1/\Omega(1)$.

	Of course, if we could reduce $m$ to $\Theta(\text{poly}(\frac{k}{\eps}))$, then we would have a PTAS. With the target bound for $m$, we would still have a $(1 \pm \eps)$-embedding for each column index $i$ with probability $1 - 1/\Omega(1)$, but we need all $d$ embeddings to be valid at once. We accomplish this in the next result which is a variant of Lemma 27 from \cite{cw15b}.

	\begin{thm} \label{constantsketch}
		Let $U \in \R^{n \times k}, A \in \R^{n \times d}$. Let $V^*$ be chosen to minimize $\lVert UV^* - A \rVert_1$. Suppose $S$ is an $m \times n$ matrix satisfying \begin{enumerate}[label=(\roman*)]
			\item $q_{\frac{1}{2} - \eps} (SUx) \geq (1 - \Theta(\eps)) \lVert Ux \rVert_1$
			\item For each $i$ with probability at least $1 - \eps^3$, $\med(S[U \ A_{:,i}]x) \geq (1 - \eps^3) \lVert [U \ A_{:,i}]x \rVert_1$ for all $x$
			\item $\med(SUV^* - SA) \leq (1 + \eps^3) \lVert UV^* - A \rVert_1$
			\item For $M > 2$, and each $i$ with probability $\exp(-\Theta(M) \cdot \text{poly} (k / \eps))$, $ q_{1 - \eps / 2}(S(UV^* - A)_{:,i}) > \frac{M}{\eps} \cdot \lVert (UV^* - A)_{:,i} \rVert_1 $
		\end{enumerate}

		Then $\med(SUV - SA) \geq (1 - O(\eps)) \lVert UV - A \rVert_1$ for arbitrary $V$.
	\end{thm}

	\begin{proof}

		We say a column index $i$ is $\it{good}$ if $$\med(S([U \ A_{:,i}]y)) \geq (1 - \eps^3)\lVert [U A_{:,i}]y \rVert_1$$ for all $y \in \R^{k+1}$, and $\it{bad}$ otherwise. We say a bad column index is $\it{large}$ if $$ \eps \lVert (UV - A)_{:,i} \rVert_1 \geq \frac{1}{1 - \eps}q_{1 - \eps / 2}(S(UV^* - A)_{:,i}) + \lVert (UV^* - A)_{:,i} \rVert_1 $$ and $\it{small}$ otherwise.

		By (ii), we know that $\ex{\sum_{\text{bad } i} \lVert (UV^* - A)_{:,i} \rVert_1} \leq \eps^3 \lVert UV^* - A \rVert_1$. By Markov's inequality, we know that with probability $1 - 1/\Omega(1)$, \begin{equation} \label{markovbadbound} \sum_{\text{bad } i} \lVert (UV^* - A)_{:,i} \rVert_1 \leq O(\eps^3) \cdot \lVert UV^* - A \rVert_1. \end{equation}

		By (iii) \begin{align*}(1 + \eps^3) \lVert UV^* - A \rVert_1
			&\geq \med (S(UV^* - A)) \\
			&\geq (1 - \eps^3) \sum_{\text{good } i} \lVert (UV^* - A)_{:,i} \rVert_1  + \sum_{\text{bad } i} \med (S(UV^* - A)_{:,i}) \\
			&\geq (1 - \eps^3) (1 - \Theta(\eps^3)) \lVert UV^* - A \rVert_1  + \sum_{\text{bad } i} \med (S(UV^* - A)_{:,i}), \\
		\end{align*}

		where the second inequality comes from the definition of good, and the third inequality comes from (\ref{markovbadbound}).

		Thus \begin{equation} \label{sketchbadbound} \sum_{\text{bad } i} \med(S(UV^* - A)_{:,i}) \leq O(\eps^3) \lVert UV^* - A \rVert_1 \end{equation}

		For $M > 2$, we say $i$ is $M$-$\it{large}$ if $$ \frac{M}{\eps} < q_{1 - \eps / 2}(S(UV^* - A)_{:,i}) \leq \frac{M + 1}{\eps} $$ and we say $i$ is $\it{tiny}$ if $q_{1 - \eps / 2}(S(UV^* - A)_{:,i}) \leq \frac{2}{\eps}.$

        Observe that \begin{align} \label{Mlargebound}
            \exx{\sum_{M > 2} \frac{M+1}{\eps} \cdot \sum_{M \text{-large } i} \lVert (UV^* - A)_{:,i} \rVert_1} &= \sum_{M > 2} \frac{M+1}{\eps} \cdot \sum_i \lVert (UV^* - A)_{:,i} \rVert_1 \cdot \pr{\text{$i$ is M-large}} \nonumber  \\
                    &\leq \sum_{M > 2} \frac{M+1}{\eps \cdot \exp(\Theta(M) \cdot \text{poly} (k / \eps))} \cdot \sum_i \lVert (UV^* - A)_{:,i} \rVert_1 \nonumber \\
                    &\leq \frac{1}{\eps \cdot \exp(1 / \eps)} \cdot \lVert (UV^* - A) \rVert_1 \cdot \sum_{M > 2} \frac{M+1}{\exp (\Theta (M))} \nonumber \\
                    &\leq O(\eps^2) \cdot \lVert (UV^* - A) \rVert_1
        \end{align} where the first equality follows from linearity of expectation and the first inequality follows from property (iv).

        We can deduce that with probability $1 - 1/\Omega(1)$, \begin{align} \label{quantilebadbound}
            \sum_{\text{bad } i} q_{1 - \eps / 2} (S(UV^* - A)_{:,i}) &= \sum_{\text{bad, tiny } i} q_{1 - \eps / 2} (S(UV^* - A)_{:,i}) \nonumber \\ &\indent + \sum_{M > 2}  \sum_{\text{bad, $M$-large } i} q_{1 - \eps / 2} (S(UV^* - A)_{:,i}) \nonumber \\
                &\leq \sum_{\text{bad } i} \frac{2}{\eps} \cdot \lVert (UV^* - A)_{:,i} \rVert_1 + \sum_{M > 2} \frac{M+1}{\eps} \cdot \sum_{\text{$M$-large } i} \lVert (UV^* - A)_{:,i} \rVert_1 \nonumber \\
                &\leq O(\eps^2) \cdot \lVert (UV^* - A) \rVert_1
        \end{align} where the first equality and inequality follow from the definitions of tiny and $M$-large, and the last inequality follows from (\ref{markovbadbound}) and (\ref{Mlargebound}) with Markov's inequality.

		We have \begin{align} \label{smallbound}
			\sum_{\text{small } i} \lVert (UV - A)_{:,i} \rVert_1 &\leq \frac{1}{\eps(1 - \eps)} \sum_{\text{small } i} q_{1 - \eps / 2} (S(UV^* - A)_{:,i}) + \frac{1}{\eps} \sum_{\text{small } i} \lVert (UV^* - A)_{:,i} \rVert_1 \nonumber \\
			&\leq \frac{1}{\eps(1 - \eps)} \sum_{\text{bad } i} q_{1 - \eps / 2} (S(UV^* - A)_{:,i}) + O(\eps^2) \lVert UV^* - A \rVert_1 \nonumber \\
			&\leq O(\eps) \lVert UV^* - A \rVert_1  + O(\eps^2) \lVert UV^* - A \rVert_1 \nonumber \\
			&\leq O(\eps)  \lVert UV^* - A \rVert_1
		\end{align} where the first inequality comes from the definition of small, the second inequality comes from (\ref{markovbadbound}) and the fact that small columns are bad columns, and the third inequality comes from (\ref{quantilebadbound}).

		\begin{claim} \label{largesketch}
			$$\sum_{\text{large } i} \med (S(UV - A)_{:,i}) \geq (1 - O(\eps)) \sum_{\text{large } i} \lVert (UV - A)_{:,i} \rVert_1$$
		\end{claim}

		\begin{proof}
			Let $i$ be large. We can write $S(UV - A)_{:,i} = SU(V - V^*)_{:,i} + S(UV^* - A)_{:,i}$.

			By (i), we know at least $\frac{1}{2} + \eps$ entries of $SU(V - V^*)_{:,i}$ are larger than $(1 - O(\eps)) \lVert U(V - V^*)_{:,i} \rVert_1$ which is at least $$(1 - O(\eps)) (\lVert (UV - A)_{:,i} \rVert_1 - \lVert (UV^* - A)_{:,i} \rVert_1)$$ by the triangle inequality. By the definition of large, this is at least $$(1 - O(\eps)) ((1 - \eps) \lVert (UV - A)_{:,i} \rVert_1 + \left( \frac{1}{1- \eps} \right) q_{1 - \eps / 2}(S(UV^* - A)_{:,i}))$$ or $$ (1 - O(\eps))^2 \lVert (UV - A)_{:,i} \rVert_1 + q_{1 - \eps / 2}(S(UV^* - A)_{:,i}).$$

			By definition, less than an $\eps / 2$ fraction of the entries of $S(UV^* - A)_{:,i}$ have an absolute value greater than $q_{1 - \eps / 2}(S(UV^* - A)_{:,i})$ so at least half of the entries of $S(UV - A)_{:,i}$ are greater than $(1- O(\eps))^2 \lVert (UV - A)_{:,i} \rVert_1$. The result follows.
		\end{proof}

		Finally \begin{align*}
			\med (S(UV - A)) &\geq \sum_{\text{good } i} \med (S(UV - A)_{:,i}) + \sum_{\text{large } i} \med (S (UV - A)_{:,i}) \\
			&\geq (1 - \eps^3) \sum_{\text{good } i} \lVert (UV - A)_{:,i} \rVert_1 + (1 - O(\eps)) \sum_{\text{large } i} \lVert (UV - A)_{:,i} \rVert_1 \\
			&\geq (1 - O(\eps)) \lVert UV - A \rVert_1 - (1 - O(\eps)) \sum_{\text{small } i} \lVert (UV - A)_{:,i} \rVert_1 \\
			&\geq (1 - O(\eps)) \lVert UV - A \rVert_1 - (1 - O(\eps)) O(\eps) \lVert UV^* - A \rVert_1 \\
			&\geq (1 - O(\eps)) \lVert UV - A \rVert_1
		\end{align*} where the first inequality occurs because large $i$ are bad $i$, the second inequality comes from the definition of good and Claim \ref{largesketch}, the third inequality comes from the definition of small, the fourth inequality comes from (\ref{smallbound}), and the last inequality holds because $V^*$ is a minimizer.
	\end{proof}

	\begin{corollary} \label{PTAS:p1}
		Let $A$ be an $n \times d$ matrix with integer entries bounded by $\poly(n)$ and let $k$ be a constant. There is a PTAS for finding the closest rank $k$ matrix to $A$ in the entrywise $\ell_1$ norm.
	\end{corollary}

	\begin{proof}
		Let $U^* \in \R^{n \times k}, V^* \in \R^{k \times d}$ be minimizers for $\lVert U^* V^* - A \rVert_1$. It suffices to prove that an $m \times n$ ($m = \Theta(\text{poly}(\frac{k}{\eps}))$) matrix $S$ with i.i.d. standard Cauchy entries satisfies the conditions of Theorem \ref{constantsketch} with $U = U^*$, then use Theorem \ref{approxalg}.

		Indeed, $S$ satisfies (i) through Lemma \ref{embedding} and (ii) with probability $1 - 1/\Omega(1)$ through Corollary \ref{medembed}. $S$ satisfies (iii) with probability $1 - 1/\Omega(1)$ via Lemma \ref{fixedbound} and (iv) via Lemma \ref{quantiles}.
	\end{proof}

	\subsection{\texorpdfstring{$1 < p < 2$}{1 < p < 2}} \label{subsec:bigp}
	We can extend these $\ell_1$ results to $\ell_p$ for $1 < p < 2$ by using $p$-stable variables (with scale 1) instead of Cauchy variables (or 1-stable variables). These have the property that if $x \in \R^n$ and $Z, Z_i$ are i.i.d $p$-stable variables (for $i = 1, \ldots, n$) then $\sum_{i=1}^n x_i Z_i \sim \lVert x \rVert_p Z$.

	\begin{definition}
		We let $\med_p$ denote the median of the absolute value of a $p$-stable variable.
	\end{definition}

	There is no convenient closed form expression for $\med_p$ unless $p = 1$, in which case $\med_1 = 1$. However, in Appendix A.2 of \cite{KNW10} it is shown that a $1 \pm \eps$ approximation of $\med_p$ can be computed efficiently. Since we are only interested in $\eps$ approximations, then this will suffice for our purposes. Our main sketch will be $\med \left( \frac{(Sx)}{\med_p} \right)$ ($S$ has i.i.d $p$-stable entries with scale 1) which will concentrate around $(1 \pm \eps) \lVert x \rVert_p$.

	We can cite similar concentration / tail bounds for $p$-stable variables like the ones we used for Cauchy variables. We can also state a series of claims analagous to the ones we used in the $\ell_1$ case.

	\begin{fact} \label{pstablefact}
		If $Z$ is a $p$-stable variable with scale $\gamma$, then \begin{enumerate}
			\item For $\tau > 1$, $\pr{\left| Z \right| > \tau \gamma \med_p} \leq \Theta (\frac{1}{\tau^p})$
			\item For small $\eps > 0$, $\pr{\left| Z \right| > (1 + \eps) \gamma \med_p} < \frac{1}{2} - \Theta (\eps)$
			\item For small $\eps > 0$, $\pr{\left| Z \right| < (1 - \eps) \gamma \med_p} < \frac{1}{2} - \Theta (\eps)$
		\end{enumerate}
	\end{fact}

	\begin{lem} \label{pquantiles}
		Let $S \in \R^{m \times n}$ have entries that are i.i.d. p-stable variables with scale 1 and let $x \in R^n$. Then \begin{enumerate}
			\item $\pr{q_{\frac{1}{2} - \Theta(\eps)} (Sx) < (1 - \eps) \lVert x \rVert_p \med_p} < \exp (-\Theta(\eps^2) m)$
			\item $\pr{q_{\frac{1}{2} + O (\eps)} (Sx) > (1 + \eps) \lVert x \rVert_p \med_p} < \exp (-\Theta(\eps^2) m)$
			\item For $M > 3$, $\pr{q_{1 - \frac{\eps}{2}} (Sx) > \frac{M}{\eps} \lVert x \rVert_p \med_p} < \exp (- \Theta (\eps) M m)$
			\item For $M > 3$, $\pr{\med (Sx) > M \lVert x \rVert_p \med_p} < \exp (- \Theta(m) M)$
		\end{enumerate}
	\end{lem}

	\begin{proof}
		The proof follows the same structure as the proof for Lemma \ref{quantiles}. We use Fact \ref{pstablefact} in combination with Chernoff bounds.
	\end{proof}

	Since $1 < p$, then we can take advantage of Minkowski's inequality and use the triangle inequality with $\lVert \cdot \rVert_p$.

	\begin{lem} \label{pembedding}
		Let $X \subset \R^n$ be a $k$-dimensional space and $\eps, \delta > 0$. Let $S$ have $O(\frac{1}{\eps^2}k \log \frac{k}{\eps \delta})$ rows, $n$ columns, and i.i.d. $p$-stable entries with scale 1. Then with probability at least $1 - \Theta(\delta)$, for all $x \in X$, $$(1 - \Theta(\eps)) \lVert x \rVert_p \leq q_{\frac{1}{2} - \eps} (Sx / \med_p) \leq q_{\frac{1}{2} + \eps} (Sx / \med_p) \leq (1 + O(\eps)) \lVert x \rVert_p $$
	\end{lem}

	\begin{proof}
		The proof follows the same structure as the proof for Lemma \ref{embedding} except we use Fact \ref{pstablefact} and $p$-well conditioned bases (\cite{DDHKM09}) to bound $\lVert (Sz) / \med_p \rVert_\infty$ for any $z \in X$. We also use the $\ell_p$ ball (which is still convex) instead of the $\ell_1$ ball.
	\end{proof}

	This automatically gives us the following corollary.
	\begin{corollary} \label{pmedembed}
		Let $X \subset \R^n$ be a $k$-dimensional space and $\eps, \delta > 0$. Let $S$ have $O(\frac{1}{\eps^2}k \log \frac{k}{\eps \delta})$ rows, $n$ columns, and i.i.d. $p$-stable entries with scale $1$. With probability at least $1 - \Theta(\delta)$, for all $x \in X$, $$(1 - \eps) \lVert x \rVert_p \leq \med(Sx / \med_p) \leq (1 + \eps) \lVert x \rVert_p$$
	\end{corollary}

	We also have an analogous version of our bound on fixed matrices. The proof structure is the same as that of Lemma \ref{fixedbound}.
	\begin{lem} \label{pfixedbound}
		Let $S$ be an $m \times n$ matrix ($m = \Theta(1 / \text{poly}(\eps))$) with i.i.d. standard $p$-stable entries and let $M$ be an $n \times d$ matrix. For $\eps > 0$, with probability $1 - 1/\Omega(1)$,
		$$(1 - \eps) \cdot \lVert M \rVert_p \leq \left(\sum_i \med(SM_{:,i})^p\right)^{1 / p} \Big/ \med_p \leq (1 + \eps) \cdot \lVert M \rVert_p $$
	\end{lem}

	Finally, we have an $\ell_p$ form of Theorem \ref{constantsketch} and it is proved analogously.

	\begin{thm}
		Let $U \in \R^{n \times k}, A \in \R^{n \times d}$. Let $V^*$ be chosen to minimize $\lVert UV^* - A \rVert_p$. Suppose $S$ is an $m \times n$ matrix satisfying \begin{enumerate}
			\item $q_{\frac{1}{2} - \eps} (SUx / \med_p) \geq (1 - \Theta(\eps)) \lVert Ux \rVert_p$
			\item For each $i$ with probability at least $1 - \eps^3$, $\med(S[U \ A_{:,i}]x / \med_p) \geq (1 - \eps^3) \lVert [U \ A_{:,i}]x \rVert_p$ for all $x$
			\item $(\sum_i \med(SUV^*_{:,i} - SA_{:,i})^p)^{1 / p} / \med_p \leq (1 + \eps^3) \lVert UV^* - A \rVert_p$
			\item $(\sum_i q_{1 - \eps / 2}(S(UV^* - A)_{:,i})^p)^{1/p} / \med_p \leq O \left( \frac{1}{\eps} \right) \lVert UV^* - A \rVert_p $
		\end{enumerate}

		Then $(\sum_i \med(SUV_{:,i} - SA_{:,i})^p)^{1/p} / \med_p \geq (1 - O(\eps)) \sum_i \lVert UV - A \rVert_p$ for arbitrary $V$.
	\end{thm}

	It follows that we have a PTAS for rank $k$ $\ell_p$ low rank approximation.

	\begin{corollary} \label{PTAS:bigp}
		Let $A$ be an $n \times d$ matrix with entries bounded by $\text{poly}(n)$ and let $k$ be a constant. There is a PTAS for finding the closest rank $k$ matrix to $A$ in entrywise $\ell_p$ norm for $1 < p < 2$.
	\end{corollary}

	\begin{proof}
		The algorithm is analogous to Algorithm \ref{alg:p1}. Correctness follows from the fact that there exist $\ell_p$ well-conditioned bases and that $\ell_p$ regression is a convex optimization problem.

		Indeed, if $p > 1$ then $\| UV_{:,i} - A_{:,i} \|_p$ is convex over vectors $V_{:,i}$ and we can calculate minima in polynomial time.
	\end{proof}

	\subsection{\texorpdfstring{$0 < p < 1$}{0 < p < 1}} \label{subsec:smallp}
	For $v \in \R^n$ we will denote $v^p$ to mean we raise each entry of $v$ to the $p$th power, i.e. $(v^p)_i = v_i^p$.

	We can extend these results to $\ell_p$ for $0 < p < 1$ as well, but more care needs to be taken for this range of $p$ because among other issues, $\lVert \cdot \rVert_p$ is no longer a norm. However, $\lVert \cdot \rVert_p^p$ satisfies the triangle inequality which will be enough for our purposes. We will prove that $\med \left( \frac{(Sx)^p}{\med_p^p} \right)$ ($S$ has i.i.d $p$-stable entries) will concentrate around $(1 \pm \eps) \lVert x \rVert_p^p$.

	\begin{lem} \label{bigpquantiles}
		Let $S \in \R^{m \times n}$ have entries that are i.i.d. p-stable variables with scale 1 and let $x \in R^n$. Then \begin{enumerate}
			\item $\pr{q_{\frac{1}{2} - \Theta(\eps)} (Sx)^p < (1 - \eps) \lVert x \rVert_p^p \med_p^p} < \exp (-\Theta(\eps^2) m)$
			\item $\pr{q_{\frac{1}{2} + O (\eps)} (Sx)^p > (1 + \eps) \lVert x \rVert_p^p \med_p^p} < \exp (-\Theta(\eps^2) m)$
			\item For $M > 3$, $\pr{q_{1 - \frac{\eps}{2}} (Sx)^p > \frac{M}{\eps} \lVert x \rVert_p^p \med_p^p} < \exp (- \Theta (\eps) M m)$
			\item For $M > 3$, $\pr{\med (Sx)^p > M \lVert x \rVert_p^p \med_p^p} < \exp (- \Theta(m) M)$
		\end{enumerate}
	\end{lem}

	\begin{proof}
		These results follow from Lemma \ref{pquantiles} and the fact that for $0 < p < 1$, we have $(1 - \eps)^p > 1 - \eps$ and $(1 + \eps)^p < 1 + \eps$.
	\end{proof}

	Using the above quantile results we can prove an embedding result similar to Lemma \ref{pembedding} by using the fact that $\lVert \cdot \rVert_p^p$ satisfies the triangle inequality.


	\begin{lem}
		Let $X \subset \R^n$ be a $k$-dimensional space and $\eps, \delta > 0$. Let $S$ have $O(\frac{1}{\eps^2}k \log \frac{k}{\eps \delta})$ rows, $n$ columns, and i.i.d. $p$-stable entries with scale 1. Then with probability at least $1 - \Theta(\delta)$, for all $x \in X$, $$(1 - \Theta(\eps)) \lVert x \rVert_p^p \leq q_{\frac{1}{2} - \eps} ((Sx)^p / \med_p^p) \leq q_{\frac{1}{2} + \eps} ((Sx)^p / \med_p^p) \leq (1 + O(\eps)) \lVert x \rVert_p^p $$
	\end{lem}

	This automatically gives us the following corollary.
	\begin{corollary}
		Let $X \subset \R^n$ be a $k$-dimensional space and $\eps, \delta > 0$. Let $S$ have $O(\frac{1}{\eps^2}k \log \frac{k}{\eps \delta})$ rows, $n$ columns, and i.i.d. $p$-stable entries with scale $1$. With probability at least $1 - \Theta(\delta)$, for all $x \in X$, $$(1 - \eps) \lVert x \rVert_p^p \leq \med((Sx)^p / \med_p^p) \leq (1 + \eps) \lVert x \rVert_p^p$$
	\end{corollary}

	We also have an analogous version of our bound on fixed matrices. Again, the proof structures is the same as that of Lemma \ref{fixedbound}.
	\begin{lem} \label{smallpfixedbound}
		Let $S$ be an $m \times n$ matrix ($m = \Theta(1 / \text{poly}(\eps))$) with i.i.d. standard $p$-stable entries and let $M$ be an $n \times d$ matrix. For $\eps > 0$, with probability $1 - O(1)$, $$(1 - \eps) \lVert M \rVert_p^p \leq \sum_i \med((SM_{:,i})^p / \med_p^p) \leq (1 + \eps) \lVert M \rVert_p^p $$
	\end{lem}

	As expected, we have an $\ell_p$ form of Theorem \ref{constantsketch} and it is proved analogously.

	\begin{thm}
		Let $U \in \R^{n \times k}, A \in \R^{n \times d}$. Let $V^*$ be chosen to minimize $\lVert UV^* - A \rVert_p^p$. Suppose $S$ is an $m \times n$ matrix satisfying \begin{enumerate}
			\item $q_{\frac{1}{2} - \eps} ((SUx)^p / \med_p^p) \geq (1 - \Theta(\eps)) \lVert Ux \rVert_p^p$
			\item For each $i$ with probability at least $1 - \eps^3$, $\med((S[U \ A_{:,i}]x)^p / \med_p^p) \geq (1 - \eps^3) \lVert [U \ A_{:,i}]x \rVert_p^p$ for all $x$
			\item $\sum_i \med((SUV^* - SA)_{:,i}^p / \med_p^p) \leq (1 + \eps^3) \sum_i \lVert (UV^* - A)_{:,i} \rVert_p^p$
			\item $\sum_i q_{1 - \eps / 2}(S(UV^* - A)_i^p / \med_p^p) \leq O \left( \frac{1}{\eps} \right) \sum_i \lVert (UV^* - A)_i \rVert_p^p) $
		\end{enumerate}

		Then $\sum_i \med((SUV - SA)_{:,i}^p / \med_p^p) \geq (1 - O(\eps)) \lVert (UV - A) \rVert_p^p$ for arbitrary $V$.
	\end{thm}

	The results above can give us the desired PTAS.

	\begin{corollary} \label{PTAS:smallp}
		Let $A$ be an $n \times d$ matrix with entries bounded by $\text{poly}(n)$ and let $k$ be a constant. There is a PTAS for finding the closest rank $k$ matrix to $A$ in entrywise $\ell_p$ norm when $0 < p < 1$.
	\end{corollary}

	\begin{proof}
		The algorithm is slightly different from Algorithm \ref{alg:p1}, because $\ell_p$ regression is no longer a convex optimization problem when $0 < p < 1$. Thus after sketching to find a minimizing $V$, we need a different approach to find a minimizing $U$. We accomplish this by sketching $UV - A$ again, but from the right and guessing the sketched $V$. We use the guessed $V$ to solve for $U$.

		Besides the above modification, we rely on the fact that $\lVert \cdot \rVert_p^p$ satisfies the triangle inequality. We also note that for $0 < p < 1$, we may not have a well-conditioned basis. However, we know that an $\ell_1$ well-conditioned basis exists so there exist $q, r = \text{poly}(k)$ such that $\frac{\lVert x \rVert_1}{q} \leq \lVert x^T V^* \rVert_1 \leq r \lVert x \rVert_1 $. By Holder's inequality, we know $\lVert x^T V^* \rVert_p^p \leq d^{1-p} \lVert x^T V^* \rVert_1^p \leq d^{1-p} r^p \lVert x \rVert_p^p$ and $\lVert x^T V^* \rVert_p^p \geq \lVert x^T V^* \rVert_1^p \geq \lVert x \rVert_1^p / q^p \geq d^{p -1} \lVert x \rVert_p^p / q^p$ so we can get a similar well-conditioned basis result saying there exist $\tilde{q}, \tilde{r} = \text{poly}(d)$ such that $\frac{\lVert x \rVert_p}{\tilde{q}} \leq \lVert x^T V^* \rVert_p \leq \tilde{r} \lVert x \rVert_p $ which will suffice for our proof.
	\end{proof}

	\subsection{\texorpdfstring{$p > 2$}{p > 2}} \label{subsec:pover2}

	There are no $p$-stable random variables when $p > 2$ so any $\ell_p$-approximation algorithms in this setting will need to rely on a different technique. Our sketch will be lifted from \cite{DDHKM09}. Rather than a matrix of $p$-stable random variables, we use a sampling matrix that samples $m$ rows of $A$ with each row $i$ having some probability $p_i$ of being sampled. Furthermore, each sampled row is reweighted by $1 / p_i$. The following claim (adapted from Theorem 5 of \cite{DDHKM09}) says we can get a subspace embedding from the right sampling matrix.

	\begin{claim}
		Suppose $U$ is an $n \times k$ matrix. Then there exists a $m \times n$ sampling matrix $S$ with $m = \text{poly}(k / \eps)$ such that $\| SUx \|_p = (1 \pm \eps) \| Ux \|_p$ for all $x$.
	\end{claim}

	\begin{thm}
		If $A$ is an $n \times d$ matrix with entries bounded by $\text{poly}(n)$, then there is a $(3 + \epsilon)$-approximation algorithm running in time $n^{\text{poly}(k / \eps)}$ for finding the closest rank $k$ matrix to $A$ in the entrywise $\ell_p$ norm for $p > 2$.
	\end{thm}

	\begin{proof}

		Let $S$ be the sampling matrix of the above claim. Let $\hat{V}$ be a minimizer for the expression $\lVert SU^* \hat{V} - SA \rVert_p$. Again, by a similar argument as that of the proof of Theorem \ref{approxalg}, we can guess $SU^*$ using $\text{poly}(n)$ tries. We can round the sampling probabilities and the entries of $U^*$ to the nearest $1 / \text{poly} (n)$ value.

		We know that \begin{align*} \lVert U^* \hat{V} - A \rVert_p &\leq \lVert U^* (\hat{V} - V^*) \rVert_p + \lVert U^* V^* - A \rVert_p \\
			&\leq (1 + O(\eps)) \lVert SU^* (\hat{V} - V^*) \rVert_p + \lVert U^* V^* - A \rVert_p \\
			&\leq (1 + O(\eps)) \lVert SU^* \hat{V} - SA \rVert_p + (1 + O(\eps)) \lVert SU^* V^* - SA  \rVert_p + \lVert U^* V^* - A \rVert_p  \\
			&\leq 2(1 + O(\eps)) \lVert SU^* V^* - SA  \rVert_p + \lVert U^* V^* - A \rVert_p \\
			&\leq (3 + \epsilon) \lVert U^* V^* - A \rVert_p
		\end{align*} where the second inequality follows from the embedding property of $S$ and the fourth inequality comes from the definition of $\hat{V}$ as a minimizer.

		The final inequality comes from a Markov bound on $S$. More specifically, since $S$ is a sampling matrix, then for an arbitrary matrix $M$, $\ex{SM} = \| M \|_p$. Thus Markov's Inequality says that with probability $1 - O(1)$, we have $SM \leq O(1) \| M \|_p$. This concludes the proof.

	\end{proof}

	\subsection{Finite Fields} \label{subsec:finitefield}

We can also study low rank approximation over finite fields. The $\ell_p$ metrics are not defined over finite fields for $p > 0$, but we can look at low rank approximation over the entrywise $\ell_0$ metric (where $\lVert M \rVert_0 = \lvert \{ (i, j) : M_{i, j} \neq 0 \}\rvert$). For the rest of this section we will work over a finite field $\mathbb{F}_q$, for some prime power $q$.

The structure of the algorithm will be similar to that of the case $0 < p < 2$ but our sketch will be based on hashing rather than $p$-stable random variables. Furthermore, we will be able to sketch in the dimension $d$ row space rather than the dimension $n$ column space and get a running time better than that of the $0 < p < 2$ algorithms. We now describe a $(1 + \eps)$-approximation sketch for the $\ell_0$ metric, where $\eps$ will be sufficiently small. This sketch is inspired by the $L_0$ streaming algorithm in \cite{KNW10}. Throughout this section, we will refer to constants $C$ and $C'$ that are sufficiently large.

Let $S_i$ denote a $n \times n$ matrix where column $i$ is the standard basis column $e_i$ with probability $p_i = \frac{1}{2^{i}}$ or the all zeroes column otherwise. In other words, $S_i$ is a sampling matrix that takes $x$ and preserves each coordinate with probability $p_i$ and otherwise maps the coordinate to $0$. Note that $p_0 = 1$. We can generate our matrices $S_i$ by uniformly sampling $n$ integers between 0 to $n$ and sampling column $j$ in $S_i$ if the leading $1$ in the $j$th integer (written in binary, with indexing starting from $1$) is before the $i$th position.  Observe that under this procedure, our subsampling is nested so that if $S_i$ does not sample entry $j$, then neither will $S_{i'}$ for any $i' > i$.

Note that by this nestedness property, we have $\lVert x \rVert_0 = \lVert S_0 x \rVert_0 \geq \lVert S_1 x \rVert_0 \geq \lVert S_2 x \rVert_0 \geq \cdots \geq \lVert S_{\log n - 1} x \rVert_0$. Let $S$ denote the $n \log n \times n$ block matrix $\begin{bmatrix} S_0 \\ S_1 \\ S_2 \\ \vdots \\ S_{\log n - 1} \end{bmatrix}$.

Let $h$ be a pairwise independent hashing function from $[n]$ to $[\frac{C'}{\eps^8}]$ and let $H_0$ denote a $\frac{C'}{\eps^8} \times n$ hashing matrix where each column equals $e_{h(i)}$. Let $H$ denote the $\frac{C'}{\eps^8} \log n \times n$ block matrix $\begin{bmatrix} H_0 S_0 \\ H_0 S_1 \\ H_0 S_2 \\ \vdots \\ H_0 S_{\log n - 1} \end{bmatrix}$ with $H^{(i)} = H_0 S_i$.

Suppose that $x = \begin{bmatrix} x^{(0)} \\ x^{(1)} \\ \vdots \\ x^{(\log n - 1)} \end{bmatrix}$ is a block vector. Then we let $\widetilde{\NNZ} (x)$ denote $\begin{bmatrix} \| x^{(0)} \|_0 \\ \| x^{(1)} \|_0 \\ \vdots \\ \| x^{(\log n - 1)} \|_0 \end{bmatrix}$.

We will abuse notation and let $\mathcal{C}^S (x) = \widetilde{\NNZ} (Sx)$ and $\mathcal{C} (x) = \widetilde{\NNZ} (HSx)$ with the understanding that $HSx$ and $Sx$ are of different dimensions but have the same number of blocks.

The main idea of the sketch is that if $\lVert x \rVert_0$ is less than a small constant and the coordinates of $x$ are hashed into a number of buckets that is a large constant, then with high probability it will be a perfect hash. Thus the number of non-zero buckets will equal $\lVert x \rVert_0$. If $x$ is subsampled with a low enough probability, then the subsampled vector will have an $\ell_0$ value that is sufficiently small and it can be hashed as we described.

We should note that the hash is needed for dimensionality reduction, not for the sketch to be an accurate estimator. For certain proofs we will analyze properties of the sketch without the hashing step (as in $\mathcal{C}^S (x)$).

So $S$ will sample $x$ with different subsampling probabilities and we will expect that one will be small enough. We can then hash that subsampled vector, count the number of non-zero entries, and rescale by the sampling probability to approximate $\lVert x \rVert_0$. It then suffices to identify a suitably subsampled vector.

To do so, we will let $\tau := \frac{C}{\eps^4}$ and define estimation functions $\est_j: \R^{\log n} \rightarrow \R$, where $\est_j (v) = \frac{v_{j}}{p_{j}}$. If $j^*$ denotes the maximum index such that $v_{j^*} > \gamma$ (for a value of $\gamma$ to be specified later) then $\est(v, \gamma) = \est_{j^*} (v)$. If such an index does not exist, then $\est(v, \gamma) = \est_0 (v)$. We let $\mathcal{E} (x, \gamma) = \est(\mathcal{C} (x), \gamma)$ and $\mathcal{E}_j (x) = \est_j (\mathcal{C} (x))$. We will also let $\mathcal{E}^S (x, \gamma) = \est(\mathcal{C}^S (x), \gamma)$ and $\mathcal{E}^S_j (x) = \est_j (\mathcal{C}^S (x))$.

Note that we can replace all instances of $n$ in the above definitions with $d$ and our algorithm will just sketch the row space rather than the column space. We use $n$ in our discussion just to keep the exposition similar to the case of $0 < p < 2$ and to emphasize the similarities in technique.

For ease of notation in our proofs, we will omit the parameter $\gamma$ in $\mathcal{E} (x)$, $\mathcal{E}^S (x)$, $\mathcal{E}_i (x)$, and  $\mathcal{E}^S_i (x)$ if it is clear that $\gamma = \tau$.

The idea is that past $j^*$ we can be confident that we are subsampling $x$ with so small of a probability that we barely sample any elements. On the other hand, if all the subsampled values are too small, then we can be confident that $\lVert x \rVert_0$ itself was small.

To sketch a vector it is enough to show that at the index $j^*$, a $p_{j^*}$ fraction of $x$ is sampled up to a relative error of $\eps$. For the purposes of our low rank approximation algorithm, we will want a slightly stronger condition that the indices around $j^*$ will be sampled ``as expected'' and that the value of $j^*$ will be approximately $\log (\lVert x \rVert_0 / \gamma)$.

Throughout this section, we will let $L_j$ denote $\lVert S_j x \rVert_0$ so $$\ex{L_j} = p_j \lVert x \rVert_0 \text{ and } \var{L_j} = p_j (1 - p_j) \lVert x \rVert_0 \leq \ex{L_j}.$$

\begin{definition}
		Given a threshold $\gamma$, let $j = \max(0,\lfloor \log_2 (\lVert x \rVert_0 / \gamma) \rfloor)$, so $\gamma \leq \frac{\lVert x \rVert_0}{2^{j}} < 2 \gamma$. Let $j^*$ be the maximum index such that $\mathcal{C}(x)_{j^*} \geq \gamma$, or $0$ if none exists.

		We say that $\mathcal{E} (x, \gamma)$ is a well-behaved sampling if \begin{enumerate}
			\item $j^* = j-1, j,$ or $j+1$
			\item If $\lVert x \rVert_0 \geq \gamma$, then $\mathcal{E}_i(x, \gamma) = (1 \pm \Theta(\eps)) \lVert x \rVert_0$ for $i = j-1, j, j+1$, and $j+2$
			\item If $\lVert x \rVert_0 < \gamma$, then $L_1 < 3 \gamma / 4$
		\end{enumerate}
	\end{definition}

To prove the correctness of our sketch, it will suffice to prove that with high probability our samplings are well-behaved samplings. We will need a folklore fact about pairwise independent hashing (the proof is included for completeness).

\begin{fact} \label{hashfact}
		If $h: [n] \to [m]$ is a pairwise independent hash function and $m \geq \Omega(n^2 / \eps)$, then with probability at least $1 - \Theta(\eps)$, $h$ will perfectly hash $[n]$.
	\end{fact}

\begin{proof}
		For $i \neq j \in [n]$ let $I_{i, j}$ be an indicator variable for the event $h(i) = h(j)$. Then $I = \sum_{i \neq j} I_{i,j}$ is the total number of collisions. We have $\ex{I} = \sum_{i \neq j} \ex{I_{i,j}} = \sum_{i \neq j} \frac{1}{m} \leq \frac{n^2}{m} \leq O(\eps)$. By Markov's Inequality, $\bpr{I \geq 1} \leq O(\eps)$ and the result follows.
	\end{proof}

To make use of this fact we will set $C'$ to be significantly larger than $C^2$. These hash sizes are chosen such that they are at least $\Omega(\gamma^2)$. Thus any subsampling past level $j^*$ will likely result in a perfect hashing.

\begin{lem} \label{wellbehaved}
		If $O(1 / \eps^4) > \gamma > \Omega(1 / \eps^3)$, then with probability at least $1 - \Theta(\eps)$ over the randomness of $S$ and $H$, $\mathcal{E}(x, \gamma)$ is a well-behaved sampling.

		In particular, this holds when $\gamma = \tau$ or $\gamma = \eps \tau$.
	\end{lem}

\begin{proof}
		We let $j^*$ and $v$ be as given in the definition of well-behaved. First we consider the case when $\lVert x \rVert_0 \geq \gamma$.

		Note that for $i = j-1, j, j+1,$ or $j+2$, we have $\ex{L_i} \geq \lVert x \rVert_0 / 2^{j+2} \geq \gamma / 4$. Since $\var{L_i} \leq \ex{L_i}$, then by Chebyshev's Inequality, we know $$\bpr{L_i \notin (1 \pm \eps) \ex{L_i}} \leq \left( \frac{\sqrt{\var{L_i}}}{\eps \ex{L_i}} \right)^2 \leq \frac{1}{\eps^2 \ex{L_i}} \leq \frac{4}{\eps^2 \gamma} \leq O(\eps).$$

		For the given values of $i$, we have $\ex{L_i} \leq \lVert x \rVert_0 / 2^{j-1} \leq 4 \gamma$. Since $H_0$ hashes to a range of size $C' / \eps^8 > (4 \gamma)^2$, then by Fact \ref{hashfact}, $H_0$ will perfectly hash the non-zero entries of $S_i x$ for the given values of $i$ with probability at least $1 - \Theta(\eps)$.

		By a union bound, $\mathcal{C}(x)_{i}  = (1 \pm \eps) \ex{L_i}$ for $i = j-1, j, j+1,$ or $j+2$ with probability at least $1 - \Theta(\eps)$. Thus $$\bpr{\mathcal{E}_i(x) = (1 \pm \eps) \lVert x \rVert_0} = \bpr{\mathcal{C}(x)_{i} = (1 \pm \eps) L_i} \geq 1 - \Theta(\eps)$$ which satisfies (ii).

		As we argued above, with probability at least $1 - \Theta(\eps)$ both $\mathcal{C}(x)_{j-1} \geq (1 - \eps) \ex{L_{j-1}} \geq 3\gamma / 2$ and $\mathcal{C}(x)_{j+2} \leq (1 + \eps) \ex{L_{j+2}} \leq 3\gamma / 4$ hold. By the nestedness of our sampling procedure, for any $i > j+2$ we have $\mathcal{C}(x)_{i} \leq 3 \gamma / 4$. Thus $j^* = j-1, j$, or $j+1$ which satisfies (i).

		Now suppose $\lVert x \rVert_0 < \gamma$. This implies $j = 0$ and $j^* = 0$ by definition which satisfies (i). If $\lVert x \rVert_0 \geq \gamma / 2$, then by our reasoning above, $L_1 < 3 \gamma / 4$ with probability at least $1 - \Theta(\eps)$. If $\lVert x \rVert_0 < \gamma / 2$, then $L_1 < \gamma / 2$ by the nestedness property of our sampling procedure. Therefore (iii) is satisfied.
\end{proof}

It follows that for a given $x$, with probability at least $1 - \Theta(\eps)$, $\mathcal{E}(x) = (1 \pm \eps) \lVert x \rVert_0$.

We can also get tail bounds for $\mathcal{E}_j(x)$ $(j = 1, \ldots, \log n)$ and $\mathcal{E}(x)$.

\begin{lem} \label{L0tail}
	Let $M$ be a large constant. Then $\bpr{\mathcal{E}_i(x) > M \lVert x \rVert_0} \leq \frac{1}{M}$ (for arbitrary $i$) and $\bpr{\mathcal{E}(x) > M \lVert x \rVert_0} \leq O \left( \frac{1}{M} + \eps \right)$. Furthermore, $\bpr{\mathcal{E}^S_i(x) > M \lVert x \rVert_0} \leq \frac{1}{M}$ and $\bpr{\mathcal{E}^S(x) > M \lVert x \rVert_0} \leq O \left( \frac{1}{M} + \eps \right)$.
\end{lem}

\begin{proof}
	Let $j^*$ be chosen so that $\mathcal{E}(x) = \frac{\mathcal{C}_{j^*} (x)}{p_{j^*}}$.

	By Markov's Inequality, we have

	\begin{align*} \bpr{\mathcal{E}_i(x) > M \lVert x \rVert_0} &= \bpr{\mathcal{C}_i(x) / p_i > M \lVert x \rVert_0} \\
		&\leq \bpr{L_i / p_i > M \lVert x \rVert_0} \\
		&\leq \frac{\ex{L_i}}{p_i M \lVert x \rVert_0} \\
		&= \frac{1}{M}
	\end{align*}

	for an arbitrary index $i$ as desired. Let $W$ denote the event that $\mathcal{E}(x)$ is well-behaved. Then Lemma~\ref{wellbehaved} tells us that

	\begin{align*} \bpr{\mathcal{E}(x) > M \lVert x \rVert_0} &= \bpr{W} \bpr{\mathcal{E}(x) > M \lVert x \rVert_0 \ | \ W} \\
        &\indent + \bpr{\overline{W}} \bpr{\mathcal{E}(x) > M \lVert x \rVert_0 \ | \ \overline{W}} \\
		&\leq \frac{\bpr{\mathcal{E}_i(x) > M \lVert x \rVert_0 \text{ for } i = j-1, j, j+1}}{\bpr{W}} + O(\eps) \\
		&\leq O \left( \frac{1}{M} + \eps \right)
	\end{align*} and the result follows.
\end{proof}

Let $K = \text{poly}(k, 1 / \delta, 1 / \eps)$ for some $\delta > 0$ and $\mathcal{E}^{(1)}, \ldots, \mathcal{E}^{(K)}$ be independent instances of the sketching procedure $\mathcal{E}$. Let $\mathcal{A} (x) = \begin{bmatrix} \mathcal{E}^{(1)} (x) \\ \vdots \\ \mathcal{E}^{(K)} (x) \end{bmatrix}$.


	For a matrix $M$, we let $\mathcal{A}(M)$ denote the matrix whose $i$th column is $\mathcal{A}(M_{:,i})$. 

	We can also define $\mathcal{A}^S(M)$ the natural way.

	We can study medians and quantiles of $\mathcal{A}(M)$ like we did the medians and quantiles of our sketches based on $p$-stable variables.

\begin{lem} \label{L0quantiles}
	\begin{enumerate}
		\item $\pr{q_{\frac{1}{2} - \Theta(\eps)} (\mathcal{A}(x)) < (1 - \eps) \lVert x \rVert_0} < \exp (-\Theta(\eps^2) K)$
		\item $\pr{q_{\frac{1}{2} + O (\eps)} (\mathcal{A} (x)) > (1 + \eps) \lVert x \rVert_0} < \exp (-\Theta(\eps^2) K)$
		\item For $T > 2$, $\pr{q_{1 - \frac{\eps}{2}} (\mathcal{A} (x)) > \frac{T}{\eps} \lVert x \rVert_0} < \exp (- \Theta (\eps) T K)$
		\item For $T > 2$, $\pr{\med (\mathcal{A} (x)) > T \lVert x \rVert_0} < \exp (- \Theta(T) K)$
	\end{enumerate}

	The analagous bounds for $\mathcal{A}^S (x)$ also hold.
\end{lem}

\begin{proof}
	We can use Chernoff bounds, Lemma \ref{wellbehaved}, and Lemma \ref{L0tail} to prove this in a similar way to how the proof of Lemma \ref{quantiles} used Chernoff bounds and the tail bounds on Cauchy sketches.
\end{proof}

We can now deduce a finite field subspace embedding result.
\begin{corollary} \label{L0medembed}
	Let $X \subset \F_q^n$ be a $k$-dimensional space. With probability at least $1 - \Theta(\delta)$, for all $x \in X$, $$(1 - \eps) \lVert x \rVert_0 \leq q_{\frac{1}{2} - \Theta(\eps)}(\mathcal{A}(x)) \leq  q_{\frac{1}{2} + O(\eps)}(\mathcal{A}(x)) \leq (1 + \eps) \lVert x \rVert_0$$ and $$(1 - \eps) \lVert x \rVert_0 \leq q_{\frac{1}{2} - \Theta(\eps)}(\mathcal{A}^S(x)) \leq  q_{\frac{1}{2} + O(\eps)}(\mathcal{A}^S (x)) \leq (1 + \eps) \lVert x \rVert_0$$

\end{corollary}

\begin{proof}
	We can use Lemma \ref{L0quantiles} and the fact that $|X| = q^k$ to deduce the result with a union bound.
\end{proof}

We can also bound the median of an $\ell_0$ sketch of a fixed matrix.
\begin{lem} \label{L0fixedbound}
	Let $M$ be an $n \times d$ matrix. For $\eps > 0$, with probability $1 - O(1)$, $$(1 - \eps) \lVert M \rVert_0 \leq \med(\mathcal{A}(M)) \leq (1 + \eps) \lVert M \rVert_0 $$ and $$(1 - \eps) \lVert M \rVert_0 \leq \med(\mathcal{A}^S(M)) \leq (1 + \eps) \lVert M \rVert_0 $$
\end{lem}

\begin{proof}
	The proof follows the same structure as the proof of Lemma \ref{fixedbound} where we bound the expected sum of $\med(\mathcal{A}(M_{:,i}))$ over values of $i$ where $\med(\mathcal{A}(M_{:,i}))$ is large and conclude with Markov's Inequality. Instead of Fact \ref{cauchyfact} and Lemma \ref{quantiles}, we use Lemma \ref{wellbehaved} and Lemma \ref{L0quantiles}.
\end{proof}

\begin{theorem} \label{L0constantsketch}
	Let $U \in \F_q^{n \times k}, A \in \F_q^{n \times d}$. With probability $1 - O(1)$, $$\med(\mathcal{A}(UV - A)) \geq (1 - O(\eps)) \lVert UV - A \rVert_0$$ for arbitrary $V$.
\end{theorem}

\begin{proof}
Let $V^*$ be chosen to minimize $\lVert UV^* - A \rVert_0$.

	For column indices $i$, let $J_i = \max(0, \log(\lVert (UV^* - A)_{:,i} \rVert_0 / \gamma))$

	By Lemmas \ref{wellbehaved} and \ref{L0fixedbound}, the following statements hold with probability $1 - O(1)$: \begin{enumerate}[label = (\roman*)]
		\item $\mathcal{E}$ is well-behaved on $Ux$ for all $x$
		\item For each $i$ with probability at least $1 - \eps^3$, $$\med(\mathcal{A}([U \ A_{:,i}]x)) \geq (1 - \eps^3) \lVert [U \ A_{:,i}]x \rVert_0$$ for all $x$
		\item $\med(\mathcal{A}(UV^* - A)) \leq (1 + \eps^3) \lVert UV^* - A \rVert_0$
		\item For $T > 2$, and each $i$ with probability $\exp(-\Theta(T) \cdot \text{poly} (k / \eps))$, $ q_{1 - \eps / 2}(\mathcal{A}(UV^* - A)_{:,i}) > \frac{T}{\eps} \cdot \lVert (UV^* - A)_{:,i} \rVert_0 $
	\end{enumerate}

	We say a column index $i$ is $\it{good}$ if $$\med(\mathcal{A}([U \ A_{:,i}]y)) \geq (1 - \eps^3)\lVert [U \ A_{:,i}]y \rVert_0$$ for all $y \in \R^{k+1}$, and $\it{bad}$ otherwise. Let $Q_i =  q_{1 - \eps / 2}(\mathcal{A}^S(UV^* - A, \eps \tau)_{:,i}))$. We say a bad column index is $\it{large}$ if $$ \eps \lVert (UV - A)_{:,i} \rVert_0 \geq \frac{2}{1 - \eps}Q_i + \lVert (UV^* - A)_{:,i} \rVert_0$$ and $\it{small}$ otherwise.

	By (ii), we know that $\ex{\sum_{\text{bad } i} \lVert (UV^* - A)_{:,i} \rVert_0} \leq \eps^3 \lVert UV^* - A \rVert_0$. By Markov's inequality, we know that with probability $1 - O(1)$, \begin{equation} \label{L0markovbadbound} \sum_{\text{bad } i} \lVert (UV^* - A)_{:,i} \rVert_0 \leq O(\eps^3) \lVert UV^* - A \rVert_0 \end{equation}

		By (iii) \begin{align*}(1 + \eps^3) \lVert UV^* - A \rVert_0
			&\geq \med (\mathcal{A}(UV^* - A)) \\
			&\geq (1 - \eps^3) \sum_{\text{good } i} \lVert (UV^* - A)_{:,i} \rVert_0  + \sum_{\text{bad } i} \med (\mathcal{A}(UV^* - A)_{:,i}) \\
			&\geq (1 - \eps^3) (1 - \Theta(\eps^3)) \lVert UV^* - A \rVert_0  + \sum_{\text{bad } i} \med (\mathcal{A}(UV^* - A)_{:,i}), \\
		\end{align*}

	where the second inequality comes from the definition of good, and the third inequality comes from (\ref{L0markovbadbound}).

	Thus \begin{equation} \label{L0sketchbadbound} \sum_{\text{bad } i} \med(\mathcal{A}(UV^* - A)_{:,i}) \leq O(\eps^3) \lVert UV^* - A \rVert_0 \end{equation}

We can also define notions of $T$-$\it{large}$ and $\it{tiny}$ analogous to those in the proof of  Theorem~\ref{constantsketch}. Using similar arguments to those in the proof, we can derive  \begin{align} \label{L0smallbound}
	\sum_{\text{small } i} \lVert (UV - A)_{:,i} \rVert_0 &\leq O(\eps)  \lVert UV^* - A \rVert_0.
\end{align}

	\begin{claim} \label{L0largesketch}
		$$\sum_{\text{large } i} \med (\mathcal{A}(UV - A)_{:,i}) \geq (1 - O(\eps)) \sum_{\text{large } i} \lVert (UV - A)_{:,i} \rVert_0$$
	\end{claim}

\begin{proof}
	Let column $i$ be large. We have $H(UV-A)_{:,i} = HU(V - V^*)_{:,i} + H(UV^* - A)_{:,i}$.

	By the triangle inequality, we have \begin{align*} &(1 - \Theta(\eps)) \lVert U(V - V^*)_{:,i} \rVert_0 \\
		\geq &(1 - \Theta(\eps)) (\lVert (UV - A)_{:,i} \rVert_0 - \lVert (UV^* - A)_{:,i} \rVert_0) \\
		\geq &(1 - \Theta(\eps)) ((1 - \eps) \lVert (UV - A)_{:,i} \rVert_0 + \frac{2}{1 - \eps} Q_i) \\
		\geq &(1 - \Theta(\eps)) \lVert (UV - A)_{:,i} \rVert_0 + Q_i \\
		\geq &(1 - \Theta(\eps)) \lVert (UV - A)_{:,i} \rVert_0
	\end{align*} where the second inequality follows from the definition of large.

	Since $ \lVert U(V - V^*)_{:,i} \rVert_0 \geq (1 - \Theta(\eps)) \lVert (UV - A)_{:,i} \rVert_0$ and $\eps  \lVert (UV - A)_{:,i} \rVert_0 \geq Q_i$, then $Q_i / \eps \leq  \lVert U(V - V^*)_{:,i} \rVert_0$.

	If we run $K$ independent instances of $H$, then by (i), we know that at least $\frac{1}{2} + \eps$ of those instances will have estimations $\mathcal{E}(U(V - V^*)_{:,i})$ that are well-behaved and satisfy $\mathcal{E}(U(V - V^*)_{:,i}) \geq (1 - \Theta(\eps)) \lVert U(V - V^*)_{:,i} \rVert_0.$

	At least $1 - \eps / 2$ of those instances will satisfy $Q_i > \est^S ((UV^* - A)_{:,i}, \eps \tau)$. In each of these instances, there is some index $t$ which is the maximum index where $\mathcal{C}^S ((UV^* - A)_{:,i}) > \eps \tau$. This index $t$ satisfies $\est ((UV^* - A)_{:,i}, \eps \tau) \geq 2^t \eps \tau$ which implies that $$t \leq \log_2 \left( \frac{\est^S ((UV^* - A)_{:,i}, \eps \tau)}{\eps \tau} \right) < \log_2 \left( \frac{Q_i}{\eps \tau} \right) \leq \log_2 \left( \frac{\lVert U(V - V^*)_{:,i} \rVert_0}{\tau} \right) \leq J_i$$ and by the nestedness property of $S$, for every index $l \geq J_i - 1$ we have $\mathcal{C}^S ((UV^* - A)_{:,i})_l < \eps \tau$. Furthermore, $\mathcal{C} ((UV^* - A)_{:,i})_l < \eps \tau$ because $\lVert H_0 y \rVert_0 \leq \lVert y \rVert_0$ for all $y$.

	Thus, for at least $\frac{1}{2} + \eps / 2$ instances of $H$, it is true that $\mathcal{E}(U(V - V^*)_{:,i})$ is well-behaved and for every index $l \geq J_i - 1$ we have $\mathcal{C} ((UV^* - A)_{:,i})_l < \eps \tau$. We first consider the case that $\lVert U(V - V^*)_{:,i} \rVert_0 > \tau$.

	We know that for $l = J_i -1, J_i,$ or $J_i + 1$, one of those values will be the maximum value such that the $l$th block of $HU(V - V^*)_{:,i}$ has at least $\tau$ non-zero entries, and all the later blocks will have at most $3 \tau / 4$ non-zero entries. Each block of $H(UV^* - A)_{:,i}$ after the $J_i - 1$th one will have fewer than $\eps \tau$ non-zero entries. By well-behavedness, it follows that $\mathcal{E} (UV - A)_{:,i} = \mathcal{E} (U(V- V^*)_{:,i} + (UV^* - A)_{:,i}) \geq (1 - \Theta(\eps)) \lVert U(V-V^*)_{:,i} \rVert_0$ because the salient blocks of $H(UV - A)_{:,i}$ will have a number of non-zero entries differing from those blocks of $HU(V - V^*)_{:,i}$ by an additive $\Theta(\eps)$ error.

	If $\lVert U(V - V^*)_{:,i} \rVert_0 \leq \tau$, then by well-behavedness we know that block $1$ of $HU(V- V^*)_{:,i}$ will have fewer than $3 \tau / 4$ non-zero entries. In this case all blocks of $H(UV^* - A)_{:,i}$ will have fewer than $\eps \tau$ non-zero entries so all blocks of $H(UV - A)_{:,i}$ besides the zeroth block will have fewer than $\tau$ non-zero entries. Thus, $\mathcal{E}(UV - A)_{:,i} \geq (1 - \Theta(\eps)) \geq \lVert U(V-V^*)_{:,i} \rVert_0$.

	Therefore in a majority of the instances of $H$, we have $\mathcal{E} (UV - A)_{:,i} \geq (1 - \Theta(\eps)) \lVert U(V - V^*)_{:,i} \rVert_0 \geq (1 - \Theta(\eps)) \lVert (UV - A)_{:,i} \rVert_0$ and the result follows.
\end{proof}

Finally, \begin{align*} \med (\mathcal{A}(UV - A)) &\geq \sum_{\text{good } i} \med (\mathcal{A}(UV - A)_{:,i}) + \sum_{\text{large } i} \med (\mathcal{A} (UV - A)_{:,i}) \\
	&\geq (1 - \eps^3) \sum_{\text{good } i} \lVert (UV - A)_{:,i} \rVert_0 + (1 - O(\eps)) \sum_{\text{large } i} \lVert (UV - A)_{:,i} \rVert_0 \\
	&\geq (1 - O(\eps)) \lVert UV - A \rVert_0 - (1 - O(\eps)) \sum_{\text{small } i} \lVert (UV - A)_{:,i} \rVert_0 \\
	&\geq (1 - O(\eps)) \lVert UV - A \rVert_0 - (1 - O(\eps)) O(\eps) \lVert UV^* - A \rVert_0 \\
	&\geq (1 - O(\eps)) \lVert UV - A \rVert_0
\end{align*} where the first inequality occurs because large $i$ are bad $i$, the second inequality comes from the definition of good and Claim \ref{L0largesketch}, the third inequality comes from the definition of small, the fourth inequality comes from (\ref{L0smallbound}), and the last inequality holds because $V^*$ is a minimizer.
\end{proof}

\begin{theorem}[$\mathbb{F}_q$ PTAS for $p = 0$]
	For $\eps \in (0,1)$ there is
        a $(1+\eps)$-approximation algorithm to \fieldprobbf running
        in $n \cdot d^{\poly(k/\eps)}$ time.
\end{theorem}

\begin{proof}
	Suppose $U^*$ and $V^*$ ($n \times k$ and $k \times d$ respectively) are minimizers for $\lVert UV - A \rVert_0$. By Theorem \ref{L0constantsketch}, $\med(\mathcal{A}(U^* V - A)) = (1 \pm \eps) \lVert U^* V - A \rVert_0$. Since $H$ is a $\frac{C'}{\eps^4} \log n \times n$ block matrix, then $HU^*$ has $\frac{C'}{\eps^4} k \cdot \log n$ entries and we need $K$ instances of $HU^*$ for a total of $(\log n) \cdot \text{poly}(k / \eps)$ entries each having $q$ possible values. Thus we can exhaustively guess all possible values of $HU^*$ in $n^{\poly(k / \eps)}$ time.

	For each guess of $HU^*$ and each column $i$, we can try all $q^k$ possible vectors $V_i$ and choose the minimizer. Once a $V$ has been identified, we can solve for its optimal $U$ and throughout this whole process keep the best $U$ and $V$ that minimize $\lVert UV - A \rVert_0$. Since there are $d$ rows, the algorithm will have a total runtime of $d \cdot n^{\poly(k / \eps)}$.

	As we stated in the opening exposition of this section, we could have sketched over the dimension $d$ row space instead. In this case we would be guessing values for $H(V^*)^T$, a $\frac{C'}{\eps^4} \log d \times k$ matrix, which would take $d^{\poly(k / \eps)}$ time. We would then minimize over each of the $n$ rows of $U$ for a total runtime of $n \cdot d^{\poly(k / \eps)}$.
\end{proof}

	\newpage
	\section{Generalized Binary Approximation} \label{sec:combinatorial}

	Given a matrix $A\in\{0,1\}^{n\times d}$ with $n\geq d$, an integer $k$,
	and an inner product function
	$\langle .,. \rangle \colon \{0,1\}^k \times \{0,1\}^k \to \mathbb{R}$,
	the \genprobk problem asks to
	find matrices $U\in\{0,1\}^{n \times k}$ and $V\in\{0,1\}^{k \times d}$
	minimizing $\|A - U \cdot V\|_0$,
	where the product $U \cdot V$ is the $n \times d$ matrix $B$
	with $B_{i,j} = \langle U_{i,:}, V_{:,j} \rangle$.
	An algorithm for the \genprobk problem is an $\alpha$-approximation,
	if it outputs matrices $U\in\{0,1\}^{n \times k}$ and $V\in\{0,1\}^{k \times d}$
	satisfying
	\[
	\|A-U\cdot V\|_0 \leq \alpha\cdot
	\min_{U' \in \{0,1\}^{n \times k}, V' \in \{0,1\}^{k \times d}} \|A-U'\cdot V'\|_0.
	\]
	Choosing an appropriate inner product function $\langle .,. \rangle$
	which also runs in time $O(k)$, we obtain the \probk problem over
	the reals, $\mathbb{F}_2$, and the Boolean semiring.
	We assume that the function $\langle .,. \rangle$ can be evaluated in time $2^{O(k)}$,
	in order to simplify our running time bounds.

	In this section, we prove Theorem~\ref{thm:main2}, restated here for convenience.
	\begingroup
	\addtocounter{thm}{-1}
	\def\thethm{\ref{thm:main2}}
	\begin{thm}[PTAS for $p = 0$]
		For any $\eps \in (0,\tfrac{1}{2})$, there is a $(1+\eps)$-approximation algorithm
		for the \genprobk problem running in time
		$(2/\eps)^{2^{O(k)}/\eps^2} \cdot nd^{1+o(1)}$
		and succeeds with constant probability~\footnote{
			The success probability can be further amplified to $1-\delta$
			for any $\delta>0$ by running $O(\log(1/\delta))$ independent
			trials of the preceding algorithm.},
		where $o(1)$ hides a factor $\left(\log\log d\right)^{1.1}/\log d$.
	\end{thm}
	\endgroup

	Our algorithm achieves a substantial generalization of the standard clustering approach
	and applies to the situation with constrained centers.
	This yields the first randomized almost-linear time approximation scheme (PTAS)
	for the \genprobk problem.
	The time complexity of the algorithm in Theorem~\ref{thm:main2} is close to optimal,
	in the sense that the running time of any PTAS for the \genprobk problem
	must depend exponentially on $1/\eps$ and doubly exponentially on $k$,
	assuming the Exponential Time Hypothesis.
	For reader's convenience, we restate our result.

	\begingroup
	\addtocounter{thm}{-1}
	\def\thethm{\ref{thm:hard0}}
	\begin{thm}[Hardness for \genprobk]
		Assuming the Exponential Time Hypothesis, \genprobk has no $(1+\eps)$-approximation algorithm
		in time $2^{1/\eps^{o(1)}} \cdot 2^{n^{o(1)}}$. Further, for any $\eps \geq 0$, \genprobk has no
		$(1+\eps)$-approximation algorithm in time $2^{2^{o(k)}} \cdot 2^{n^{o(1)}}$.
	\end{thm}
	\endgroup

	Regarding the dependence on $\eps$, assume w.l.o.g. that $n \ge d$,
	and thus the input size is $O(n)$. Even for $k=1$ the problem is known to be
	NP-hard~\cite{GV15,dajw15}.
	Under ETH, no NP-hard problem has a $2^{n^{o(1)}}$-time algorithm~\footnote{
		ETH postulates that 3-SAT is not in time $2^{o(n)}$.
		Here we only need the weaker hypothesis that 3-SAT is not in time $2^{n^{o(1)}}$.}.
	We can restrict to $\eps \ge 1/n^2$, since any better approximation already yields
	an optimal solution. It follows for $k=1$ that the \genprobk problem has no $(1+\eps)$-approximation
	algorithm in time $2^{1/\eps^{o(1)}} \cdot 2^{n^{o(1)}}$. In other words, in order to improve our exponential dependence on $1/\eps$ to subexponential, we would need to pay an exponential factor in $n$.

	Regarding the dependence on $k$, note that for any $\eps$ a $(1+\eps)$-approximation
	algorithm for our problem decides whether the answer is $0$ or larger. In particular,
	over the Boolean semiring it solves the problem whether a given bipartite graph can be
	covered with $k$ bicliques. For this problem, Chandran et al.~\cite{CIK17} proved that
	even for $k = O(\log n)$ there is no $2^{2^{o(k)}}$-time algorithm, unless ETH fails.
	It follows that for any $\eps \ge 0$, \genprobk has no $(1+\eps)$-approximation algorithm
	in time $2^{2^{o(k)}} \cdot 2^{o(n)}$.
	In other words, in order to improve our doubly exponential dependence on $k$,
	we would need to pay an exponential factor in $n$.

	\paragraph{Organization}
	In Subsection~\ref{sec:samplingtheorem}, we state our core sampling result.
	In Subsection~\ref{subsec:SimpleAlgo},
	we give a simple but inefficient deterministic PTAS for the \genprobk problem,
	which serves as a blueprint for our efficient randomized PTAS.
	We present first the deterministic PTAS as it is conceptually simple and exhibits
	the main algorithmic challenge, namely, to design an efficient sampling procedure.
	In Subsection~\ref{subsec:ProofOfTheoremMainTech}, we prove our core sampling result
	by extending the analysis of Alon and Sudakov~\cite{AlonS99}
	to clustering problems with constrained centers,
	and by further strengthening an additive $\pm\eps mn$ approximation guarantee
	to a multiplicative factor $(1+\eps)$-approximation.
	In Subsection~\ref{sec:sampling}, we design an efficient sampling procedure,
	and this yields our efficient randomized PTAS.
	Our approach uses ideas from clustering algorithms pioneered by
	Kumar et al.~\cite{kumar2004simple} and refined in~\cite{kumar2005linear,ACMN2010}.

	\subsection{Setup - A Sampling Theorem}\label{sec:samplingtheorem}

	We denote the optimal value of \genprobk by
	\[
	\opt = \optk \Def \min_{U \in \{0,1\}^{n \times k},\, V \in \{0,1\}^{k \times d}}
	\nnz{A - U \cdot V}.
	\]
	Further, for a fixed matrix $V\in\{0,1\}^{k \times d}$ we let
	\[
	\optV \Def \min_{U\in\{0,1\}^{n \times k}}\nnz{ A - U\cdot V },
	\]
	and we say that a matrix $U \in \{0,1\}^{n \times k}$ is a \emph{best response} to $V$,
	if $\nnz{ A - U\cdot V }=\optV$.
	\smallskip

	Given a matrix $A \in \{0,1\}^{n \times d}$, a positive integer $k$, and
	an inner product function $\langle .,. \rangle \colon \{0,1\}^k \times \{0,1\}^k \to \R$,
	let $V \in \{0,1\}^{k \times d}$ be arbitrary and $U \in \{0,1\}^{n \times k}$ be
	a \emph{best response} to $V$.
	Partition the columns of $V$ (equivalently the columns of $A$) into
	\[
	\CVy := \{j \mid V_{:,j} = y\},
	\]
	for $y \in \{0,1\}^k$.
	For any row $i$, vector $y \in \{0,1\}^k$, and $c \in \{0,1\}$ we consider
	\[ Z_{i,y,c}:=|\{ j\in \CVy\, |\, A_{ij}=c\}| \quad\text{and}\quad Z_{i,y,\neq c}:=|\{ j\in \CVy\, |\, A_{ij}\neq c\}|. \]
	We define the \emph{exact cost} of a row $i$ for any vector $x \in \{0,1\}^k$ as
	\begin{equation}\label{eq:exactCostEix}
		E_{i,x} := \nnz{A_{i,:} - x^T \cdot V} = \sum_{y\in\{0,1\}^k} Z_{i,y,\neq\langle x,y\rangle}.
	\end{equation}
	Observe that $U_{i,:} \in \{0,1\}^k$ is a vector $x$ minimizing $E_{i,x}$
	(this follows from $U$ being a best response to $V$), and let $E_i:=E_{i,U_{i,:}}$.

	We do not know the partitioning $\CVy$, however, as we will see later we can assume that (1) we can sample elements from each $\CVy$ and (2) we know approximations of the sizes~$|\CVy|$.

	For (1), to set up notation let $\tC = (\tCy)_{y \in \{0,1\}^k}$ be a family, where $\tCy$ is a random multiset with elements from $\CVy$. Specifically, we will work with the following \emph{distribution}~$\D{V,t}$ for some $t \in \N$:
	For any $y \in \{0,1\}^k$, if $|\CVy| < t$ let $\tCy = \CVy$,
	otherwise sample $t$ elements from $\CVy$ with replacement and let the resulting multiset
	be $\tCy$.

	For (2), we say that a sequence $\alpha=(\ah{y})_{y \in \{0,1\}^k}$ is a sequence of \emph{$\delta$-approximate cluster sizes} if for any $y \in \{0,1\}^k$ with $|\CVy| < t$ we have $\alpha_y = |\CVy|$, and for the remaining $y \in \{0,1\}^k$ we have
	\[
	|\CVy| \le \alpha_y \le (1+\delta) |\CVy|.
	\]

	Then corresponding to $Z_{i,y,c}$ and $Z_{i,y,\neq c}$ we have random variables
	\[ \tZ_{i,y,c} := |\{j \in \tCy \mid A_{i,j} = c\}| \quad\text{and}\quad \tZ_{i,y,\neq c} := |\{j \in \tCy \mid A_{i,j} \neq c\}|. \]
	Given $\tC$ and $\alpha$, we define the \emph{estimated cost} of row $i$ and vector $x \in \{0,1\}^k$ as
	\begin{equation}\label{eq:estCostEix}
		\tE_{i,x} := \sum_{y\in\{0,1\}^k} \frac{\ah{y} }{|\tCy|} \tZ_{i,y,\neq\langle x,y\rangle}.
	\end{equation}
	If $\CVy=\emptyset$ for some $y\in\{0,1\}^k$, then $\tZ_{i,y,\neq\langle x,y\rangle}=0$ and we define the corresponding summand in~\eqref{eq:estCostEix} to be 0.
	Observe that if the approximation $\alpha_y$ is exact, i.e., $\alpha_y = |\CVy|$, then $\tE_{i,x}$ is an unbiased estimator for the exact cost $E_{i,x}$.

	We now simplify the problem to optimizing the estimated cost instead of the exact cost. Specifically, we construct a matrix $\tU \in \{0,1\}^{n \times k}$ by picking for each row $i$ any
	$$ \tU_{i,:} \in \textup{argmin} \{ \tE_{i,x} \mid x \in \{0,1\}^k\}. $$
	Note that matrix $\tU$ depends on the input $(A,k,\langle .,. \rangle)$,
	on the sequence $\alpha$,
	and on the sampled multisets $\tC = (\tCy)_{y \in \{0,1\}^k}$.
	When it is clear from the context, we suppress the dependence on
	$A,k,\langle .,. \rangle$, and write $\tU = \tU(\tC,\alpha)$.
	We show that this matrix yields
	a good approximation to the optimal cost.

	\begin{thm} \label{thm:maintech}
		For any matrix $V\in\{0,1\}^{k \times d}$,
		let $\alpha$ be a sequence of $\tfrac \eps 6$-approximate cluster sizes and draw $\tC$ according to distribution $\D{V,t}$ for $t = t(k,\eps) := 2^{4k+14} / \eps^2$.
		Then we have
		\[ \Ex_{\tC}\big[ \nnzs{A - \tU(\tC,\alpha) \cdot V}\big] \le (1+\eps) \optV. \]
	\end{thm}

	We defer the proof of Theorem~\ref{thm:maintech} to Section~\ref{subsec:ProofOfTheoremMainTech},
	and first show how it yields a simple but inefficient deterministic PTAS for the \genprobk problem
	running in time $n\cdot d^{\poly(2^k/\eps)}$, see Section~\ref{subsec:SimpleAlgo}.
	Then, in Section~\ref{sec:sampling}, we design a sampling procedure that improves
	the running time to $(2/\eps)^{2^{O(k)}/\eps^2} \cdot mn^{1+o(1)}$,
	where $o(1)$ hides a factor $\left(\log\log d\right)^{1.1}/\log d$.

	\subsection{Simple PTAS}\label{subsec:SimpleAlgo}

	In this subsection, we show how Theorem~\ref{thm:maintech} leads to
	a simple but inefficient deterministic PTAS, see Algorithm~\ref{alg:ptasSimplySlow},
	for the \genprobk problem.

	A basic, but crucial property used in our analysis is that
	given a matrix $A\in\{0,1\}^{n\times d}$, an integer $k$ and a matrix~$V$,
	we can compute a \emph{best response} matrix $U$ minimizing $\nnz{A - U \cdot V}$
	in time $2^{O(k)} nd$.
	Indeed, we can split
	$\nnz{A - U \cdot V} = \sum_{i=1}^d \nnz{A_{i,:} - U_{i,:} \cdot V}$ and
	brute-force the optimal solution $U_{i,:} \in \{0,1\}^k$
	minimizing the $i$-th summand $\nnz{A_{i,:} - U_{i,:} \cdot V}$.
	Symmetrically, given $U$ we can compute a best response $V$ in time $2^{O(k)} nd$.
	In particular, if $(U,V)$ is an optimal solution then $U$ is a best response for $V$,
	and $V$ is a best response for $U$.

	We now present the pseudocode of Algorithm~\ref{alg:ptasSimplySlow}.

	\begin{algorithm}[H]
		\textbf{Input:} A matrix $A \in \{0,1\}^{n \times d}$, an integer $k$,
		an inner product $\langle .,. \rangle$, and $\eps \in (0,1)$.

		\textbf{Output:} Matrices $\tU \in \{0,1\}^{n \times k}$, $\tV \in \{0,1\}^{k \times d}$
		such that $\nnzs{A - \tU \cdot \tV} \le (1+\eps) \cdot \opt_k$.
		\medbreak

		1. \emph{(Guess column set sizes)} Let $U,V$ be an optimal solution. Exhaustively guess all sizes $|\CVy| =: \alpha_y$ for $y \in \{0,1\}^k$. There are $d^{2^k}$ possibilities.
		\smallskip\smallskip

		2. \emph{(Guess column multisets)}
		Theorem~\ref{thm:maintech} implies existence of a family $\tC = (\tCy)_{y \in \{0,1\}^k}$ such that $\nnzs{A - \tU(\tC,\alpha) \cdot V} \le (1+\eps) \optk$, where each $\tCy$ is a multiset consisting of at most $t$ indices in $\{1,\ldots,d\}$. Exhaustively guess such a family $\tC$.
		There are $d^{O(t \cdot 2^k)}$ possibilities.
		\smallskip\smallskip

		3. \emph{(Compute $\tU$)}
		Now we know $A,k,\langle .,. \rangle,|\CVy|$ for all $y \in \{0,1\}^k$, and $\tC$, thus we can compute the matrix $\tU = \tU(\tC,\alpha)$, where row $\tU_{i,:}$ is any vector $x$ minimizing the estimated cost $\tE_{i,x}$. Since each row $\tU_{i,:} \in \{0,1\}^k$ can be optimized independently, this takes time $2^{O(k)} nd$.
		If we guessed correctly, we have $\nnzs{A - \tU \cdot V} \le (1+\eps) \opt_k$.
		\smallskip\smallskip

		4. \emph{(Compute $\tV$)}
		Compute $\tV$ as a best response to $\tU$. This takes time $2^{O(k)} nd$.
		If we guessed correctly, by best-response and Step 3, we have
		$$\nnzs{A - \tU \cdot \tV} \le \nnzs{A - \tU \cdot V} \le (1+\eps) \opt_k.$$

		5. \textbf{Return} the pair $(\tU,\tV)$ minimizing $\nnzs{A - \tU \cdot \tV}$ over all exhaustive guesses.

		\caption{(PTAS for \genprobk)}
		\label{alg:ptasSimplySlow}
	\end{algorithm}

	The correctness of Algorithm~\ref{alg:ptasSimplySlow} immediately follows
	from Theorem~\ref{thm:maintech}.
	The running time is dominated by the exhaustive guessing in Step 2,
	so we obtain time $m\cdot n^{\poly(2^k/\eps)}$.

	\subsection{Proof of the Sampling Theorem~\ref{thm:maintech}}\label{subsec:ProofOfTheoremMainTech}

	We follow the notation in Section~\ref{sec:samplingtheorem}, in particular $V \in \{0,1\}^{k \times d}$ is an arbitrary matrix and $U \in \{0,1\}^{n \times k}$ is a best response to $V$.
	We define $D_{i,x}$ as the difference of the cost of row $i$ w.r.t.\ a vector $x$ and
	the cost of row $i$ w.r.t.\ the optimal vector $U_{i,:}$, i.e.,
	\begin{align}\label{eq:Dix}
		D_{i,x} := E_{i,x} - E_{i}
		&= \| A_{i,:} - x^T \cdot V \|_0 - \| A_{i,:} - U_{i,:} \cdot V \|_0 \\
		&= \sum_{y\in\{0,1\}^k} Z_{i,y,\neq\langle x,y\rangle}
		- Z_{i,y,\neq\langle U_{i,:},y\rangle}. \notag
	\end{align}
	Note that a vector $x$ is suboptimal for a row $i$ if and only if $D_{i,x} > 0$.
	By a straightforward splitting of the expectation, we obtain the following.

	\begin{claim} \label{cla:firstclaim}
		For every $V\in\{0,1\}^{k \times d}$, we have
		$$\Ex_{\tC}\big[ \nnzs{A - \tU \cdot V} \big] = \optV + \sum_{i=1}^n \sum_{\substack{x \in \{0,1\}^k \\ D_{i,x} > 0}} \Pr\big[ \tU_{i,:} = x \big] \cdot D_{i,x}.$$
	\end{claim}
	\begin{proof}
		We split $\nnzs{A - \tU \cdot V} = \sum_{i=1}^n \nnzs{A_{i,:} - \tU_{i,:} \cdot V}$. This yields
		\begin{align*}
			\Ex_{\tC}\big[ \nnzs{A - \tU \cdot V} \big] &= \sum_{i=1}^n \Ex_{\tC}\big[ \nnzs{A_{i,:} - \tU_{i,:} \cdot V} \big]
			= \sum_{i=1}^n \sum_{x \in \{0,1\}^k} \Pr\big[ \tU_{i,:} = x \big] \cdot \nnzs{A_{i,:} - x^T \cdot V}.
		\end{align*}
		By definition of $D_{i,x}$, we have
		\begin{align*}
			\Ex_{\tC}\big[ \nnzs{A - \tU \cdot V} \big] &= \sum_{i=1}^n \sum_{x \in \{0,1\}^k} \Pr\big[ \tU_{i,:} = x \big] \cdot (\nnzs{A_{i,:} - U_{i,:} \cdot V} + D_{i,x}) \\
			&= \sum_{i=1}^n \Big( \nnzs{A_{i,:} - U_{i,:} \cdot V} + \sum_{x \in \{0,1\}^k} \Pr\big[ \tU_{i,:} = x \big] \cdot D_{i,x} \Big)  \\
			\LV{&= \optV + \sum_{i=1}^n \sum_{x \in \{0,1\}^k} \Pr\big[ \tU_{i,:} = x \big] \cdot D_{i,x}  \\}
			&= \optV + \sum_{i=1}^n \sum_{\substack{x \in \{0,1\}^k\\ D_{i,x} > 0}} \Pr\big[ \tU_{i,:} = x \big] \cdot D_{i,x}. \qedhere
		\end{align*}
	\end{proof}

	Similarly to $D_{i,x}$, we define an estimator
	\begin{align}\label{eq:tildeDix}
		\tD_{i,x} := \tE_{i,x} - \tE_{i,U_{i,:}}
		&= \sum_{y\in\{0,1\}^k} \frac{\ah{y} }{|\tCy|} \cdot \Big(\tZ_{i,y,\neq\langle x,y\rangle}
		- \tZ_{i,y,\neq\langle U_{i,:},y\rangle}\Big).
	\end{align}
	Note that $\tU_{i,:}$ is chosen among the vectors $x \in \{0,1\}^k$ minimizing $\tD_{i,x}$. Hence, our goal is to show that significantly suboptimal vectors (with $D_{i,x} > \tfrac \eps 3 \cdot E_{i}$) satisfy $\tD_{i,x} > 0$ with good probability, and thus these vectors are not picked in $\tU$.

	To this end, we split the rows $i$ and suboptimal vectors $x$ into:
	\begin{align*}
		L_0 &:= \{(i,x) \mid 0 < D_{i,x} \le \tfrac \eps 3 \cdot E_{i} \},  \\
		L_1 &:= \{(i,x) \mid \tfrac \eps 3 \cdot E_{i} < D_{i,x} \le E_{i} \},  \\
		L_2 &:= \{(i,x) \mid D_{i,x} > E_{i}\}.
	\end{align*}
	Observe that $\sum_{(x,i)\in L_0}\Pr\big[ \tU_{i,:} = x \big] \cdot D_{i,x}\leq \tfrac \eps 3 \cdot\optV$.
	By Claim~\ref{cla:firstclaim}, we can ignore all tuples $(i,x)\in L_0$,
	since
	\begin{equation}\label{eq:sumZ1Z2}
		\Ex_{\tC}\big[ \nnzs{A - \tU \cdot V} \big]
		\leq (1+\tfrac \eps 3)\optV  +
		\sum_{(i,x)\in L_1 \cup L_2}\Pr\big[ \tU_{i,:} = x \big] \cdot D_{i,x}.
	\end{equation}
	Hence, our goal
	is to upper bound the summation
	$\sum_{(i,x)\in L_1\cup L_2}\Pr\big[ \tU_{i,:} = x \big] \cdot D_{i,x}$.

	We next establish a sufficient condition for $\tU_{i,:} \ne x$, for any suboptimal vector $x$.
	Note that by definition of $D_{i,x}$ we have
	\begin{equation}\label{eq:DixZ}
		D_{i,x} = \sum_{y\in \{0,1\}^k} Z_{i,y,\neq\langle x,y\rangle}
		-Z_{i,y,\neq\langle U_{i,:},y\rangle} = \sum_{y\in \hYix} Z_{i,y,\neq\langle x,y\rangle}
		-Z_{i,y,\neq\langle U_{i,:},y\rangle},
	\end{equation}
	where
	$\hYix := \{ y \in \{0,1\}^k \mid \langle x,y \rangle \ne \langle U_{i,:},y \rangle\}$.
	Similarly, for the estimator we have
	\begin{equation}\label{eq:tDixZ}
		\tD_{i,x} = \sum_{y\in\{0,1\}^k} \frac{\ah{y} }{|\tCy|} \cdot \Big(\tZ_{i,y,\neq\langle x,y\rangle}
		- \tZ_{i,y,\neq\langle U_{i,:},y\rangle}\Big)
		= \sum_{y\in \hYix} \frac{\ah{y} }{|\tCy|} \cdot \Big(\tZ_{i,y,\neq\langle x,y\rangle}
		- \tZ_{i,y,\neq\langle U_{i,:},y\rangle}\Big).
	\end{equation}

	Let $\Wix$ be the event
	that for every $y\in \Yix  := \{ y \in \hYix \mid |\tCy| = t\}$ and every $c\in\{0,1\}$, we have
	\[
	\left| \tZ_{i,y,c} - \frac{|\tCy|}{|\CVy|}\cdot Z_{i,y,c} \right| \leq \Delta_y,
	\quad \text{ where } \quad \Delta_y:= \frac{t \cdot D_{i,x}}{2^{k+2} \cdot \ah{y}}.
	\]
	We now show that conditioned on the event $\Wix$, we have $\tD_{i,x} > 0$ for any $(i,x) \in L_1 \cup L_2$, and thus $\tU_{i,:} \ne x$.

	\begin{lem}\label{lem:condition}
		For any vector $x\in\{0,1\}^k$ and row $i\in[m]$,
		if event $\Wix$ occurs then we have
		$\tD_{i,x}\geq \tfrac 12 \cdot D_{i,x}- \tfrac \eps 6\cdot E_{i}$.
		In particular, if additionally $D_{i,x}> \tfrac \eps 3\cdot E_{i}$ then
		$\tD_{i,x}>0$.
	\end{lem}
	\begin{proof}
		Observe that
		$\tZ_{i,y,\neq c}\in\{\tZ_{i,y,0},\ \tZ_{i,y,1},\:\tZ_{i,y,0}+\tZ_{i,y,1}\}$ for any $i,y,c$.
		Since $\Ex[\tZ_{i,y,0}+\tZ_{i,y,1}]=|\tCy|=\tZ_{i,y,0}+\tZ_{i,y,1}$,
		conditioned on the event $\Wix$ for any $y\in \Yix$ all three random variables $\tZ_{i,y,0},\ \tZ_{i,y,1},\:\tZ_{i,y,0}+\tZ_{i,y,1}$
		differ from their expectation by at most
		$\Delta_y$. Hence, we have
		\begin{equation*}
			\left|\frac{\ah{y}}{|\tCy|}\tZ_{i,y,\neq\langle x,y\rangle}-
			\frac{\ah{y}}{|\CVy|}Z_{i,y,\neq\langle x,y\rangle}\right|
			\leq\frac{D_{i,x}}{2^{k+2}}.
		\end{equation*}
		The same inequality also holds for $y \in \hYix \setminus \Yix$, since then $\tZ_{i,y,\neq\langle x,y\rangle} = Z_{i,y,\neq\langle x,y\rangle}$ and $|\tCy| = |\CVy|$ (by definition of the distribution $\D{V,t}$). In combination with (\ref{eq:tDixZ}) we obtain
		\begin{equation} \label{eq:SpotOn}
			\tD_{i,x} \ge - \frac{D_{i,x}}{2} +\sum_{y\in \hYix} \frac{\ah{y} }{|\CVy|} \cdot \Big(Z_{i,y,\neq\langle x,y\rangle}
			- Z_{i,y,\neq\langle U_{i,:},y\rangle}\Big).
		\end{equation}

		Let $\ah{y}=(1 + \gamma_{y})|\CVy|$ with $0\leq\gamma_{y}\leq \tfrac \eps 6$
		for any $y\in\{ 0,1\} ^{k}$.
		By~\eqref{eq:DixZ}, and since $\ah{y}=|\CVy|=|\tCy|$ for every $y \in \hYix \setminus \Yix$ (by definition of distribution $\D{V,t}$), we have
		\begin{eqnarray*}
			\sum_{y\in\hYix}\frac{\ah{y}}{|\CVy|}\left(Z_{i,y,\neq\langle x,y\rangle}-Z_{i,y,\neq\langle U_{i,:},y\rangle}\right)\nonumber
			&=&D_{i,x}+\sum_{y\in\Yix}\gamma_{y}\left(Z_{i,y,\neq\langle x,y\rangle}-Z_{i,y,\neq\langle U_{i,:},y\rangle}\right)  \\
			&\geq& D_{i,x} - \sum_{y\in\Yix}\gamma_{y}Z_{i,y,\neq\langle U_{i,:},y\rangle}\nonumber \\
			&\geq&D_{i,x} - \frac{\eps}{6}\sum_{y\in\{0,1\}^{k}}Z_{i,y,\neq\langle U_{i,:},y\rangle}=D_{i,x} - \frac{\eps}{6}E_{i}.
		\end{eqnarray*}
		Together with (\ref{eq:SpotOn}), we have $\tD_{i,x} \ge \tfrac 12 D_{i,x} - \tfrac \eps 6 E_i$.
	\end{proof}

	We next upper bound the probability of picking a suboptimal vector $x$.

	\begin{claim} \label{cla:usethis}
		For any $x \in \{0,1\}^k$ with $D_{i,x} > \tfrac \eps 3\cdot E_{i}$, we have
		\[
		\Pr[\tU_{i,:} = x] \le \sum_{y \in \Yix } \min_{c \in \{0,1\}}
		\Pr\Big[ |\tZ_{i,y,c} - \Ex[\tZ_{i,y,c}]| > \Delta_y \Big].
		\]
	\end{claim}
	\begin{proof}
		For any $y \in \{0,1\}^k$, we have
		$\tZ_{i,y,0} + \tZ_{i,y,1} = |\tCy| = \Ex[\tZ_{i,y,0}] + \Ex[\tZ_{i,y,1}]$.
		Further, it holds that $|\tZ_{i,y,0} - \Ex[\tZ_{i,y,0}]| = |\tZ_{i,y,1} - \Ex[\tZ_{i,y,1}]|$,
		and thus
		\LV{\begin{align*}
				&\Pr\Big[ |\tZ_{i,y,0} - \Ex[\tZ_{i,y,0}]| \le \Delta_y \Big]
				= \Pr\Big[ |\tZ_{i,y,1} - \Ex[\tZ_{i,y,1}]| \le \Delta_y \Big] \\
				&= \Pr\Big[ |\tZ_{i,y,0} - \Ex[\tZ_{i,y,0}]| \le \Delta_y \text{ and } |\tZ_{i,y,1} - \Ex[\tZ_{i,y,1}]| \le \Delta_y \Big].
			\end{align*}
		}
		\SV{$\Pr\big[ |\tZ_{i,y,0} - \Ex[\tZ_{i,y,0}]| \le \Delta_y \big]
			= \Pr\big[ |\tZ_{i,y,1} - \Ex[\tZ_{i,y,1}]| \le \Delta_y \big] = \Pr\big[ |\tZ_{i,y,0} - \Ex[\tZ_{i,y,0}]| \le \Delta_y \text{ and } |\tZ_{i,y,1} - \Ex[\tZ_{i,y,1}]| \le \Delta_y \big]$.
		}
		Since $\tU_{i,:} = x$ can only hold if $\tD_{i,x} \le 0$, the claim follows by Lemma~\ref{lem:condition} and a union bound over $y \in \Yix$.
	\end{proof}

	In the following subsections, we bound the summation in~\eqref{eq:sumZ1Z2}
	over the sets $L_1$ and $L_2$.

	\subsubsection{Case 1: Small Difference}

	We show first that $|L_1|$ is small (see Claim~\ref{cla:sizebound}).
	Then, we use a simple bound for $\Pr[\tU_{i,:} = x]$ which is based on Lemma~\ref{lem:cheb2}
	(see Claim~\ref{cla:firstcase}).

	\begin{claim} \label{cla:sizebound}
		We have $\sum_{(i,x) \in L_1} \sum_{y \in \Yix } |\CVy| \le 2^{k+2}\cdot \optV$.
	\end{claim}
	\begin{proof}
		Fix $(i,x) \in L_1$ and let $y \in \Yix $. Note that since $\langle x,y \rangle \ne \langle U_{i,:},y \rangle$ we have
		\[
		\{ j \in \CVy \mid A_{i,j} \ne \langle x,y \rangle \} \cup \{ j \in \CVy \mid A_{i,j} \ne \langle U_{i,:},y \rangle\} = \CVy.
		\]
		Note that this union is not necessarily disjoint, e.g., if $\langle x,y \rangle \not\in \{0,1\}$.
		Since $E_{i,x} = D_{i,x} + E_{i}$ (by \eqref{eq:Dix}) and $D_{i,x} \le E_{i}$ (by definition of $L_1$), we have
		\begin{equation}\label{eq:YixVyUB}
			\sum_{y \in \Yix } |\CVy| \leq
			\sum_{y \in \Yix } Z_{i,y,\neq\langle x,y \rangle} + Z_{i,y,\neq\langle U_{i,:},y \rangle}
			\leq E_{i,x} + E_{i} \leq 3 E_{i}.
		\end{equation}
		Fixing $x$ and summing over all $i$ with $(i,x) \in L_1$,
		the term $E_{i}$ sums to at most $\optV$.
		Also summing over all $x\in\{0,1\}^k$ yields another factor $2^k$.
		Therefore, the claim follows.
	\end{proof}

	\begin{claim} \label{cla:firstcase}
		We have $ \sum_{(i,x) \in L_1} \Pr[\tU_{i,:} = x] \cdot D_{i,x} \le \tfrac \eps 3 \cdot \optV$.
	\end{claim}
	\begin{proof}
		Note that for any row $i$, vector $y \in \Yix $, and $c \in \{0,1\}$, the random variable $\tZ_{i,y,c}$ is a sum of independent Bernoulli random variables, since the $t$ samples from $\CVy$ forming $\tCy$ are independent, and each sample contributes either 0 or 1 to $\tZ_{i,y,c}$. Hence, our instantiations of Chebyshev's inequality, Lemmas~\ref{lem:cheb1} and \ref{lem:cheb2}, are applicable.
		We use Lemma~\ref{lem:cheb2} to bound
		\LV{
			\begin{align*}
				\Pr\Big[ |\tZ_{i,y,c} - \Ex[\tZ_{i,y,c}]| > \Delta_y \Big] \le \frac{\sqrt{t}}{\Delta_y}.
			\end{align*}
		}
		\SV{
			$\Pr\big[ |\tZ_{i,y,c} - \Ex[\tZ_{i,y,c}]| > \Delta_y \big] \le \sqrt{t}/\Delta_y$.%
		}
		Since $\Delta_y=t \cdot D_{i,x} / (2^{k+2} \cdot \ah{y})$ and $\ah{y}\leq(1+\tfrac \eps 6)|\CVy|<2|\CVy|$, we have
		\LV{
			\begin{align*}
				\Pr\Big[ |\tZ_{i,y,c} - \Ex[\tZ_{i,y,c}]| > \Delta_y \Big] \le \frac{2^{k+3} |\CVy|}{\sqrt{t} \cdot D_{i,x}},
			\end{align*}
		}
		\SV{%
			$\Pr\big[ |\tZ_{i,y,c} - \Ex[\tZ_{i,y,c}]| > \Delta_y \big] \le 2^{k}2^3 |\CVy|/(\sqrt{t} \cdot D_{i,x})$,%
		}
		and thus by Claim~\ref{cla:usethis}, we obtain
		\LV{
			\begin{align*}
				\Pr[\tU_{i,:} = x] \le \frac{2^{k+3}}{\sqrt{t} \cdot D_{i,x}} \sum_{y \in \Yix } |\CVy|.
			\end{align*}
		}%
		\SV{%
			$\Pr[\tU_{i,:} = x] \le 2^{k+3}/(\sqrt{t}  D_{i,x}))  \cdot \sum_{y \in \Yix } |\CVy|$.%
		}
		Claim~\ref{cla:sizebound} now yields
		\begin{align*}
			\sum_{(i,x) \in L_1} \Pr[\tU_{i,:} = x] \cdot D_{i,x} \le \frac{2^{k+3}}{\sqrt{t}} \sum_{(i,x) \in L_1} \sum_{y \in \Yix } |\CVy| \le \frac{2^{2k+5}}{\sqrt{t}} \optV.
		\end{align*}
		Since we chose $t \geq 2^{4k+14} / \eps^2$, see Theorem~\ref{thm:maintech}, we obtain the upper bound $\tfrac \eps 3 \optV$.
	\end{proof}

	\subsubsection{Case 2: Large Difference}

	We use here the stronger instantiation of Chebyshev's inequality,
	Lemma~\ref{lem:cheb1}, and charge $\mu = \Ex[\tZ_{i,y,c}]$ against $\optV$.

	\begin{claim} \label{cla:secondcase}
		We have $\sum_{(i,x) \in L_2} \Pr[\tU_{i,:} = x] \cdot D_{i,x} \le \tfrac  \eps 3 \cdot \optV$.
	\end{claim}
	\begin{proof}
		Fix $(i,x) \in L_2$ and let $y \in \Yix $. As in the proof of Claim~\ref{cla:firstcase}, we see that our instantiation of Chebyshev's inequality, Lemma~\ref{lem:cheb1}, is applicable to $\tZ_{i,y,c}$ for any $c \in \{0,1\}$. We obtain%
		\LV{
			\begin{align*}
				\Pr\Big[ |\tZ_{i,y,c} - \Ex[\tZ_{i,y,c}]| > \Delta_y \Big] \le \Ex[\tZ_{i,y,c}] / \Delta_y^2.
			\end{align*}
		}
		\SV{%
			$\Pr\big[ |\tZ_{i,y,c} - \Ex[\tZ_{i,y,c}]| > \Delta_y \big] \le \Ex[\tZ_{i,y,c}] / \Delta_y^2$.
		}
		Note that $\Ex[\tZ_{i,y,c}] = Z_{i,y,c} \cdot t/|\CVy|$, since $|\tCy| = t$. Using
		$\min_{c\in\{0,1\}}Z_{i,y,c}\le Z_{i,y,\neq\langle U_{i,:},y\rangle}$,
		we have
		\begin{align*}
			\min_{c \in \{0,1\}} \Pr\Big[ |\tZ_{i,y,c} - \Ex[\tZ_{i,y,c}]| > \Delta_y \Big]
			\le
			\frac{t}{|\CVy| \Delta_y^2} \cdot Z_{i,y,\neq\langle U_{i,:},y\rangle}.
		\end{align*}
		Since $\Delta_y=t \cdot D_{i,x} / (2^{k+2} \cdot \ah{y})$ and $\ah{y}\leq(1+\eps/6)|\CVy|<2|\CVy|$, we have
		\begin{align*}
			\min_{c\in\{0,1\}}\Pr\Big[|\tZ_{i,y,c}-\Ex[\tZ_{i,y,c}]|>\Delta_{y}\Big]\le
			\frac{2^{2k+6}\cdot|\CVy|}{t\cdot(D_{i,x})^{2}}
			\cdot Z_{i,y,\neq\langle U_{i,:},y\rangle}.
		\end{align*}
		Summing over all $y \in \Yix $, Claim~\ref{cla:usethis} yields
		\begin{align} \label{eq:five}
			\Pr[\tU_{i,:}=x]\le
			\sum_{y\in\Yix}\frac{2^{2k+6}\cdot|\CVy|}{t\cdot (D_{i,x})^{2}}
			\cdot Z_{i,y,\neq\langle U_{i,:},y\rangle}.
		\end{align}
		We again use inequality (\ref{eq:YixVyUB}), i.e., $\sum_{y \in \Yix } |\CVy|\leq E_{i,x} + E_{i}$.
		Since $E_{i,x} = D_{i,x} + E_{i}$ (by \eqref{eq:Dix}) and $E_{i}<D_{i,x}$ (by definition of $L_2$), we have $|\CVy|\leq 3 D_{i,x}$ for any $y \in \Yix $.
		Together with \eqref{eq:five}, and then using the definition of $E_i$, we have
		\begin{align*}
			\Pr[\tU_{i,:}=x]\cdot D_{i,x} \le \frac{2^{2k+8}}{t}
			\sum_{y\in\Yix}Z_{i,y,\neq\langle U_{i,:},y\rangle}
			\leq \frac{2^{2k+8}}{t} E_{i},
		\end{align*}
		Fixing $x$ and summing over all $i$ with $(i,x) \in L_2$, the term $E_{i}$ sums to at most $\optV$. Also summing over all $x \in \{0,1\}^k$ yields another factor $2^k$. Thus, it follows that
		\begin{align*}
			\sum_{(i,x) \in L_2} \Pr[\tU_{i,:} = x] \cdot D_{i,x} \le  \frac{2^{3k+8}}{t} \optV.
		\end{align*}
		Since we chose $t \ge 2^{3k+10}/\eps$, see Theorem~\ref{thm:maintech}, we obtain the upper bound $\tfrac \eps 3 \cdot \optV$.
	\end{proof}

	\subsubsection{Finishing the Proof}

	Taken together, Claims~\ref{cla:firstclaim}, ~\ref{cla:firstcase}, and ~\ref{cla:secondcase} prove Theorem~\ref{thm:maintech}.

	\begin{proof}[Proof of Theorem~\ref{thm:maintech}]
		Using Claim~\ref{cla:firstclaim}, splitting into $L_0$, $L_1$ and $L_2$, and using Claims~\ref{cla:firstcase} and~\ref{cla:secondcase}, we obtain for any $\eps\in(0,1)$
		and $t = 2^{4k+12} / \eps^2$ that
		\begin{align*}
			\Ex_\tV\big[ \nnzs{A - \tU \cdot V} \big]
			&\leq (1+ \tfrac \eps 3)\optV + \sum_{(i,x) \in L_1} \Pr\big[ \tU_{i,:} = x \big] \cdot D_{i,x} + \sum_{(i,x) \in L_2} \Pr\big[ \tU_{i,:} = x \big] \cdot D_{i,x}  \\
			&\leq (1+\eps) \optV.
		\end{align*}
		This completes the proof.
	\end{proof}

	\subsection{Efficient Sampling Algorithm}\label{sec:sampling}

	The conceptually simple PTAS in Section~\ref{subsec:SimpleAlgo} has two
	running time bottlenecks, due to the exhaustive enumeration in Step 1 and Step 2.
	Namely, Step 1 guesses exactly the sizes $|\CVy|$ for each $y\in\{0,1\}^k$,
	and there are $d^{O(2^k)}$ possibilities; and
	Step 2 guesses among all columns of matrix $A$ the multiset family $\tC$,
	guaranteed to exist by Theorem~\ref{thm:maintech} and there are $d^{O(t\cdot 2^k)}$
	possibilities.

	Since Theorem~\ref{thm:maintech} needs only approximate cluster sizes,
	it suffices in Step 1 to guess numbers $\alpha_y$ with
	$|\CVy| \le \alpha_y \le (1 + \tfrac \eps 6) |\CVy|$ if $|\CVy| \ge t$,
	and $\alpha_y = |\CVy|$ otherwise, where $t = 2^{4k+12} / \eps^2$.
	Hence, the runtime overhead for Step 1 can be easily improved to
	$(t + \eps^{-1} \log d)^{2^k}$.

	To reduce the exhaustive enumeration in Step 2, we design an efficient sampling procedure,
	see Algorithm~\ref{alg:Sampling}, that uses ideas from clustering algorithms pioneered
	by Kumar et al.~\cite{kumar2004simple} and refined in~\cite{kumar2005linear,ACMN2010}.
	Our algorithm reduces the total exhaustive enumeration in Step 2 and
	the guessing overhead for the approximate cluster sizes in Step 1 to
	$(2^k/\eps)^{2^{O(k)}}\cdot (\log d)^{(\log\log d)^{0.1}}$ possibilities.

	This section is structured as follows. We first replace an optimal solution $(U,V)$
	by a ``well-clusterable'' solution $(U,W)$, which will help in our correctness proof.
	In Subsection~\ref{subsec:Pseudocode} we present pseudocode for our sampling algorithm.
	We then prove its correctness in Subsection~\ref{subsec:correcntesssamplelaog} and
	analyze its running time in Subsection~\ref{subsec:runnigntimesamplingalgo}.
	Finally, we show how to use the sampling algorithm designed in Subsection~\ref{subsec:Pseudocode}
	together with the ideas of the simple PTAS from Subsection~\ref{subsec:SimpleAlgo}
	to prove Theorem~\ref{thm:main2}, see Subsection~\ref{subsec:completePTASfinally}.

	\subsubsection{Existence of a \texorpdfstring{$(U,V,\eps)$}{(U, V, eps)}-Clusterable Solution}

	For a matrix $B \in \{0,1\}^{n \times d}$ we denote by $\colsupp(B)$ the set of unique columns of $B$.
	Note that if the columns of $U$ are linearly independent then $U \cdot \colsupp(V)$ denotes the set of distinct columns of $U \cdot V$.
	In the clustering formulation of the \genprobk problem as discussed in the introduction, the set $U \cdot \colsupp(V)$ corresponds to the set of cluster centers.

	Given matrices $U,V$, we will first replace $V$ by a related matrix $W$ in a way that makes all elements of $U \cdot \colsupp(W)$ sufficiently different without increasing the cost too much, as formalized in the following.

	\begin{lem}\label{lem:CR}
		For any $U\in\{0,1\}^{n \times k}$, $V\in\{0,1\}^{k \times d}$ and $\eps\in(0,1)$,
		there exists a matrix $W\in\{0,1\}^{k\times n}$ such that
		$\| A - U \cdot W \|_0 \le (1+\eps) \| A - U \cdot V \|_0$
		and for any distinct $y, z\in \colsupp(W)$ we have
		\begin{enumerate}
			\item[(i)] $\nnz{ Uy - Uz } > \eps \cdot 2^{-k} \cdot \| A - U \cdot V \|_0 / \min\{|\CWy|,|\CWz|\}$, and

			\item[(ii)] $\nnz{A_{:,j}-Uy}\leq\nnz{A_{:,j}-Uz}$ for every $j\in\CWy$.
		\end{enumerate}
		We say that such a matrix $W$ is \emph{$(U,V,\eps)$-clusterable}.
	\end{lem}
	\begin{proof}
		The proof is by construction of $W$. We initialize $W := V$ and then iteratively resolve violations of \emph{(i)} and \emph{(ii)}. In each step, resulting in a matrix $W'$, we ensure that $\colsupp(W') \subseteq \colsupp(W)$. We call this \emph{support-monotonicity}.

		We can resolve all violations of \emph{(ii)} at once by iterating over all columns $j \in [d]$ and replacing $W_{:,j}$ by the vector $z \in \colsupp(W)$ minimizing $\|A_{:,j} - U z \|_0$. This does not increase the cost $\| A - U \cdot W \|_0$ and results in a matrix $W'$ without any violations of \emph{(ii)}.

		So assume that there is a violation of \emph{(i)}. That is, for distinct
		$y, z\in \colsupp(W)$, where we can assume without loss of generality that
		$|\CWy| \le |\CWz|$, we have
		$$\nnz{ Uy - Uz } \le \eps \cdot 2^{-k} \cdot \| A - U \cdot V \|_0 / |\CWy|.$$
		We change the matrix $W$ by replacing for every $j \in \CWy$ the column $W_{:,j} = y$ by $z$.
		Call the resulting matrix $W'$. Note that the cost of any replaced column $j$ changes to
		\begin{align*}
			\|A_{:,j} - U \cdot W'_{:,j} \|_0
			= \|A_{:,j} - U z \|_0 &\le \|A_{:,j} - U y \|_0 + \|U y - U z\|_0  \\
			&\le \| A_{:,j} - U \cdot W_{:,j} \|_0 + \eps \cdot 2^{-k} \cdot \| A - U \cdot V \|_0 / |\CWy|.
		\end{align*}
		Since the number of replaced columns is $|\CWy|$, it follows that
		the overall cost increases by at most $\eps \cdot 2^{-k} \cdot \| A - U \cdot V \|_0$.
		Note that after this step the size of $\colsupp(W)$ is reduced by 1, since we removed any occurrence of column $y$. By support-monotonicity, the number of such steps is bounded by $2^k$. Since resolving violations of \emph{(ii)} does not increase the cost, the final cost is bounded by $(1+\eps) \| A - U \cdot V \|_0$.

		After at most $2^k$ times resolving a violation of \emph{(i)} and then all violations of \emph{(ii)}, we end up with a matrix $W$ without violations and the claimed cost bound.
	\end{proof}

	\subsubsection{The Algorithm Sample}\label{subsec:Pseudocode}

	Given $A \in \{0,1\}^{n \times d}$, $k \in \N$, $\eps \in (0,1)$, and $t \in \N$, fix any optimal solution $U,V$, that is $\nnz{ A- U \cdot V }=\optk$.
	Our proof will use the additional structure provided by well clusterable solutions. Therefore, fix any $(U,V,\eps)$-clusterable matrix $W$ as in Lemma~\ref{lem:CR}. Since $\nnz{ A- U \cdot W } \le (1+ \eps) \nnz{ A- U \cdot V }$, we can restrict to matrix $W$. Specifically, we fix the optimal partitioning $\CW$ of $[d]$ for the purpose of the analysis and for the guessing steps of the algorithm.
	Our goal is to sample from the distribution~$\D{W,t}$.

	Pseudocode of our sampling algorithm $\textbf{Sample}_{A,k,\eps,t}(M, \fN, \tR, \tC, \alpha)$ is given below.
	The arguments of this procedure are as follows.
	Matrix $M$ is the current submatrix of $A$ (initialized as the full matrix $A$).
	Set $\fN \subseteq \{0,1\}^k$ is the set of clusters that we did not yet sample from (initialized to $\{0,1\}^k$). The sequence $\tR$ stores ``representatives'' of the clusters that we already sampled from (initialized to undefined entries $(\perp,\ldots,\perp)$). The sequence $\tC$ contains our samples, so in the end we want $\tC$ to be drawn according to $\D{W,t}$ ($\tC$ is initialized such that $\tCy = \emptyset$ for all $y \in \{0,1\}^k$). Finally, $\alpha$ contains guesses for the sizes of the clusters that we already sampled from, so in the end we want it to be a sequence of $\tfrac \eps 6$-approximate cluster sizes ($\alpha$ is initialized such that $\alpha_y = 0$ for all $y \in \{0,1\}^k$).
	This algorithm is closely related to algorithm ``Irred-$k$-means'' by Kumar et al.~\cite{kumar2004simple}, see the introduction for a discussion.

	In this algorithm, at the base case we call $\text{\textbf{EstimateBestResponse}}_{A,k}(\tC, \alpha)$,
	which computes matrix $\tU = \tU(\tC,\alpha)$ and a best response $\tV$ to $\tU$.
	Apart from the base case, there are three phases of algorithm \textbf{Sample}.
	In the \emph{sampling} phase, we first guess some $y \in \fN$ and an approximation $\alpha_y$
	of $|\CWy|$. Then, from the current matrix $M$ of dimension $n\times \dM$,
	we sample $\min\{t,\alpha_y\}$ columns to form a multiset $\tCy$,
	and we sample one column from $M$ to form $\tR_y$. We make a recursive call with $y$ removed
	from $\fN$ and updated $\tR, \tC, \alpha$ by the values $\tR_y, \tCy, \alpha_y$.
	As an intermediate solution, we let $U^{(1)}, V^{(1)}$ be the best solution returned
	by the recursive calls over all exhaustive guesses.
	In the \emph{pruning} phase, we delete the $\dM/2$ columns of $M$ that are closest to $\tR$,
	and we make a recursive call with the resulting matrix $M'$, not changing the remaining arguments.
	Denote the returned solution by $U^{(2)}, V^{(2)}$.
	Finally, in the \emph{decision} phase we return the better solution between $U^{(1)}, V^{(1)}$
	and $U^{(2)}, V^{(2)}$.

	\begin{algorithm}[H]
		$\text{\textbf{Sample}}_{A,k,\eps,t}(M,\ \fN,\ \tR,\ \tC, \ \alpha)$

		$\quad$ let $\dM$ be the number of columns of $M$

		$\quad$ set $\nu:= (\eps /2^{k+4})^{2^k+2 - |\fN|}$

		\smallskip\smallskip

		$\,\,\,$1. \textbf{if} $\fN=\emptyset$ \textbf{or} $\dM=0$:
		\textbf{return} $(\tU,\tV)=\text{\textbf{EstimateBestResponse}}_{A,k}(\tC, \alpha)$

		\smallskip\smallskip

		$\quad$ \textbf{*} \textbf{Sampling phase} {*}

		$\,\,\,$2. \textbf{guess} $y\in\fN$

		$\,\,\,$3. \textbf{guess} whether $|\CWy| < t$:

		$\,\,\,$4. $\quad$ \textbf{if} $|\CWy| < t$: $\;$\textbf{guess} $\ah{y}:=|\CWy|$ exactly, i.e.
		$\ah{y}\in\{0,1,\dots,t-1\}$

		$\,\,\,$5. $\quad$ \textbf{otherwise}:$\;$ \textbf{guess}
		$\nu \cdot \dM\leq\ah{y}\leq\dM$ such that
		$|\CWy|\leq\ah{y}\leq(1+ \frac \eps 6)|\CWy|$

		$\,\,\,$6. \textbf{if} $\alpha_y=0$:
		$(U^{(y,\alpha_y)},V^{(y,\alpha_y)})=\text{\textbf{EstimateBestResponse}}_{A,k}(\tC, \alpha)$

		$\,\,\,$7. \textbf{else}

		$\,\,\,$8. $\quad$ sample u.a.r.\ $\min\{t,\alpha_y\}$ columns from $M$; let $\tCy$ be the resulting multiset\footnotemark

		$\,\,\,$9. $\quad$ sample u.a.r.\ one column from $M$; call it $\tR_y$

		10. $\quad$ $(U^{(y,\alpha_y)},V^{(y,\alpha_y)})=\text{\textbf{Sample}}_{A,k,\eps,t}
		(M,\ \fN\backslash\{y\},\ \tR \cup \{\tR_y\},\ \tC \cup \{\tC_y\}, \ \alpha \cup \{\alpha_y\})$

		11. let $(U^{(1)},V^{(1)})$ be the pair minimizing $\nnz{A-U^{(y,\alpha_y)}V^{(y,\alpha_y)}}$
		over all guesses $y$ and $\alpha_y$

		\smallskip\smallskip

		$\quad$ \textbf{*} \textbf{Pruning phase }{*}

		12. let $M'$ be matrix $M$ after the deleting $\dM/2$ closest columns to $\tR$,

		$\quad\;\,$ i.e., the $\dM/2$ columns $M_{:,j}$ with smallest values $\min_{y \in \{0,1\}^k \setminus \fN}\nnz{ M_{:,j} - \tR_y}$

		13. $(U^{(2)},\ V^{(2)})=\text{\textbf{Sample}}_{A,k,\eps,t}(M',\ \fN,\ \tR,\ \tC, \ \alpha)$

		\smallskip\smallskip

		$\quad$ \textbf{*} \textbf{Decision} {*}

		14. \textbf{return} $(U^{(\ell)},V^{(\ell)})$ with the minimal value
		$\nnz{ A - U^{(\ell)} V^{(\ell)} }$ over $\ell\in\{1,2\}$.

		\caption{Estimating Best Response}
		\label{alg:Sampling}
	\end{algorithm}
	\footnotetext{\label{fn:multiset}
		Given a submatrix $M$ of $A$, and $t$ columns sampled u.a.r.\ from $M$,
		we denote by $\tCy$ the resulting multiset of column indices with respect to the
		original matrix $A$.}

	\begin{algorithm}[H]

		$\text{\textbf{EstimateBestResponse}}_{A,k}(\tC, \alpha)$

		1. \emph{(Compute $\tU$)}
		Compute a matrix $\tU = \tU(\tC,\alpha)$, where row $\tU_{i,:}$ is any vector $x$ minimizing the estimated cost $\tE_{i,x}$. Note that each row $\tU_{i,:} \in \{0,1\}^k$ can be optimized independently.
		\smallskip

		2. \emph{(Compute $\tV$)}
		Compute $\tV$ as a best response to $\tU$.
		\smallskip

		3. \textbf{Return} $(\tU,\tV)$

		\caption{Estimating Best Response}
		\label{alg:EstBestResp}
	\end{algorithm}

	\subsubsection{Correctness of Algorithm Sample} \label{subsec:correcntesssamplelaog}

	With notation as above, we now prove correctness of algorithm \textbf{Sample}.

	\begin{thm}\label{thm:RecInv}
		Algorithm \emph{\textbf{Sample}}$_{A,k,\eps,t}$ generates a recursion tree
		which with probability at least $(\tfrac{\eps}{2t})^{2^{O(k)} \cdot t}$
		has a leaf calling $\text{\emph{\textbf{EstimateBestResponse}}}_{A,k}(\tC, \alpha)$
		such that
		\begin{compactenum}
			\item[(i)] $\alpha$ is a sequence of $\tfrac \eps 6$-approximate cluster sizes (w.r.t.\ the fixed matrix $W$), and
			\item[(ii)] $\tC$ is drawn according to distribution $\D{W,t}$.
		\end{compactenum}
	\end{thm}

	The rest of this section is devoted to proving Theorem~\ref{thm:RecInv}.
	Similarly as in the algorithm, we define parameters
	\[
	\gamma:=\eps / 2^{k+4} \quad\quad\text{and}\quad\quad \nu_i:=\gamma^{2^k + 2 - i}.
	\]
	Sort $\{0,1\}^k = \{y_1,\ldots,y_{2^k}\}$ such that $|\CW_{y_1}| \le \ldots \le |\CW_{y_{2^k}}|$.
	We construct the leaf guaranteed by the theorem inductively. In each depth $\tau=0,1,\ldots$ we consider one recursive call
	$$\text{\textbf{Sample}}_{A,k,\eps,t}(M^{(\tau)},\ \fN^{(\tau)},\ \tR^{(\tau)},\ \tC^{(\tau)}, \ \alpha^{(\tau)}).$$
	We consider the partitioning $P^{(\tau)}:=\{P^{(\tau)}_y\}_{y\in\{0,1\}^k}$ induced by the partitioning~$\CW$ on~$M^{(\tau)}$, i.e., $P^{(\tau)}_y$ is the set $\CWy$ restricted to the columns of $A$ that appear in the submatrix~$M^{(\tau)}$.
	We claim that we can find a root-to-leaf path such that the following inductive invariants hold with probability at least $(\nu_0 / t)^{(2^k-|\fN^{(\tau)}|)(t+1)}$:
	\begin{enumerate}
		\item[I1.] $P^{(\tau)}_y = \CWy$ for all $y \in \fN^{(\tau)}$, i.e., no column of an unsampled cluster has been removed,
		\item[I2.] $\fN^{(\tau)} = \{y_1,\ldots,y_{|\fN^{(\tau)}|}\}$, i.e., the remaining clusters are the $|\fN^{(\tau)}|$ smallest clusters,
		\item[I3.] For any $y \in \{0,1\}^k \setminus \fN^{(\tau)}$ the value $\alpha^{(\tau)}_y$ is an $\tfrac \eps 6$-approximate cluster size, i.e., if $|\CWy| < t$ we have $\alpha^{(\tau)}_y = |\CWy|$, and otherwise $|\CWy| \le \alpha^{(\tau)}_y \le (1+\tfrac \eps 6) |\CWy|$,
		\item[I4.] For any $y \in \{0,1\}^k \setminus \fN^{(\tau)}$ the multiset $\tCy^{(\tau)}$ is sampled according to distribution $\D{W,t}$, i.e., if $|\CWy| < t$ then $\tCy^{(\tau)} = \CWy$ and otherwise $\tCy^{(\tau)}$ consists of $t$ uniformly random samples from $\CWy$ with replacement, and
		\item[I5.] For any $y \in \{0,1\}^k \setminus \fN^{(\tau)}$ the vector $\tR^{(\tau)}_y$ satisfies $\nnz{ \tR^{(\tau)}_y - U y } \le 2 \nnz{A - U W} / |\CWy|$.
	\end{enumerate}

	For shorthand, we set $d^{(\tau)} := d_{M^{(\tau)}}$ and $\nu^{(\tau)} := \nu_{|\fN^{(\tau)}|}$.

	\paragraph{Base Case:}
	Note that the recursion may stop in Step 1 with $\fN^{(\tau)} = \emptyset$ or $d^{(\tau)} = 0$, or in Step 6 with $\alpha^{(\tau)}_y = 0$ for some guessed $y \in \fN$.
	Since we only want to show existence of a leaf of the recursion tree, in the latter case we can assume that we guessed $y = y_{|\fN^{(\tau)}|}$ and $\alpha^{(\tau)}_y = |\CWy|$, and thus we have $|\CWy| = 0$. Hence, in all three cases we have $|\CWy| = 0$ for all $y \in \fN^{(\tau)}$, by invariant I2 and sortedness of $y_1,\ldots,y_{2^k}$. Since we initialize $\tCy^{(0)} = \emptyset$ and $\alpha^{(0)}_y = 0$, we are done for all $y \in \fN^{(\tau)}$. By invariants I3 and I4, we are also done for all $y \in \{0,1\}^k \setminus \fN^{(\tau)}$.
	The total success probability is at least
	\[
	\left(\frac{\nu_0}{t}\right)^{2^k(t+1)}
	= \left(\frac{\eps}{2^{k+4} t}\right)^{2^k(2^k+2)(t+1)}
	= \left(\frac{\eps}{2t}\right)^{2^{O(k)} \cdot t}.
	\]

	The proof of the inductive step proceeds by case distinction.

	\paragraph*{Case 1 (Sampling):}
	Suppose $|P^{(\tau)}_y|\ge\nu^{(\tau)}\cdot d^{(\tau)}$ for some $y\in\fN^{(\tau)}$.
	Since $P^{(\tau)}_y = \CWy$ (by invariant I1) and sortedness, we have $|\CWy| \ge \nu^{(\tau)} d^{(\tau)}$ for $y := y_{|\fN^{(\tau)}|}$.
	We may assume that we guess $y = y_{|\fN^{(\tau)}|}$ in Step 2, since we only want to prove existence of a leaf of the recursion tree.
	Note that there is a number
	$$\nu^{(\tau)} d^{(\tau)} \le \alpha_y \le d^{(\tau)}\quad\text{with}\quad
	|\CWy| \le \alpha_y \le (1+\tfrac \eps 6)|\CWy|$$
	(in particular $\alpha_y = |\CWy|$ would work), so we can guess such a number in Step 5. Together with Steps 3 and 4, we can assume that $\alpha^{(\tau+1)}$ satisfies invariant I3.

	In Step 8 we sample a multiset $\tCy$ of $\min\{t,\alpha_y\}$ columns from $M$.
	If $|\CWy| \ge t$, we condition on the event that all these columns lie in $\CWy$. Then $\tCy$ forms a uniform sample from $\CWy$ of size $t$. Since $|\CWy| \ge \nu^{(\tau)} d^{(\tau)}$, this event has probability at least $(\nu^{(\tau)})^t$.
	Otherwise, if $|\CWy| = \alpha_y < t$, we condition on the event that all $\alpha_y$ samples lie in $\CWy$ and are distinct. Then $\tCy = \CWy$. The probability of this event is at least $$(1/d^{(\tau)})^{\alpha_y} \ge (\nu^{(\tau)} / \alpha_y)^{\alpha_y} \ge (\nu^{(\tau)} / t)^t.$$
	In total, $\tC^{(\tau+1)}$ satisfies invariant I4 with probability at least $(\nu^{(\tau)} / t)^t.$

	In Step 9 we sample one column $\tR_y$ uniformly at random from $M$. With probability at least $\nu^{(\tau)}$, $\tR_y$ belongs to $\CWy$, and conditioned on this event $\mathcal{E}_y$ we have
	\[
	\Ex_{\tR_y}\left[\nnz{ \tR_y - Uy }\ \Big|\ \mathcal{E}_y\right]
	=\frac{1}{|\CWy|}\sum_{j\in \CWy}\nnz{ A_{:,j}  - Uy }
	\leq\frac{\nnz{A - UW}}{|\CWy|}.
	\]
	By Markov's inequality, with probability at least $\nu^{(\tau)}/2$
	we have  $\nnz{ \tR_y - Uy }\leq2 \nnz{A - UW} / |\CWy|$, and thus invariant I5 holds for $\tR^{(\tau+1)}$.

	Finally, since we did not change $M^{(\tau)}$, invariant I1 is maintained.
	We conditioned on events that hold with combined probability at least
	$$(\nu^{(\tau)} / t)^t \cdot \nu^{(\tau)}/2 \ge (\nu_0 / t)^{t+1}.$$
	Since we decrement $|\fN^{(\tau)}|$ by removing $y = y_{|\fN^{(\tau)}|}$ from $\fN^{(\tau)}$,
	we maintain invariant I2, and we obtain total probability at least
	$$(\nu_0 / t)^{(2^k-|\fN^{(\tau+1)}|)(t+1)}.$$

	\paragraph*{Case 2 (Pruning):} Suppose $|P^{(\tau)}_y| < \nu^{(\tau)}\cdot d^{(\tau)}$ for every $y\in\fN^{(\tau)}$. (Note that cases 1 and~2 are complete.)
	In this case, we remove the $d^{(\tau)}/2$ columns of $M^{(\tau)}$ that are closest to~$\tR^{(\tau)}$, resulting in a matrix $M^{(\tau+1)}$, and then start a recursive call on $M^{(\tau+1)}$. Since we do not change $\fN^{(\tau)}, \tR^{(\tau)}, \tC^{(\tau)}$, and $\alpha^{(\tau)}$, invariants I2-I5 are maintained.

	Invariant I1 is much more difficult to verify, as we need to check that the $d^{(\tau)}/2$ deleted columns do not contain any column from an unsampled cluster.
	We first show that \emph{some} column of a cluster we already sampled from \emph{survives} to depth
	$\tau+1$ and has \emph{small} distance to $\tR^{(\tau)}$ (see Claim~\ref{clm:ExistYP}). Then we show that \emph{every} column of a cluster that we did not yet sample from has \emph{large} distance to $\tR^{(\tau)}$ (see Claim~\ref{clm:everyonehaslargedistance}). Since we delete the $d^{(\tau)}/2$ closest columns to~$\tR^{(\tau)}$, it follows that every column of a cluster that we did not yet sample from \emph{survives}.

	\begin{claim}\label{clm:ExistYP}
		There exists $x \in \{0,1\}^k \setminus \fN^{(\tau)}$ and column $j \in P_x^{(\tau+1)}$ with $$\nnzs{A_{:,j} - \tR^{(\tau)}_x} \le 2^{k+4} \nnzs{A - UW} / d^{(\tau)}.$$
	\end{claim}
	\begin{proof}
		By Case 2, we have $|P^{(\tau)}_y| < \nu^{(\tau)}\cdot d^{(\tau)}$ for every $y\in\fN^{(\tau)}$,
		and since $\nu^{(\tau)} \le \nu_{2^k} \le 2^{-k-2}$ it follows that
		\[
		\sum_{y \in \fN^{(\tau)}} |P^{(\tau)}_y| < 2^k \nu^{(\tau)} d^{(\tau)} \le d^{(\tau)} / 4.
		\]
		Combining $|P_y^{(\tau)}| \ge |P_y^{(\tau+1)}|$ and
		$\sum_{y \in \{0,1\}^k} |P_y^{(\tau+1)}| = d^{(\tau)}/2$, yields
		$$\sum_{y \in \{0,1\}^k \setminus \fN^{(\tau)}} |P_y^{(\tau+1)}| \ge d^{(\tau)} / 4.$$
		Hence, there is $x \in \{0,1\}^k \setminus \fN^{(\tau)}$ such that
		\begin{equation} \label{eq:myveryfirstequation}
			|P_x^{(\tau+1)}| \ge 2^{-k-2} d^{(\tau)}.
		\end{equation}

		By the minimum-arithmetic-mean inequality, some $j \in P_x^{(\tau+1)}$ satisfies
		\begin{align*}
			\nnzs{A_{:,j} - \tR^{(\tau)}_x}
			&\le \frac 1{|P_x^{(\tau+1)}|} \cdot \sum_{j' \in P_x^{(\tau+1)}} \nnzs{A_{:,j'} - \tR^{(\tau)}_x}  \\
			&\le \nnzs{\tR^{(\tau)}_x - Ux} + \frac 1{|P_x^{(\tau+1)}|} \cdot \sum_{j' \in P_x^{(\tau+1)}} \nnzs{A_{:,j'} - U x},
		\end{align*}
		where the last step uses the triangle inequality.
		For the first summand we use invariant I5, and for the second we use that $P_x^{(\tau+1)}$ is by definition part of an induced partitioning of $\CW$ on a smaller matrix, and thus the summation is bounded by $\nnzs{A - UW}$. This yields
		\[ \nnzs{A_{:,j} - \tR^{(\tau)}_x} \le \bigg( \frac{2}{|\CWx|} + \frac 1{|P_x^{(\tau+1)}|} \bigg)\cdot \nnzs{A - UW}. \]
		By $P_x^{(\tau+1)} \subseteq \CWx$ and by (\ref{eq:myveryfirstequation}), we obtain the claimed bound.
	\end{proof}

	\begin{claim}\label{clm:arbXcxUx}
		For any $y \in \{0,1\}^k \setminus \fN^{(\tau)}$ we have $|\CWy| \ge \nu^{(\tau)} d^{(\tau)} / \gamma$.
	\end{claim}
	\begin{proof}
		Since $y \not\in \fN$, we sampled from this cluster in some depth $\tau' < \tau$. In the call corresponding to $\tau'$, we had $\fN^{(\tau')} \supseteq \fN^{(\tau)} \cup \{y\}$ and thus $|\fN^{(\tau')}| \ge |\fN^{(\tau)}| +1$, and we had $d^{(\tau')} \ge d^{(\tau)}$. Since we sampled from $\CWy$ in depth $\tau'$, Case 1 was applicable, it follows that
		\begin{eqnarray*}
			\,\,\,\qquad\qquad\qquad|\CWy| & \stackrel{\text{(I1)}}{=} & |P_y^{(\tau')}| \ge \nu^{(\tau')} \cdot d^{(\tau')} = \nu_{|\fN^{(\tau')}|} \cdot d^{(\tau')} \\
			& \ge & \nu_{|\fN^{(\tau)}|+1} \cdot d^{(\tau)} = \nu_{|\fN^{(\tau)}|} \cdot d^{(\tau)} / \gamma = \nu^{(\tau)} \cdot d^{(\tau)} / \gamma.\qquad\qquad\qquad
		\end{eqnarray*}
	\end{proof}

	\begin{claim}\label{clm:everyonehaslargedistance}
		For any $y \in \{0,1\}^k \setminus \fN^{(\tau)}$, $z \in \fN^{(\tau)}$, and $j \in P^{(\tau)}_z$ we have $$\nnzs{A_{:,j} - \tR^{(\tau)}_y} > 2^{k+4} \nnzs{A - UW} / d^{(\tau)}.$$
	\end{claim}
	\begin{proof}
		By triangle inequality, we have
		$$\nnzs{Uy - Uz} \le \nnzs{A_{:,j} - Uy} + \nnzs{A_{:,j} - Uz}.$$
		Since $j \in P^{(\tau)}_z = \CWz$ and by property \emph{(ii)} of $(U,V,\eps)$-clustered (see Lemma~\ref{lem:CR}), the first summand is at least as large as the second, and we obtain
		\[
		\nnzs{Uy - Uz} \le 2 \nnzs{A_{:,j} - Uy}.
		\]
		We use this and the triangle inequality to obtain
		\begin{align*}
			\nnzs{A_{:,j} - \tR^{(\tau)}_y} & \ge \nnzs{A_{:,j} - Uy} - \nnzs{\tR^{(\tau)}_y - Uy} \\
			& \ge \frac 12\nnzs{Uy - Uz} - \nnzs{\tR^{(\tau)}_y - Uy}.
		\end{align*}
		For the first summand we use property \emph{(i)} of $(U,V,\eps)$-clustered (see Lemma~\ref{lem:CR}), for the second we use invariant I5. This yields
		\[ \nnzs{A_{:,j} - \tR^{(\tau)}_y} > \frac \eps {2^{k+1}} \cdot \frac{\| A - U V \|_0}{|\CWz|} - \frac{2 \nnz{A - U W}}{|\CWy|}. \]
		Since $y \in \{0,1\}^k \setminus \fN^{(\tau)}$, Claim~\ref{clm:arbXcxUx} yields $|\CWy| \ge \nu^{(\tau)} d^{(\tau)} / \gamma$. Since $z \in \fN^{(\tau)}$, by invariant I1,
		and since we are in Case~2, we have
		$$|\CWz| = |P^{(\tau)}_z| < \nu^{(\tau)}\cdot d^{(\tau)}.$$
		Moreover, by the properties of $(U,V,\eps)$-clustered (see Lemma~\ref{lem:CR}),
		it follows that $\nnzs{A - UW} \le (1+\eps) \nnzs{A - UV}$ and thus
		$\nnzs{A - UV} \ge \tfrac 12 \nnzs{A - UW}$. Together, this yields
		\begin{align*}
			\nnzs{A_{:,j} - \tR^{(\tau)}_y} & >
			\bigg(\frac \eps {2^{k+2}} - 2\gamma \bigg) \cdot \frac{\| A - U W \|_0}{\nu^{(\tau)} d^{(\tau)}} \\
			& = \frac \eps {2^{k+3} \nu^{(\tau)}}\cdot \frac{\| A - U W \|_0}{d^{(\tau)}}  \\
			& \ge 2^{k+4} \cdot \frac{\| A - U W \|_0}{d^{(\tau)}},
		\end{align*}
		since $\gamma = \eps /2^{k+4}$ and
		$$\nu^{(\tau)} \le \nu_{2^k} = \gamma^2 \le \eps /2^{2k+7}.$$
	\end{proof}

	Together, Claims~\ref{clm:ExistYP} and \ref{clm:everyonehaslargedistance} prove that no column $j \in P^{(\tau)}_y$ with $y \in \fN^{(\tau)}$ is removed. Indeed, we remove the columns with smallest distance to $\tR^{(\tau)}$, some of the columns in distance $2^{k+4} \nnzs{A - UW} / d^{(\tau)}$ survives, and any column $j \in P^{(\tau)}_y$ with $y \in \fN$ has larger distance to $\tR^{(\tau)}$. It follows that invariant I1 is maintained, completing our proof of correctness.

	\subsubsection{Running Time Analysis of Algorithm Sample} \label{subsec:runnigntimesamplingalgo}

	We now analyze the running time of Algorithm~\ref{alg:EstBestResp}.

	\begin{lem}\label{lem:BR}
		Algorithm \emph{\textbf{EstimateBestResponse}} runs in time $2^{O(k)} nd$.
	\end{lem}
	\begin{proof}
		Note that if $\tC$ is drawn according to distribution $\D{W,t}$, then its total size
		$\sum_{y \in \{0,1\}^k} |\tCy|$ is at most $n$. Hence, we can ignore all calls violating
		this inequality.
		We can thus evaluate the estimated cost $\tE_{i,x}$ in time $2^{O(k)} d$.
		Optimizing over all $x \in \{0,1\}^k$ costs another factor $2^k$,
		and iterating over all rows $i$ adds a factor $n$.
		Thus, Step 1 runs in time $2^{O(k)}nd$.
		Further, Step 2 finds a best response matrix, which can be computed in the same running time.
	\end{proof}

	We proceed by analyzing the time complexity of Algorithm~\ref{alg:Sampling}.

	\begin{lem}\label{lem:samplingRT}
		For any $t = \poly(2^k/\eps)$,
		Algorithm \emph{\textbf{Sample}}$_{A,k,\eps,t}$ runs in time
		$(2/\eps)^{2^{O(k)}} \cdot nd^{1+o(1)}$,
		where $o(1)$ hides a factor $\left(\log\log d\right)^{1.1}/\log d$.
	\end{lem}
	\begin{proof}
		Consider any recursive call $\text{\textbf{Sample}}_{A,k,\eps,t}(M,\ \fN,\ \tR,\ \tC, \ \alpha)$. We express its running time as $T(a,b)$ where $a := |\fN|$ and $b := \log(\dM)$. For notational convenience, we let $\log(0) =: -1$ and assume that $\dM$ is a power of 2.

		If we make a call to algorithm \textbf{EstimateBestResponse} then this takes time $2^{O(k)} nd$ by the preceding lemma. Note that here we indeed have the size $d$ of the original matrix and not the size $\dM$ of the current submatrix, since we need to determine the cost with respect to the original matrix.

		In the sampling phase, in Step 2 we guess $y$, with $|\fN| \le 2^k$ possibilities. Moreover, in Steps 3,4,5 we guess either $\alpha_y \in \{0,1,\ldots,t-1\}$ or $\nuN \cdot \dM\leq\ah{y}\leq\dM$ such that $|\CWy|\leq\ah{y}\leq(1+ \frac \eps 6)|\CWy|$.
		Note that there are $O(\log (1/\nuN) / \log(1+ \eps/6)) = \poly(2^k/\eps)$ possibilities for the latter, and thus $\poly(2^k/\eps)$ possibilities in total. For each such guess we make one recursive call with a decremented $a$ and we evaluate the cost of the returned solution in time $2^{O(k)} nd$.

		In the pruning phase, we delete the $\dM/2$ columns that are closest to $\tR$, which can be performed in time $2^{O(k)}n\dM$ (using median-finding in linear time). We then make one recursive call with a decremented $b$.

		Together, we obtain the recursion
		\[ T(a,b) \le \poly(2^k/\eps) nd + \poly(2^k/\eps) \cdot T(a-1,b) + T(a,b-1), \]
		with base cases $T(0,b) = T(a,-1) = 2^{O(k)} nd$. The goal is to upper bound $T(2^k, \log d)$.

		Let $Y = \poly(2^k/\eps)$ and $X = Y \cdot nd$ such that
		\[ T(a,b) \le X + Y \cdot T(a-1,b) + T(a,b-1), \]
		and $T(0,b), T(a,-1) \le X$.
		We prove by induction that $T(a,b) \le X \cdot (2Y(b+2))^a$.
		This works in the base cases where $a=0$ or $b = -1$.
		Inductively, for $a > 0$ and $b \ge 0$ we bound~\footnote{
			Using similar arguments, for any $\alpha\in[0,1]$ the recurrence
			$T(a,b)\leq\left(1+\alpha\right)^{b}\cdot X+Y\cdot T(a-1,b)+T(a,b-1)$
			is upper bounded by $X\cdot\left(2Y\right)^{a}\cdot\left(b^{1-\alpha}+2\right)^{a}\cdot\left(1+\alpha\right)^{b}$.
			In particular, we obtain the following upper bound
			$T(a,b)\leq X\cdot\left(2Y\right)^{a}\cdot\min\left\{ \left(b+2\right)^{a},2^{a+b}\right\}$ and thus
			$T(2^k,\log d)\leq (2/\eps)^{2^{O(k)}} \cdot nd \cdot \min\big\{ (\log d)^{2^k}, \,d \big\}$.
		}
		\begin{align}
			T(a,b) &\le X + Y \cdot X \cdot (2Y(b+2))^{a-1} + X \cdot (2Y(b+1))^a  \nonumber\\
			&= X \cdot (2Y(b+2))^a \cdot \bigg( \frac 1 {(2Y(b+2))^a} + \frac 1 {2(b+2)}
			+ \bigg( \frac{b+1}{b+2} \bigg)^a \bigg)  \nonumber\\
			&\le X \cdot (2Y(b+2))^a \cdot
			\bigg( \frac 1 {2(b+2)} + \frac 1 {2(b+2)} + \frac{b+1}{b+2} \bigg)  \nonumber\\
			&= X \cdot (2Y(b+2))^a.\label{eq:recBound}
		\end{align}
		Let $C$ be a constant to be determined soon. Using~\eqref{eq:recBound},
		the total running time is bounded by
		\[
		T(2^k, \log d)
		\leq X \cdot (2Y(\log(d)+2))^{2^k}
		\leq (2/\eps)^{2^{(C+1)\cdot k}} \cdot nd \cdot \log^{2^k} d
		\leq (2/\eps)^{2^{(C+1)\cdot k}} \cdot nd^{1+o(1)},
		\]
		where the last inequality follows by noting that
		$\log^{2^{k}}d>(2/\eps)^{2^{(C+1)\cdot k}}$ iff
		$k<\log\left(\frac{\log\log d}{\log(2/\eps)}\right)^{1/C}$ and
		in this case
		\[
		\left(\log d\right)^{2^{k}}\leq
		\left(\log d\right)^{\left(\frac{\log\log d}{\log(2/\eps)}\right)^{1/C}}
		\leq n^{o(1)},
		\]
		where $o(1)$ hides a factor $\left(\log\log d\right)^{1+1/C}/\log d$.
		The statement follows for any $C\geq10$.
	\end{proof}

	\subsubsection{The Complete PTAS} \label{subsec:completePTASfinally}

	Finally, we use Algorithm \textbf{Sample} to obtain an efficient PTAS
	for the \genprobk problem. Given $A,k,\eps$, we call \textbf{Sample}$_{A,k,\eps/4,t}$ with
	$$t = t(k,\tfrac \eps 4) := 2^{4k+16}/\eps^2.$$
	(This means that we replace all occurrences of $\eps$ by $\eps/4$,
	in particular we also assume that $W$ is $(U,V,\tfrac \eps 4)$-clusterable.)
	By Theorem~\ref{thm:RecInv}, with probability at least
	$$(\tfrac \eps{2t})^{2^{O(k)} \cdot t} = (\eps/2)^{2^{O(k)} / \eps^2}$$
	at least one leaf of the recursion tree calls
	$\text{\textbf{EstimateBestResponse}}_{A,k}(\tC, \alpha)$ with proper $\tC$
	and $\alpha$ such that the Sampling Theorem~\ref{thm:maintech} is applicable.
	By choice of $t = t(k,\tfrac \eps 4)$, this yields
	\[
	\Ex\big[ \nnzs{A - \tU(\tC,\alpha) \cdot W}\big]
	\le (1+\tfrac \eps 4) \optW\leq(1+ \tfrac \eps 4)^2\optk,
	\]
	where we used that $W$ is $(U,V,\tfrac \eps 4)$-clusterable in the second step (see Lemma~\ref{lem:CR}). The algorithm \textbf{EstimateBestResponse} computes the matrix $\tU = \tU(\tC,\alpha)$ and a best response $\tV$ to $\tU$. This yields
	\[ \Ex\big[ \nnzs{A - \tU \cdot \tV} \big] \le \Ex\big[ \nnzs{A - \tU \cdot W}\big] \le (1+ \tfrac \eps 4)^2\optk. \]
	By Markov's inequality, with probability at least $1 - 1/(1 + \eps/4) \ge \tfrac \eps 5$ we have
	\[
	\nnzs{A - \tU \cdot \tV} \le (1+ \tfrac \eps 4) \cdot
	\Ex\big[ \nnzs{A - \tU \cdot \tV} \big] \le (1+ \tfrac \eps 4)^3\optk \le (1+ \eps)\optk.
	\]
	Hence, with probability at least $p = (\eps/2)^{2^{O(k)} / \eps^2}$ at least one solution $\tU,\tV$ generated by our algorithm is a $(1+\eps)$-approximation. Since we return the best of the generated solutions, we obtain a PTAS, but its success probability $p$ is very low.

	The success probability can be boosted to a constant by running
	$O(1/p) = (2/\eps)^{2^{O(k)} / \eps^2}$ independent trials of Algorithm~\textbf{Sample}.
	By Lemma~\ref{lem:samplingRT}, each call runs in time
	$(2/\eps)^{2^{O(k)}} \cdot nd^{1+o(1)}$,
	where $o(1)$ hides a factor $\left(\log\log d\right)^{1.1}/\log d$,
	yielding a total running time of
	$(2/\eps)^{2^{O(k)} / \eps^2} \cdot nd^{1+o(1)}$.
	This finishes the proof of Theorem~\ref{thm:main2}.
	The success probability can be further amplified to $1-\delta$ for any $\delta>0$,
	by running $O(\log(1/\delta))$ independent trials of the preceding algorithm.

	\newpage
	\section{Hardness} \label{sec:hardness}

	In this section, we prove hardness of approximately computing the best rank $k$-approximation of a given $n \times d$ matrix $A$, where $n \geq d$.
	Indeed all hardness results in this section hold when $k = d - 1$, indicating that reducing the rank by $1$ is indeed hard to even approximate.
	This complements our efficient approximation schemes when $k = O(1)$.

	Our results for $p \in (1, 2)$ assume the Small Set Expansion Hypothesis.
	Originally conjectured by Raghavendra and Stuerer~\cite{RS10}, it is still the only assumption that implies strong hardness results for various graph problems such as
	Uniform Sparsest Cut~\cite{RST12} and Bipartite Clique~\cite{Manurangsi18}.
	Assuming this hypothesis, we prove even stronger results than above that rules out {\em any} constant factor approximation in $\poly(n, k)$.
	The following theorem immediately implies Theorem~\ref{thm:hardp} in the introduction.

	\begin{thm}
		Fix $p \in (1, 2)$ and $r > 1$.
		Assuming the Small Set Expansion Hypothesis, there is no $r$-approximation algorithm for   rank $k$ approximation of a matrix $A \in \R^{n \times d}$ with $n \geq d$ and $k = d - 1$
		in the entrywise $\ell_p$ norm that runs in time $poly(n)$.

		Consequently, additionally assuming the Exponential Time Hypothesis, there exists $\delta := \delta(p, r) > 0$ such that
		there is no $r$-approximation algorithm for rank $k$ approximation of a matrix $A \in \R^{n \times d}$ with $n \geq d$ and $k = d - 1$
		in the entrywise $\ell_p$ norm that runs in time $2^{n^{\delta}}$.
		\label{thm:hardness_small_p}
	\end{thm}

	For $p \in (2, \infty)$, we do not rely on the Small Set Expansion Hypothesis though the hardness factor is bounded by a constant. Recall that $\gamma_p := \E_{g} [|g|^p]^{1/p}$ where $g$ is a standard Gaussian,
	which is strictly greater than $1$ for $p > 2$.

	\begin{thm}
		Fix $p \in (2, \infty)$ and $\eps > 0$.
		Assuming $\classP \neq \NP$,
		there is no $(\gamma_p^p - \eps)$-approximation algorithm for rank $k$ approximation of a matrix $A \in \R^{n \times d}$ with $n \geq d$ and $k = d - 1$
		in the entrywise $\ell_p$ norm that runs in time $poly(n)$.

		Consequently, assuming the Exponential Time Hypothesis, there exists $\delta := \delta(p, \eps) > 0$ such that
		there is no $(\gamma_p^p - \eps)$-approximation algorithm for rank $k$ approximation of a matrix $A \in \R^{n \times d}$ with $n \geq d$ and $k = d - 1$
		in the entrywise $\ell_p$ norm that runs in time $2^{n^{\delta}}$.
		\label{thm:hardness_big_p}
	\end{thm}

	We also prove similar hardness results for $\ell_0$-low rank approximation in finite fields. The following theorem immediately implies Theorem~\ref{thm:hard0field}  in the introduction.

	\begin{thm}
		Fix a finite field $\F$ and $r > 1$.
		Assuming $\classP \neq \NP$,
		there is no $r$-approximation algorithm for rank $k$ approximation of a matrix $A \in \F^{n \times d}$ with $n \geq d$ and $k = d - 1$
		in the entrywise $\ell_0$ metric that runs in time $poly(n)$.

		Consequently, assuming the Exponential Time Hypothesis, there exists $\delta := \delta(r) > 0$ such that
		there is no $r$-approximation algorithm for rank $k$ approximation of a matrix $A \in \F^{n \times d}$ with $n \geq d$ and $k = d - 1$
		in the entrywise $\ell_0$ metric that runs in time $2^{n^{\delta}}$.
		\label{thm:hardness_fields}
	\end{thm}

	Section~\ref{subsec:lp_to_minpq} proves Lemma~\ref{lem:lp_to_minpq}, showing that
	computing $\minpq{p^*}{p}{A}$ is equivalent to finding the best rank $k$ approximation of $A \in \R^{n \times d}$ when $n \geq d$ and $k = d - 1$.
	Section~\ref{subsec:minpq_to_maxpq} proves Lemma~\ref{lem:minpq_to_maxpq}, reducing
	$\norm{2}{p^*}{\cdot}$ to $\minpq{p^*}{p}{\cdot}$.
	Section~\ref{subsec:bbhksz} presents the Barak et al.~\cite{BBHKSZ12}'s proof of hardness of $\norm{2}{p^*}{\cdot}$ with modifications for all $q > 2$, finishing the proof of Theorem~\ref{thm:hardness_small_p} for $p \in (1, 2)$.
	Section~\ref{subsec:grsw} proves the hardness of $\minpq{p^*}{p}{\cdot}$ for $p > 2$, using the result of~\cite{GRSW16}, and finishes the proof of Theorem~\ref{thm:hardness_big_p}.
	Finally, Theorem~\ref{thm:hardness_fields} is proved in Section~\ref{subsec:hardness_fields}.

	\paragraph{Numerical issues.}
	In the proofs of Theorem~\ref{thm:hardness_small_p} and Theorem~\ref{thm:hardness_big_p},
	we consider our matrices as having real entries for simplicity, but our results will hold even when all entries are rescaled to polynomially bounded integers.
	The instance in Theorem~\ref{thm:hardness_big_p} is explicitly constructed and it can be easily checked that all entries are polynomially bounded integers.
	For Theorem~\ref{thm:hardness_small_p}, our hard instance $B$ for $\norm{p}{p^*}{\cdot}$ is simply a projection matrix and
	the final instance $A$ is obtained by $(\eps I + B)^{-1}$, so by ensuring that $\eps \geq 1/\poly(n)$, we can ensure that eigenvalues of $A$ are within $[1, \poly(n)]$.

	\newpage
	\subsection{\texorpdfstring{$\ell_p$}{lp}-Low Rank Approximation and \texorpdfstring{$\minpq{p^*}{p}{A}$}{minpq}}
	\label{subsec:lp_to_minpq}
	In this subsection, we prove the following lemma showing that computing $\minpq{p^*}{p}{A}$ is equivalent to finding the best rank $k$ approximation of $A \in \R^{n \times d}$ when $n \geq d$ and $k = d - 1$.

	\begin{lem} [Restatement of Lemma~\ref{lem:lp_to_minpq}]
		Let $p \in (1, \infty)$.
		Let $A \in \R^{n \times d}$ with $n \geq d$ and $k = d - 1$.
		Then \[
		\min_{U \in \R^{n \times k}, V \in \R^{k \times d}} \norm{p}{U V - A} = \min_{x \in \R^d, \norm{p^*}{x} = 1} \norm{p}{Ax}.
		\]
	\end{lem}
	\begin{proof}
		Assume that the rank of $A$ is $d$; otherwise the lemma becomes trivial.
		We first prove $(\geq)$.
		Given $V^* \in \R^{k \times d}$ that achieves the best rank $k$ approximation, assume without loss of generality that the rank of $V^*$ is $k = d - 1$.
		Let $x \in \R^d$ be the unique vector (up to sign) that is orthogonal to the rowspace of $V^*$ and $\norm{p^*}{x} = 1$.
		Let $a_1, \dots, a_n$ be the rows of $A$.
		For fixed $V^*$, for $i \in [n]$, the $i$th row $u^*_i \in \R^{k}$ of $U^*$ must be obtained by computing
		\[
		\min_{u^*_i \in \R^{k}} \norm{p}{u^*_i V^* - a_i} =
		\min_{y \in \rowspace(V^*)} \norm{p}{y - a_i} =
		\min_{z \in \R^d : \langle x, z \rangle = -\langle x, a_i \rangle} \norm{p}{z}.
		\]
		Note that by \Holder 's inequality, the last quantity is at least $|\langle x, z \rangle| / \norm{p^*}{x} = |\langle x, a_i \rangle| / \norm{p^*}{x} = |\langle x, a_i \rangle|$.
		Indeed, taking $z \in \R^d$ with $z_j := (-\langle x, a_i \rangle) \cdot (\sgn(x_j)|x_j|^{p^* / p})$ for each $j \in [d]$ implies
		\[
		\langle x, z \rangle = -\langle x, a_i \rangle \cdot \sum_{j \in [d]} \sgn(x_j) |x_j|^{p^* / p} \cdot x_j =
		-\langle x, a_i \rangle \cdot \norm{p^*}{x}^{p^*} =  -\langle x, a_i \rangle,
		\]
		and
		\[
		\norm{p}{z} =|\langle x, a_i \rangle| \cdot (\sum_{j \in [d]} |x_j|^{p^*})^{1/p} = |\langle x, a_i \rangle| \cdot \norm{p^*}{x}^{p^*/p} = |\langle x, a_i \rangle|,
		\]
		so we can conclude $\norm{p}{u^*_i V^* - a_i} = \min_{z \in \R^d : \langle x, z \rangle = -\langle x, a_i \rangle} \norm{p}{z} = |\langle x, a_i \rangle|$.
		Summing over $i \in [n]$,
		\[
		\norm{p}{U^*V^* - A}
		= \big( \sum_{i \in [n]} \norm{p}{u^*_i V^* - a_i}^p \big) ^{1/p}
		= \big( \sum_{i \in [n]} |\langle x, a_i \rangle|^p \big) ^{1/p}
		= \norm{p}{Ax}.
		\]
		This proves that $\min_{U \in \R^{n \times k}, V \in \R^{k \times d}} \norm{p}{U V - A} \geq \min_{x \in \R^d, \norm{p^*}{x} = 1} \norm{p}{Ax}$.
		For the other direction, given $x \in \R^d$ with $\norm{p^*}{x} = 1$, let $V^* \in \R^{k \times d}$ be a matrix whose rowspace is a $k$-dimensional subspace orthogonal to $x$, and compute $U^*$ as above.
		The above analysis shows that $\norm{p}{U^*V^* - A}  = \norm{p}{Ax}$, which completes the proof.
	\end{proof}

	\subsection{Reducing \texorpdfstring{$\|\cdot \|_{2\rightarrow p^*}$}{|.|2->p*} \texorpdfstring{$\minpq{p^*}{p}{\cdot}$}{minpq}}
	\label{subsec:minpq_to_maxpq}
	In this subsection, we show that computing $\minpq{p^*}{p}{\cdot}$ is as hard as computing $\norm{2}{p^*}{\cdot}$, proving the following lemma.
	\begin{lem} [Restatement of Lemma~\ref{lem:minpq_to_maxpq}]
		For any $\eps >0, p\in (1,\infty)$, there is an algorithm that runs in $\poly(n, \log (1/\eps))$ and on a non-zero input matrix $A$,
		computes a matrix $B$ satisfying
		\[
		(1-\eps)\norm{2}{p^*}{A}^{-2} \leq \minpq{p^*}{p}{B} \leq (1+\eps)\norm{2}{p^*}{A}^{-2}.
		\]
	\end{lem}
	The lemma is proved in the following two steps.

	\paragraph{Reducing $\|\cdot\|_{2\rightarrow p^*}$ to $\|\cdot \|_{p\rightarrow p^*}$.}
	We first prove the following claim.
	This follows from standard tools from Banach space theory that {\em factor} an operator from $\ell_p$ to $\ell_p^*$ via $\ell_2$.

	\begin{claim}
		\label{psd:p:to:p*}
		$\norm{p}{p^*}{A A^T} = \norm{2}{p^*}{A}^2$.
	\end{claim}

	\begin{proof}
		By the definitions of $p \to q$ norms,
		\[
		\norm{p}{p^*}{A A^T}
		= \sup_x \frac{\norm{p^*}{AA^T x}}{\norm{p}{x}}
		\leq \sup_x \frac{\norm{2}{p^*}{A} \norm{2}{A^T x}}{\norm{p}{x}}
		\leq \norm{2}{p^*}{A} \norm{p}{2}{A^T}
		= \norm{2}{p^*}{A}^2,
		\]
		where the last line follows from the fact that
		\[
		\norm{2}{p^*}{A} =\sup_{\norm{p}{y}=1}\,\sup_{\norm{2}{x}=1}\mysmalldot{y}{Ax}
		= \sup_{\norm{2}{x}=1}\,\sup_{\norm{p}{y}=1}\mysmalldot{A^T y}{x}  = \norm{p}{2}{A^T}.
		\]
		For the other direction,
		\begin{align*}
			\norm{p}{p^*}{AA^T}
			~&=~
			\sup_{\norm{p}{x}=1}\,\sup_{\norm{p}{y}=1}\mysmalldot{y}{AA^Tx}
			~=~
			\sup_{\norm{p}{x}=1}\,\sup_{\norm{p}{y}=1}\mysmalldot{A^Ty}{A^Tx} \\
			~&\geq~
			\sup_{\norm{p}{x}=1}\norm{2}{A^Tx}^2
			~=~
			\norm{p}{2}{A^T}^2
			~=~
			\norm{2}{p^*}{A}^2,
		\end{align*}
		which completes the proof.
	\end{proof}

	\paragraph{Reducing $\|\cdot \|_{p\rightarrow p^*}$ to $\minpq{p^*}{p}{\cdot}$.}
	We now relate two quantities $\| A \|_{p\rightarrow p^*}$ and $\minpq{p^*}{p}{B}$ for two related matrices $A$ and $B$.
	If $A$ is invertible, this can be seen easily.
	\begin{fact}
		\label{min:and:max}
		If $A$ is an invertible matrix, then $\minpq{p}{q}{A^{-1}} = (\norm{q}{p}{A})^{-1}$
	\end{fact}

	\begin{proof}
		First observe that the condition $A^{-1}x\neq 0$ is equivalent to the condition $x\neq 0$ since $A$ is invertible.
		Then we have,
		\[
		\inf_{x\neq 0} \frac{\norm{q}{A^{-1}x}}{\norm{p}{x}}
		~=~
		\inf_{Ax\neq 0} \frac{\norm{q}{A^{-1}x}}{\norm{p}{x}}
		~=~
		\bigg( \sup_{A^{-1}x\neq 0} \frac{\norm{p}{x}}{\norm{q}{A^{-1}x}} \bigg)^{-1}
		~=~
		\bigg( \sup_{y\neq 0} \frac{\norm{p}{A^{-1}y}}{\norm{q}{y}} \bigg)^{-1}. \]
		The leftmost quantity is $\minpq{p}{q}{A^{-1}}$ and the rightmost quantity is $(\norm{q}{p}{A})^{-1}$.
	\end{proof}

	Even if $A$ is not invertible, there is an invertible matrix $B$ whose $p \to q$ norm is close to that of $A$ for any $p$ and $q$.

	\begin{claim}
		Let $A$ be a non-zero $n \times d$ matrix. For any $p, q \in (1, \infty)$ and any $\eps>0$,
		there is an invertible and polynomial time computable $\max(n,d) \times \max(n,d)$ matrix $B$ such that
		$(1-\eps)\norm{p}{q}{A}\leq \norm{p}{q}{B} \leq (1+\eps)\norm{p}{q}{A}$.

	\end{claim}

	\begin{proof}
		Let $\oplus$ denote vector concatenation.
		We start by exhibiting a square matrix with the same norm. If $d\geq n$, we pad $0$'s to the bottom
		of $A$ to obtain an $d\times d$ matrix $A'$. Now for any $x\in \R^d$,  $\norm{q}{A'x} =
		\norm{q}{Ax \oplus 0^{d-n}} = \norm{q}{Ax}$. So $\norm{p}{q}{A} = \norm{p}{q}{A'}$.

		If $d\leq n$, we pad $0$'s to the right of $A$ to obtain an $n\times n$ matrix $A'$.
		Consider any $y\in \R^n$ and let $x\in \R^d$, $z\in \R^{n-d}$ be such that $y = x\oplus z$.
		Then we have $\norm{q}{A'y} = \norm{q}{Ax}$. Now since $\norm{p}{y}\geq \norm{p}{x}$,
		we have $\norm{p}{q}{A} \geq \norm{p}{q}{A'}$. On the other hand, $\norm{p}{q}{A} \leq
		\norm{p}{q}{A'}$ since $\norm{q}{A'(x\oplus 0^{n-d})} = \norm{q}{Ax}$ and $\norm{p}{x\oplus 0^{n-d}}
		= \norm{p}{x}$.

		\smallskip

		Next to obtain an invertible matrix, we set $B := A' + \eps'\cdot \mathrm{I}$ where
		$\eps' := \eps \cdot M/ \norm{p}{q}{\mathrm{I}}$ and $M$ is the max magnitude of an entry
		of $A$ which must be non-zero since $A$ is non-zero. First we observe that $\norm{p}{q}{A}\geq M$
		since one can substitute $x=e_i$ where $i$ is the index of the column containing the max magnitude entry.
		Lastly, applying triangle inequality (since $\norm{p}{q}{\cdot}$ is a norm) implies the claim.
	\end{proof}

	\subsection{Hardness of \texorpdfstring{$2 \to q$}{2 -> q} norm for all \texorpdfstring{$q \in (2, \infty)$}{q in (2, infinity)}}
	\label{subsec:bbhksz}
	In this subsection, we prove Theorem~\ref{thm:hardness:SSE} for hardness of $\norm{2 \to q}{\cdot}$ for $q \in (2, \infty)$.
	Barak et al.~\cite{BBHKSZ12} proved that under the Small Set Expansion Hypothesis, for any $r > 1$ and an even integer $q \geq 4$, it is NP-hard to approximate the $2 \to q$ norm problem within a factor $r$. The same proof essentially works for all $q \in (2, \infty)$ with slight modifications.
	For completeness, we present their proof here, with additional remarks when we generalize an even integer $q \geq 4$ to all $q \in (2, \infty)$.

	\paragraph{Preliminaries for Small Set Expansion.}
	For a vector $x \in \R^d$, every $p$-norm in this subsection denotes the expectation norm
	defined as $\enorm{p}{x} := (\E_{i \in [d]}  [|x_i|^p])^{1/p}$.
	For a regular graph $G = (V, E)$ and a subset $S \subseteq V$, we define the measure of $S$ to be $\mu(S) = |S| / |V|$ and we define $G(S)$ to be the distribution obtained by picking a random $x \in S$ and then outputting a random neighbor $y$ of $x$. We define the expansion of $S$ to be $$\Phi_G(S) = \Pr_{y \in G(S)}[y \notin S].$$
	For $\delta \in (0, 1)$, we define $\Phi_G(\delta) = \min_{S \subseteq V : \mu(S) \leq \delta}\Phi_G(S)$.
	We identify $G$ with its normalized adjacency matrix. For every $\lambda \in [-1, 1]$, we denote by $V_{\geq \lambda}(G)$ the subspace spanned by the eigenvectors of $G$ with eigenvalue at least $\lambda$. The projector into this subspace is denoted $P_{\geq \lambda}(G)$. For a distribution $D$, we let $\cp(D)$ denote the collision probability of $D$ (the probability that two independent samples from $D$ are identical). The Small Set Expansion Hypothesis, posed by Raghavendra and Steurer~\cite{RS10} states the following.

	\begin{hypothesis}
		For any $\eps > 0$, there exists $\delta > 0$ such that it is NP-hard to decide whether $\Phi_G (\delta) \leq \eps$ or $\Phi_G (\delta) \geq 1-\eps$.
	\end{hypothesis}
	This implies strong hardness results for various graph problems such as Uniform Sparsest Cut~\cite{RST12} and Bipartite Clique~\cite{Manurangsi18}.
	The main theorem of this subsection is the following, which corresponds to Theorem 2.4 of~\cite{BBHKSZ12}.

	\begin{thm}
		For every regular graph $G, \lambda \in (0, 1),$ and $q \in (2, \infty)$,
		\begin{enumerate}
			\item For all $\delta > 0, \eps > 0$, $\enorm{2}{q}{P_{\geq \lambda}(G)} \leq \eps/\delta^{(q-2)/2q}$ implies that $\Phi_G(\delta) \geq 1 - \lambda - \eps^2$.
			\item There is a constant $a = a(q)$ such that for all $\delta > 0$, $\Phi_G(\delta) > 1 - a \lambda^{2q}$ implies $$\enorm{2}{q}{P_{\geq \lambda}(G)} \leq 2/\sqrt{\delta}.$$
		\end{enumerate}
		\label{thm:bbh:main}
	\end{thm}

	Given this theorem, the hardness of $2 \to q$ norm can be proved as follows. This corresponds to Corollary 8.1 of~\cite{BBHKSZ12}.
	\begin{proof}[Proof of Theorem~\ref{thm:hardness:SSE}]
		Using~\cite{RST10}, the Small Set Expansion Hypothesis implies that
		for any sufficiently small numbers $0 < \delta \leq \delta'$,
		there is no polynomial time algorithm that can distinguish between the following cases for a given graph $G$:
		\begin{itemize}
			\item Yes case: $\Phi_G(\delta) < 0.1$.
			\item No case: $\Phi_G(\delta') > 1 - 2^{-a' \log(1/\delta')}$. ($a'$ is a fixed universal constant.)
		\end{itemize}
		In particular, for all $\eta > 0$, if we let $\delta' = \delta^{(q-2)/8q}$ and make $\delta$ small enough, then in the No case $\Phi_G(\delta^{(q-2)/8q}) > 1 - \eta$.
		(Since $q > 2$, $\delta' \to 0$ as $\delta \to 0$.)

		Using Theorem~\ref{thm:bbh:main}, in the Yes case we know $\enorm{2}{q}{P_{\geq 1/2}} \geq 1 / (10\delta^{(q-2)/2q})$,
		while in the No case, if we choose $\delta$ sufficiently small so that $\eta$ is smaller than $a (1/2)^{2q}$, then we know that $\enorm{2}{q}{P_{\geq 1/2}} \leq 2/\sqrt{\delta'}= 2 / \delta^{(q-2)/4q}$.
		The gap between the Yes case and the No case is at least $\delta^{-(q-2) / 4q} / 20$, which goes to $\infty$ as $\delta$ decreases.
	\end{proof}

	We now prove Theorem~\ref{thm:bbh:main}. The first part that proves small set expansion of $G$ given a $2 \to q$ norm bound indeed follows from older work (e.g.,~\cite{KV05}).
	\begin{lem}[Lemma B.1 of~\cite{BBHKSZ12}]
		For all $\delta > 0, \eps > 0$, $\enorm{2}{q}{P_{\geq \lambda}(G)} \leq \eps/\delta^{(q-2)/2q}$ implies that $\Phi_G(\delta) \geq 1 - \lambda - \eps^2$
		\label{thm:bbh:first}
	\end{lem}
	\begin{proof}
		Let $q^* = q/(q-1)$ be the \Holder conjugate of $q$ such that $1/q + 1/q^* = 1$.
		Since $P_{\geq \lambda}$ is a projector,
		\[
		\enorm{q^*}{2}{P_{\geq \lambda}(G)} = \enorm{q^*}{2}{P_{\geq \lambda}(G)^T} = \enorm{2}{q}{P_{\geq \lambda}(G)}.
		\]

		Given $S \subseteq V$ with $\mu(S) = \mu \leq \delta$, let $f = 1_S / \sqrt{\mu}$ be the normalized indicator vector of $S$ so that $\enorm{2}{f} = 1$.
		Let $f = f' + f''$ where $f'$ is its projection to the eigenvalues at least $\lambda$ (i.e., $f' = P_{\geq \lambda}f$) and $f''$ is its projection to the eigenvalues strictly less than $\lambda$.
		Since $\enorm{q^*}{1_S} = \mu^{1/q^*} = \mu^{(q-1)/q}$,
		we have $\enorm{q^*}{f} = \mu^{((q-1)/q) - 1/2} \leq \delta^{((q-1)/q) - 1/2}$ (since $q > 2$ and $\delta \geq \mu$), and
		\[
		\enorm{2}{f'} \leq \enorm{q^*}{f} \cdot \enorm{q^* \to 2}{P_{\geq \lambda}(G)} \leq \delta^{((q-1)/q) - 1/2} \cdot (\eps / \delta^{(q-2)/2q}) = \eps.
		\]
		Then
		\[
		\langle f, Gf \rangle =
		\langle f', Gf' \rangle +
		\langle f'', Gf'' \rangle
		\leq \enorm{2}{f'}^2 + \lambda \enorm{2}{f''}^2
		\leq \eps^2  + \lambda.
		\]
		Since $\Phi_G(S) = 1 - \langle f, Gf \rangle$, the lemma follows.
	\end{proof}

	The second part of Theorem~\ref{thm:bbh:main} requires more technical proofs.
	\begin{lem}[Lemma 8.2 of~\cite{BBHKSZ12}]
		There is a constant $a = a(q)$ such that for all $\delta > 0$, $\Phi_G(\delta) > 1 - a \lambda^{2q}$ implies $\enorm{2}{q}{P_{\geq \lambda}(G)} \geq 2/\sqrt{\delta}$.
		\label{lem:bbh:second}
	\end{lem}
	\begin{proof}
		Let $f$ be a function in $V_{\geq \lambda}$ with $\enorm{2}{f} = 1$ that maximizes $\enorm{q}{f}$.
		We write $f = \sum_{i =1}^m \alpha_i \chi_i$ where $\chi_1, \dots, \chi_m$ denote the eigenfunctions of $G$ with values $\lambda_1, \dots, \lambda_m$ that are at least $\lambda$.
		Assume towards contradiction that $\enorm{q}{f} < 2 / \sqrt{\delta}$. We will prove that $g = \sum_{i =1}^m (\alpha_i / \lambda_i) \chi_i$ satisfies $\enorm{q}{g} \geq 5\enorm{q}{f} / \lambda$.
		Note that $g$ is defined such that $f = Gg$.
		This is a contradiction since (using $\lambda_i \in [\lambda, 1]$) $\enorm{2}{g} \leq \enorm{2}{f} / \lambda$, and we assumed $f$ is a function in $V_{\geq \lambda}$ with a maximal ratio $\enorm{q}{f} / \enorm{2}{f}$.

		Let $U \subseteq V$ be the set of vertices such that $|f(x)| \geq 1/\sqrt{\delta}$ for all $x \in U$. Using the Markov inequality and the fact that $\E_{x \in V}[f(x)^2] = 1$, we know that $\mu(U) = |U|/|V| \leq \delta$.
		On the other hand, because $\enorm{q}{f}^q \geq 2^q / \delta^{q/2}$, we know that $U$ contributes at least half of the term $\enorm{q}{f}^q = \E_{x \in V}[ |f(x)|^q]$. That is, if we define $\alpha$ to be $\mu(U) \E_{x \in U} [|f(x)|^q]$ then $\alpha \geq \enorm{q}{f}^q / 2$. We will prove the lemma by showing that $\enorm{q}{g}^q \geq (10\lambda^{-1})^q \alpha$.

		Let $c = c(q)$ and $d = d(c, q)$ be sufficiently large constants that will be determined later, and $e = d \cdot \lambda^{-q}$.
		By the variant local Cheeger bound obtained in Theorem 2.1 of~\cite{Steurer10}, there exists $a = a(d, q)$ such that
		$\Phi_G(\delta) > 1 - a \lambda^{2q}$ implies that $\cp(G(S)) \leq 1/(e|S|)$ for all $S$ with $\mu(S) \leq \delta$.

		We define $U_i$ to be the set $\{ x \in U : f(x) \in [c^i / \sqrt{\delta}, c^{i+1}/\sqrt{\delta}] \}$, and let $I$ be the maximal $i$ such that $U_i$ is non-empty. Thus, the sets $U_0, \dots, U_I$ form a partition of $U$ (where some of these sets may be empty). We let $\alpha_i$ be the contribution of $U_i$ to $\alpha$. That is, $\alpha_i = \mu_i \E_{x \in U_i} [|f(x)|^q]$, where $\mu = \mu(U_i)$. Note that $\alpha = \alpha_0 + \dots + \alpha_I$. We will show that there are some indices $i_1, \dots, i_J$ such that
		\begin{enumerate}
			\item $\alpha_{i_1} + \dots + \alpha_{i_J} \geq \alpha / (2c^{q})$.
			\item For all $j \in [J]$, there is a non-negative function $g_j : V \to \R$ such that $\E_{x \in V} [|g_j(x)|^q] \geq e \alpha_{i_j} / (10c^2)^{q/2}$.
			\item For every $x \in V$, $g_1(x) + \dots + g_J(x) \leq |g(x)|$.
		\end{enumerate}
		Showing these will complete the proof, since it is easy to see that for non-negative functions $g', g''$ and $q \in [1, \infty)$
		\[
		\E[(g'(x) + g''(x))^q] \geq \E[g'(x)^q] + \E[g''(x)^q],
		\]
		and hence 2. and 3. imply that
		\begin{equation}
			\enorm{q}{g}^q = \E[|g(x)|^q] \geq (e / (10c^2)^{q/2}) \sum_j \alpha_{i_j}.
			\label{eq:bbh:1}
		\end{equation}
		Using 1., we conclude that for $e \geq 2c^{q} \cdot (10c^2)^{q/2} \cdot (10 / \lambda)^q$, the right-hand side of~\eqref{eq:bbh:1} will be larger than $(10 / \lambda)^q \alpha$.
		In particular, we set $d = d(c, q) = 2c^{q} \cdot (10c^2)^{q/2} \cdot 10^q$.

		We find the indices $i_1, \dots, i_J$ iteratively. We let $\calI$ be initially the set $\{0, ..., I \}$ of all indices. For $j = 1, 2, \dots,$ we do the following as long as $\calI$ is not empty:
		\begin{itemize}
			\item Let $i_j$ be the largest index in $\calI$.
			\item Remove from $\calI$ every index $i$ such that $\alpha_i \leq c^{q} \alpha_{i_j} / 2^{i - i_j}$.
		\end{itemize}

		We let $J$ denote the step we stop. Note that our indices $i_1, \dots, i_J$ are sorted in descending order. For every step $j$, the total of the $\alpha_i$'s for all indices we removed is less than $c^{q} \alpha_{i_j}$ and hence we satisfy 1. We use the following claim, whose proof is omitted here since it does not involve $q$ at all.
		This follows from the fact that $\cp(G(S)) \leq 1/(e|S|)$ for all $S$ with $\mu(S) \leq \delta$.
		\begin{claim}[Claim 8.3 of~\cite{BBHKSZ12}]
			Let $S \subseteq V$ and $\beta > 0$ such that $\mu(S) \leq \delta$ and $|f(x)| \geq \beta$ for all $x \in S$. Then there is a set of size at least $e|S|$ such that $\E_{x \in T} [g(x)^2] \geq \beta^2 / 4$.
			\label{claim:bbh:collision}
		\end{claim}

		We will construct the functions $g_1, \dots, g_J$ by applying iteratively Claim~\ref{claim:bbh:collision}.
		We do the following for $j = 1, \dots, J$:
		\begin{enumerate}
			\item let $T_j$ be the set of size $e|U_{i_j}|$ that is obtained by applying Claim~\ref{claim:bbh:collision} to the function $f$ and the set $U_{i_j}$. Note that $\E_{x \in T_j} [g(X)^2] \geq \beta_{i_j}^2 / 4$, where we let $\beta_i = c^i / \sqrt{\delta}$ (and hence for every $x \in U_i$, $\beta_i \leq |f(x)| \leq c \beta_i$).
			\item Let $g'_j$ be the function on input $x$ that outputs $\gamma \cdot |g(x)|$ if $x \in T_j$ and $0$ otherwise, where $\gamma \leq 1$ is a scaling factor that ensures that $\E_{x \in T_j} [g'(x)^2]$ equals exactly $\beta_{i_j}^2 / 4$.
			\item We define $g_j(x) = \max(0, g'_j(x) - \sum_{k < j} g_k(x))$.
		\end{enumerate}

		Note that the second step ensures $g'_j(x) \leq |g(x)|$, while the third step ensures that $g_1(x) + \dots + g_j(x) \leq g'_j(x)$ for all $j$, and in particular $g_1(x) + \dots + g_J(x) \leq |g(x)|$. Hence the only thing left to prove is the following.
		\begin{claim} [Claim 8.5 of~\cite{BBHKSZ12}]
			$\E_{x \in V} [|g_j(x)|^q] \geq e \alpha_{i_j} / (10c^2)^{q/2}$.
		\end{claim}
		\begin{proof}
			Recall that for every $i$, $\alpha_i = \mu_i \E_{x \in U_i} [|f(x)|^q]$, and hence (using $f(x) \in [\beta_i, c\beta_i)$ for $x \in U_i$):
			\begin{equation}
				\mu_i \beta_i^q \leq \alpha_i \leq \mu_i c^q \beta_i^q.
				\label{eq:bbh:2}
			\end{equation}
			Now fix $T = T_j$. Since $\E_{x \in V}[ |g_j(x)|^q ]= \mu(T)\cdot \E_{x \in T} [|g_j(x)|^q]$ and $\mu(T) = e\mu(U_{i_j})$, we can use~\eqref{eq:bbh:2} and
			$\E_{x \in T} [|g_j(x)|^q] \geq (\E_{x \in T} [g_j(x)^2])^{q/2}$ (since $q > 2$), to reduce proving the claim to showing the following:
			\begin{equation}
				\E_{x \in T} [g_j(x)^2] \geq (c \beta_{i_j})^2 / (10c^2) = \beta^2_{i_j} / 10.
				\label{eq:bbh:3}
			\end{equation}

			We know that $\E_{x \in T} [g'_j(x)^2] = \beta_{i_j}^2 / 4$. We claim that~\eqref{eq:bbh:3} will follow by showing that for every $k < j$,
			\begin{equation}
				\E_{x \in T}[ g'_k(x)^2 ]\leq 100^{-i'} \cdot \beta_{i_j}^2 / 4,
				\label{eq:bbh:4}
			\end{equation}
			where $i' = i_k - i_j$. (Note that $i' > 0$ since in our construction the indices $i_1, \dots, i_J$ are sorted in descending order.)

			Indeed,~\eqref{eq:bbh:4} means that if we let momentarily $\enorm{2}{g_j}$ denote $\sqrt{\E_{x \in T} [g_j(x)^t]}$ then
			\begin{equation}
				\enorm{2}{g_j} \geq
				\enorm{2}{g_j'} - \enorm{2}{\sum_{k < j} g_k} \geq
				\enorm{2}{g_j'} - \sum_{k < j} \enorm{2}{ g_k} \geq
				\enorm{2}{g'_j} (1 - \sum_{i' = 1}^{\infty} 10^{-i'})
				\geq 0.8 \enorm{2}{g'_j}.
				\label{eq:bbh:5}
			\end{equation}
			The first inequality holds we can write $g_j$ as $g'_j - h_j$, where $h_j = \min(g'_j, \sum_{k < j} g_k)$.
			Then, on the other hand, $\enorm{2}{g_j} \geq \enorm{2}{g'_j} - \enorm{2}{h_j}$, and on the other hand,
			$\enorm{2}{h_j} \leq \enorm{2}{\sum_{k < j} g_k}$ since $g'_j \geq 0$.
			The second inequality holds because $\enorm{2}{g_k} \leq \enorm{2}{g'_k}$. By squaring~\eqref{eq:bbh:5} and plugging in the value of $\enorm{2}{g'_j}^2$ we get~\eqref{eq:bbh:3}.
			\begin{proof}[Proof of~\eqref{eq:bbh:4}]
				By our construction, it must hold that
				\begin{equation}
					c^{q}\alpha_{i_k} / 2^{i'} \leq \alpha_{i_j},
					\label{eq:bbh:6}
				\end{equation}
				since otherwise the index $i_j$ would have been removed from the $\calI$ at the $k$th step. Since $\beta_{i_k} = \beta_{i_j} c^{i'}$, we can plug~\eqref{eq:bbh:2} in~\eqref{eq:bbh:6} to get
				\[
				\mu_{i_k} c^{q + qi'} / 2^{i'} \leq c^q \mu_{i_j}
				\]
				or
				\[
				\mu_{i_k} \leq \mu_{i_j} \cdot  2^{i'} \cdot c^{- qi'}.
				\]

				Since $|T_i| = e|U_i|$ for all $i$, it follows that $|T_k| / |T| \leq 2^{i'} \cdot c^{- qi'}$. On the other hand, we know that
				$\E_{x \in T_k} [g'_k(x)^2] = \beta^2_{i_k} / 4 = c^{2i'} \beta^2_{i_j} /4$. Thus,
				\[
				\E_{x \in T}[ g'_k(x)^2 ]\leq 2^{i'} c^{2i' - qi'} \beta^2_{i_j} /4 = (2 / c^{q-2})^{i'} \beta^2_{i_j} / 4,
				\]
				and we now just choose $c$ sufficiently large so that $2 / c^{q-2} > 100$.
			\end{proof}
		\end{proof}
	\end{proof}

	\subsection{Hardness of \texorpdfstring{$\minpq{p^*}{p}{\cdot}$}{minpq}}
	\label{subsec:grsw}
	In this subsection, we prove Theorem~\ref{thm:hardness:NP} that for any $\eps > 0$ and $p \in (2, \infty)$, it is NP-hard to approximate
	$\minpq{p^*}{p}{\cdot}$ within a factor $(\gamma_p - \eps)$, where $\gamma_p = (\E_{g \sim \mathcal{N}(0, 1)}[|g|^p])^{1/p} > 1$ is the absolute $p$th moment of the standard Gaussian.

	Our result is obtained by using the result of Guruswami et al.~\cite{GRSW16} that proved the same hardness of $\minpq{2}{p}{\cdot}$.
	When $\enorm{p}{\cdot}$ denotes expectation $p$-norm defined by $\enorm{p}{x} := \E_{i}[ |x_i|^p ]^{1/p}$, since $p^* < 2$, any $x$ satisfies $\enorm{p^*}{x} \leq \enorm{2}{x}$.
	This implies that for any matrix $A$, the optimal value of $\minpq{p^*}{p}{A}$ is at least the optimal value of $\minpq{2}{p}{A}$.
	We modify the reduction of~\cite{GRSW16} slightly such that in the Yes case, $x$ that minimizes $\minpq{2}{p}{A}$ has either $+1$ or $-1$ in each coordinate.
	This implies $\enorm{2}{x} = \enorm{p^*}{x}$, and certifies that $\minpq{p^*}{p}{A} = \minpq{2}{p}{A}$.
	In the No case, $\minpq{p^*}{p}{A}$ is always at least $\minpq{2}{p}{A}$, so the gap between the Yes case and the No case for $\minpq{p^*}{p}{\cdot}$ is at least as large as the gap for
	$\minpq{2}{p}{\cdot}$.

	Our presentation closely follows the recent work by Bhattiprolu et al.~\cite{BGGLT18}.
	To present the reduction, we introduce standard backgrounds on Fourier analysis and Label Cover problems.

	\paragraph{Fourier Analysis.}
	We introduce some basic facts about Fourier analysis of Boolean functions.
	Let $R \in \N$ be a positive integer, and consider a function $f : \{ \pm 1 \}^R \to \R$.
	For any subset $S \subseteq [R]$ let $\chi_S := \prod_{i \in S} x_i$.
	Then we can represent $f$ as
	\begin{equation}
		\label{eq:inverse_fourier}
		f(x_1, \dots, x_R) = \sum_{S \subseteq [R]} \hatf(S) \cdot \chi_S(x_1, \dots x_R),
	\end{equation}
	where
	\begin{equation}
		\label{eq:fourier}
		\hatf(S) = \E_{x \in \{ \pm 1 \}^R} [f(x) \cdot \chi_S(x)] \mbox{ for all } S \subseteq [R].
	\end{equation}
	The {\em Fourier transform} refers to a linear operator $F$ that maps $f$ to $\hatf$ as defined as~\eqref{eq:fourier}.
	We interpret $\hatf$ as a $2^R$-dimensional vector whose coordinates are indexed by $S \subseteq [R]$.
	In this subsection, we let $\cnorm{p}{\cdot}$ to denote the counting $p$-norm and $\enorm{p}{\cdot}$ to denote the expectation $p$-norm.
	Endow the expectation norm and the expectation norm to $f$ and $\hatf$ respectively; i.e.,
	\[
	\enorm{p}{f} := \left( \E_{x \in \{ \pm 1 \}^R}{|f(x)|^p} \right)^{1/p}
	\quad \mbox{ and } \quad
	\cnorm{p}{\hatf} := \left( \sum_{S \subseteq [R]} | \hatf(S) |^p \right)^{1/p}.
	\]
	as well as the corresponding inner products $\mysmalldot{f}{g}$ and $\mysmalldot{\hatf}{\hatg}$ consistent with their $2$-norms.
	We also define the {\em inverse Fourier transform} $F^T$ to be a linear operator
	that maps a given $\hatf : 2^R \to \R$ to $f : \{ \pm 1 \}^R \to \R$ defined as in~\eqref{eq:inverse_fourier}.
	We state the following well-known facts from Fourier analysis.
	\begin{observation} [Parseval's Theorem]
		For any $f : \{ \pm 1 \}^R \to \R$, $\enorm{2}{f} = \cnorm{2}{F f}$.
	\end{observation}
	\begin{observation} $F$ and $F^T$ form an adjoint pair; i.e., for any $f : \{ \pm 1 \}^R \to \R$
		and $\hatg : 2^R \to \R$,
		\[
		\mysmalldot{\hatg}{Ff} =
		\mysmalldot{F^T \hatg}{f}.
		\]
	\end{observation}
	\begin{observation}
		$F^T F$ is the identity operator.
	\end{observation}

	\paragraph{Smooth Label Cover.}
	An instance of Label Cover is given by a quadruple $\calL = (G, [R], [L], \Sigma)$ that consists of a regular connected graph $G = (V, E)$, a label set $[R]$ for some positive integer $n$, and a collection $\Sigma = ((\pi_{e, v}, \pi_{e, w}) : e = (v, w) \in E)$ of pairs of maps both from $[R]$ to $[L]$ associated with the endpoints of the edges in $E$. Given a {\em labeling} $\ell : V \to [R]$, we say that an edge $e = (v, w) \in E$ is {\em satisfied } if $\pi_{e, v}(\ell(v)) = \pi_{e, w}(\ell(w))$. Let $\opt(\calL)$ be the maximum fraction of satisfied edges by any labeling.

	The following hardness result for Label Cover, given in~\cite{GRSW16}, is a slight variant of the original construction due to~\cite{Khot02}. The theorem also describes the various structural properties, including smoothness, that are identified by the hard instances.
	\begin{thm}
		\label{thm:smooth_label_cover}
		For any $\xi >0$ and $J \in \N$, there exist positive integers $R = R(\xi, J), L = L(\xi, J)$ and $D = D(\xi)$, and a Label Cover instance $(G, [R], [L], \Sigma)$ as above such that
		\begin{itemize}
			\item (Hardness): It is NP-hard to distinguish between the following two cases:
			\begin{itemize}
				\item Yes case: $\opt(\calL) = 1$.
				\item No case: $\opt(\calL) \leq \xi$.
			\end{itemize}

			\item (Structural Properties):
			\begin{itemize}
				\item ($J$-Smoothness): For every vertex $v \in V$ and distinct $i, j \in [R]$, we have
				\[
				{\text{\bf Pr}}_{e \ni v} \Big[ \pi_{e,v}(i)=\pi_{e,v}(j) \Big] \leq 1 / J.
				\]
				\item ($D$-to-$1$): For every vertex $v \in V$, edge $e \in E$ incident on $v$, and $i \in [L]$, we have $|\pi^{-1}_{e, v}(i)| \leq D$; that is at most $D$ elements in $[R]$ are mapped to the same element in $[L]$.

				\item (Weak Expansion): For any $\delta > 0$ and vertex set $V' \subseteq V$ such that $|V'| = \delta \cdot |V|$, the number of edges among the vertices in $|V'|$ is at least $(\delta^2 / 2) |E|$.
			\end{itemize}

		\end{itemize}
	\end{thm}

	\paragraph{Reduction.}
	Let $\calL = (G, [R], [L], \Sigma)$ be an
	instance of Label Cover with $G = (V, E)$.
	Our reduction will construct a linear operator
	$\bfA : \R^{N} \to \R^{M}$ with $N = |V| \cdot 2^R$ and $M = 2|V| \cdot 2^R - |V| + |E| \cdot |L|$.
	The space $\R^{N}$ will be endowed the expectation norm (and call its elements functions) and
	$\R^{M}$ will be endowed the counting norm (and call its elements vectors).
	We define $\bfA$ by giving a linear transformation from a function $\bff : V \times \{ \pm 1 \}^R \to \R$ to a vector $\bfa \in \R^M$.
	Let $C := M^3$.
	Given $\bff$, a vertex $v \in V$ induces $f_v \in \R^{2^R}$ defined by $f_v(x) := \bff(v, x)$ for $x \in \{\pm 1\}^R$.
	Let $\bfhatg \in V \times [R]$ be the vectors of linear coefficients; $\bfhatg(v, i) = \hatf_v(i)$ for $v \in V, i \in [R]$.
	Given $\bff$ (that determines $\{ \hatf_v \}_v \in V$ and $\bfhatg$), $\bfa = \bfA \bff$ is defined as follows.

	\begin{itemize}
		\item For $v \in V$ and $x \in \{ \pm 1 \}^R$, $\bfa(v, x) = \sum_{i = 1}^R \bfhatg(v, i) x_i$.
		\item For $v \in V$ and $S \subseteq [R]$ with $|S| \neq 1$, $\bfa(v, S) = C \cdot \hat{f_v}(S)$.
		\item For $e = (u, v) \in E$ and $i \in [L]$, $\bfa(e, i) = C \cdot \big( \sum_{j \in \pi^{-1}_{e, u}(i)} \hat{f_u}(i) - \sum_{j \in \pi^{-1}_{e, v}(i)} \hat{f_v}(i) \big)$.
	\end{itemize}
	Since $\bfhatg$ and $\bfa$ are all linear in $\bff$, the matrix $\bfA$ that satisfies $\bfa = \bfA \bff$ is well-defined, which is our instance of $\minpq{p^*}{p}{\cdot}$.
	Intuitively, $C$ will be chosen large enough so that every $\hat{f_v}$ has almost all Fourier mass on its linear coefficients,
	and their linear coefficients correctly indicate the labels that satisfy all constraints of the Label Cover instance.

	\paragraph{Completeness.}
	We prove the following lemma for the Yes case.

	\begin{lem} [Completeness] Let $\ell : V \to [R]$ be a labeling that satisfies every
		edge of $\calL$. There exists a function $\bff \in \R^{V \times 2^R}$ such that
		$\bff(v, x)$ is either $+1$ or $-1$ for all $v \in V, x \in \{ \pm 1 \}^R$ and
		$\cnorm{p}{\bfA \bff} = (|V|\cdot 2^R)^{1/p}$.
		In particular, $\cnorm{p}{\bfA \bff} / \enorm{p^*}{\bff} = (|V|\cdot 2^R)^{1/p}$.
	\end{lem}
	\begin{proof} Let $\bff(v, x) := x_{\ell(v)}$ for
		every $v \in V, x \in \{ \pm 1 \}^R$.  Consider $\bfa = \bfA \bff$.
		Since every $\hatf_v$ is linear, for each $v \in V$ and $S \subseteq [R]$ with $|S| \neq 1$, $\bfa(v, S) = 0$.
		For each $v \in V$ and $i \in [R]$, $\bfhatg(v, i) = 1$ if and only if $i = \ell(v)$ and $0$ otherwise.
		Since $\ell$ satisfies every edge of $\calL$, $\bfa(e, i) = 0$ for every $e \in E$ and $i \in [L]$.
		This implies that for every $v \in V, x \in \{ \pm 1 \}^R$, $\bfa(v, x) = x_{\ell(v)} = \bff(v, x)$.
		Therefore,  $\cnorm{p}{\bfA \bff} = (|V| \cdot 2^R)^{1/p}$.
	\end{proof}

	\paragraph{Soundness.}
	We prove the following lemma for the soundness.
	Combined with Theorem~\ref{thm:smooth_label_cover} for hardness of Label Cover and observing that $\enorm{p^*}{\bff} \leq \enorm{2}{\bff}$, it finishes the proof of Theorem~\ref{thm:hardness:NP}.
	\begin{lem}
		For any $\eta > 0$, there exists $\xi > 0$ (that determines $D = D(\xi)$ as in Theorem~\ref{thm:smooth_label_cover}) and $J \in \N$ such that if $\opt(\calL) \leq \xi$, $\calL$ is $D$-to-$1$ and $\calL$ is $J$-smooth, for every $\bff$ with $\enorm{2}{\bff} = 1$, $\cnorm{p}{\bfA \bff} \geq (\gamma_p - \eta) (|V| \cdot 2^R)^{1/p}$.
	\end{lem}
	\begin{proof}
		We will prove contrapositive; if $\cnorm{p}{\bfA \bff} \leq (\gamma_p - \eta) (|V| \cdot 2^R)^{1/p}$ for some $\bff$ is small then $\opt(\calL) \geq \xi$ with the choice of the parameters that will determined later.
		Fix such an $\bff$ with $\enorm{2}{\bff} = 1$ that determines $f_v$ and $\hat{f_v}$ for each $v \in V$.
		Let $\bfa = \bfA \bff$.
		Suppose that there is $v \in V$ and $S \subseteq [R]$ with $|S| \neq 1$ such that $|\hat{f_v}(S)| > 1/M^2$.
		It means that $|\bfa(v, S)| > C / M^2$. Since $C = M^3$, it already implies $\cnorm{p}{\bfa} \geq M \gg (\gamma_p - \eta) (|V|\cdot 2^R)^{1/p}$,
		so suppose that there is no such $v$ and $S$.

		Let $\bfhatg \in V \times [R]$ be defined as above; $\bfhatg(v, i) = \hatf_v(i)$ for $v \in V, i \in [R]$.
		By Parseval,
		\[
		\sum_{v \in V} \cnorm{2}{\hatf_v}^2 = \sum_{v \in V} \enorm{2}{f_v}^2 = |V| \cdot \E_{v \in V} \enorm{2}{f_v}^2 = |V|\cdot \enorm{2}{\bff}^2 = |V|.
		\]
		and the fact that $|\hat{f}_v(S)| < 1/M^2$ for every $v \in V$, $S \subseteq [R]$ with $|S| \neq 1$,
		we have $\cnorm{2}{\bfhatg} \in [\sqrt{|V|} - 1/M, \sqrt{|V|}]$.

		Furthermore,
		suppose that there is $e = (u, v) \in E$ and $i \in [L]$ such that
		\[
		\bigg|\sum_{j \in \pi^{-1}_{e, u}(i)} \bfhatg(u, i) - \sum_{j \in \pi^{-1}_{e, v}(i)} \bfhatg(v, j) \bigg| \geq 1 / M^2.
		\]
		This implies that $|\bfa(e, i)| \geq C / M^2$.
		Since $C = M^3$, it already implies $\cnorm{p}{\bfa} \geq M \gg (\gamma_p - \eta) (|V|\cdot 2^R)^{1/p}$,
		so we can assume that there is no such $e$ and $i$.

		To bound $\cnorm{p}{\bfa}$, it only remains to analyze
		\begin{equation}
			\sum_{v \in V} \sum_{x \in \{ \pm 1 \}^R} \bigg|\bfa(v, x) \bigg|^p
			=
			\sum_{v \in V} \sum_{x \in \{ \pm 1 \}^R} \bigg|\sum_{i \in [R]} \bfhatg(v, i) x_i \bigg|^p.
			\label{eq:gammap}
		\end{equation}
		The rest of the proof closely follows~\cite{GRSW16}, and we explain high-level intuitions and why their proofs work in our settings.
		First, let us assume that $\cnorm{2}{\bfhatg} = \sqrt{|V|}$. It involves a multiplicative error of $(1 - 1/M)$, which is negligible in our proof.
		To simplify notations, let $\hatg_v \in \R^{R}$ be such that $\hatg_v(i) := \hatf_v(\{ i \} ) = \bfhatg(v, i)$ for each $v \in V$ and $i \in [R]$.
		Call a vertex $v \in V$ {\em $\tau$-irregular} if there exists $i \in [R]$ such that $|\bfhatg(v, i)| > \tau \cnorm{2}{\hatg_v}^2$. If not, $v$ is {\em $\tau$-regular}.
		Also, call a vertex $v \in V$ {\em small} if $\cnorm{2}{\hatg_v} < 1 / M$. Otherwise, call it {\em big}.

		For each $v \in V$, we consider
		$
		\sum_{x \in \{ \pm 1 \}^R} \big|\sum_{i \in [R]} \hatg(v, i) x_i \big|^p.
		$
		By Khintchine inequality, it is at most $2^R \cdot \gamma_p^p \cdot \cnorm{2}{\hatg}^p$.
		The following lemma, based on standard applications of the Berry-Esseen theorem, shows that the converse is almost true when $v$ is $\tau$-regular,
		implying the contribution from irregular vertices to~\eqref{eq:gammap} is large.
		\begin{lem}[\cite{KNS10}]
			For sufficiently small $\tau$ (depending only on $p$), if $v \in V$ is $\tau$-regular, then
			\[
			\sum_{x \in \{ \pm 1 \}^R} \bigg|\sum_{i \in [R]} \hatg(v, i) x_i \bigg|^p \geq 2^R \cdot \gamma_p^p \cdot \cnorm{2}{\hatg}^p (1 - \sqrt{\tau}).
			\]
			\label{lem:kns}
		\end{lem}
		Let $S$ be the set of big $\tau$-irregular vertices.
		Based on the above, the following lemma shows that $S$ must be a large set.
		Originally,~\cite{GRSW16} only argued for $\tau$-irregular vertices. (The notion of big and small vertices does not appear there.)
		However, since the contribution of small vertices to~\eqref{eq:gammap} is negligible, the same proof essentially works.

		\begin{lem}[Lemma 4.4 of~\cite{GRSW16}]
			There are $\tau$ and $\theta$, depending only on $p$ and $\eta$, such that $S$, the set of big $\tau$-irregular vertices, satisfies $|S| \geq \theta |V|$.
			\label{lem:grsw1}
		\end{lem}

		By the weak expansion property of $\calL$ guaranteed in Theorem~\ref{thm:smooth_label_cover}, $S$ induces at least $\theta^2 |E|$ edges of $\calL$.
		To finish the proof,~\cite{GRSW16} showed that we can satisfy a significant fraction of the edges from $\calL$.
		The only difference in their setting and our setting is that
		\begin{itemize}
			\item \cite{GRSW16}: $S$ is the set of all $\tau$-irregular veritces. For each $e = (u, v)$ and $i \in [L]$,
			\begin{equation}
				\sum_{j \in \pi^{-1}_{e, u}(i)} \hatg_u(j) = \sum_{j \in \pi^{-1}_{e, v}(i)} \hatg_v(j).
				\label{eq:grsw:lc}
			\end{equation}
			\item Here: $S$ is the set of all big $\tau$-irregular veritces. For each $e = (u, v)$ and $i \in [L]$,
			\begin{equation}
				\bigg| \sum_{j \in \pi^{-1}_{e, u}(i)} \hatg_u(j) - \sum_{j \in \pi^{-1}_{e, v}(i)} \hatg_v(j) \bigg| < 1/M^2.
				\label{eq:grsw:lc2}
			\end{equation}
		\end{itemize}
		These differences do not affect their proof since in the only place~\eqref{eq:grsw:lc} was used for $e = (u, v)$ and $i \in [L]$,
		they indeed used the fact the left-hand side of~\eqref{eq:grsw:lc2} is at most $0.3\tau \cdot \max(\cnorm{2}{\hatg_u}, \cnorm{2}{\hatg_v})$.
		Since we additionally assumed that $S$ is big, $\cnorm{2}{\hatg_u} \geq 1/M$ for each $u \in S$, so it is always satisfied from~\eqref{eq:grsw:lc2}.

		\begin{lem}[\cite{GRSW16}]
			Let $\beta := 10000D^4 / \tau^4 J$. Then
			$\opt(\calL) \geq (\tau^4 / 16)  (\theta^2 - 2/\beta)$.
		\end{lem}
		Since $\theta$ and $\tau$ only depend on $\eta$ and $p$,
		fixing small enough $\xi$ (that determines $D$) and large enough $J$ will ensure
		$\opt(\calL) \geq (\tau^4 / 16)  (\theta^2 - 2/\beta) \geq \xi$, finishing the proof of the lemma.
	\end{proof}

	\subsection{Hardness for Finite Fields}
	\label{subsec:hardness_fields}
	In this subsection, we prove Lemma~\ref{lem:hardness_fields}, which in turn finishes the proof of Theorem~\ref{thm:hardness_fields} for hardness of $\ell_0$-row lank approximation for matrices whose entries are from a finite field $\F$.
	\begin{lem}[Restatement of Lemma~\ref{lem:hardness_fields}]
		Let $\F$ be a finite field and $A \in \F^{n \times d}$ with $n \geq d$ and $k = d - 1$.
		Then \[
		\min_{U \in \F^{n \times k}, V \in \F^{k \times d}} \norm{0}{U V - A} = \min_{x \in \F^d, x \neq 0} \norm{0}{Ax}.
		\]
	\end{lem}
	\begin{proof}
		Assume that the rank of $A$ is $d$; otherwise the lemma becomes trivial.
		We first prove $(\geq)$.
		Given $V^* \in \F^{k \times d}$ that achieves the best rank $k$ approximation, assume without loss of generality that the rank of $V^*$ is $k = d - 1$.
		Let $x \in \F^d$ be a nonzero vector that is orthogonal to the rowspace of $V$; i.e., $\langle v, x \rangle = 0$ if and only if $v \in \rowspace(V)$.
		Note that unlike in $\R$, $x$ can be in $\rowspace(V)$, but it does not affect the proof.
		Let $a_1, \dots, a_n$ be the rows of $A$.
		For fixed $V^*$ and $i \in [n]$, if $a_i \in \rowspace(V)$, then we can compute the $i$th row of $U^*$ (denoted by $u^*_i$) such that $u^*_i V^* = a_i$.
		Otherwise, $\langle a_i, x \rangle = b$ for some $b \neq 0$, since $x$ is nonzero, there is $u^*_i$ such that $\norm{0}{u^*_i V^* = a_i} = 1$.
		Therefore, $\norm{0}{U^* V^* - A} = \norm{0}{Ax}$, which implies that
		$\min_{U \in \F^{n \times k}, V \in \F^{k \times d}} \norm{0}{U V - A} \geq \min_{x \in \F^d, x \neq 0} \norm{0}{Ax}$.

		For the other direction, given $x \in \F^d \setminus \{ 0 \}$,
		the set of vectors $u$ with $\langle u, x \rangle = 0$ forms a $k$-dimensional subspace. (Again, this space may contain $x$ unlike in $\R$, but it does not matter.) Let $V^* \in \R^{k \times d}$ be a matrix whose rows span that space, and compute $U^*$ as above.
		The above analysis shows that $\norm{0}{U^*V^* - A}  = \norm{0}{Ax}$, which completes the proof.
	\end{proof}

	\newpage
	\section{Additional Results} \label{sec:misc}

	Here we list some additional results on variants of the $\ell_p$ low rank approximation problem.

	\subsection{Bicriteria Algorithm}

	In this section we show that we can develop low rank approximations that apply to matrices whose entries are not bounded by $\text{poly}(n)$ so long as we accept bicriteria algorithms. That is, instead of a target rank $k$ approximation, the algorithm will output an approximating matrix of rank $3k$.

	\begin{thm} If $A$ is an $n \times d$ matrix, our target rank $k$ is a constant, and $1 \leq p < 2$, then there exists a polynomial time algorithm that outputs a matrix $M$ of rank at most $3k$ such that $\lVert M - A \rVert_p \leq (1 + \epsilon) OPT $ where $OPT$ is the best rank $k$ $\ell_p$-low rank approximation value for $A$ with probability $1 - O(1)$. \end{thm}

	\begin{proof}
		Let $C_{l}$ denote the best rank $l$ approximation to a matrix $C$ in the $\ell_p$ norm (i.e. the matrix that minimizes $\lVert C_l - C \rVert_p$).

		Let $B$ be the best rank $k$ approximation to $A$ in the Frobenius norm. Then $$\lVert A - B \rVert_p \leq \text{poly}(n) \lVert A - B \rVert_F \leq \text{poly}(n) \lVert A - A_k \rVert_F \leq \text{poly} (n) OPT.$$

		We can find a rank $2k$ $(1 + \epsilon)$-approximation to $A - B$ using the same techniques as in Theorem \ref{approxalg}, where we sample a matrix $S$ of $p$-stable variables, guess values for $SU^*$, and then minimize $\lVert SU^* V^* - S(A - B) \rVert_p$. Now the entries of $A - B$ are not necessarily bounded by $\text{poly}(n)$ so we need to justify that it suffices to guess $\text{poly}(n)$ values for $SU^*$.

		Indeed, by a well-conditioned basis argument, no entry of $U^*$ has absolute value greater than $\text{poly}(n) \lVert A - B \rVert_p$. Furthermore, we can round each entry of $U^*$ (similar to the proof of Theorem \ref{approxalg}) to the nearest multiple of $\frac{\epsilon \lVert A - B \rVert_p}{\text{poly}(n)}$ and incur an additive error of at most $\epsilon OPT$ because $\lVert A - B \rVert_p \leq \text{poly}(n) OPT$. This error is small enough for the purposes of our approximation.

		Let $(A - B)^*_{2k} = U^* V^*$ and let $M = (A - B)^*_{2k} + B$. We have \begin{align*} \lVert M - A \rVert_p &= \lVert (A - B)^*_{2k} + B - A \rVert_p \\
			&= \lVert (A - B)^*_{2k} - (A - B) \rVert_p \\
			&\leq (1 + \epsilon) \lVert (A - B)_{2k} - (A - B) \rVert_p + \epsilon OPT \\
			&\leq (1 + \epsilon) \lVert A_{k} - B - (A - B) \rVert_p + \epsilon OPT \\
			&\leq (1 + \epsilon) OPT
		\end{align*} where the first inequality follows from our argument above and the second inequality follows because $A_k - B$ has rank at most $k + k = 2k$.

		Since $M$ has rank at most $2k + k = 3k$, then the result follows.
	\end{proof}

	\subsection{Weighted Low Rank Approximation}

	For $0 < p < 2$, we can also design a PTAS for the weighted $\ell_p$ low rank approximation problem. In this setting we have a matrix $A$, a weight matrix $W$ of rank $r$, and we want to output a rank $k$ matrix $A'$ such that, for $\epsilon > 0$,
	\[
	\lVert W \circ (A - A') \rVert_p^p \leq
	(1 + \epsilon) \min_{\text{rank k }A_k} \lVert W \circ (A - A_k) \rVert_p^p.
	\]

	Our main tool will be a multiple regression concentration result based on that of \cite{rsw16}.

	\begin{thm} \label{multipleregression}
		Let $S$ be a $\text{poly}(k / \epsilon) \times n$ matrix whose entries are i.i.d $p$-stable random variables with scale 1. Let $M^{(1)}, M^{(2)}, \ldots, M^{(m)}$ be $n \times d$ matrices and let $b^{(1)}, b^{(2)}, \ldots, b^{(m)} \in \R^n$. Let $$x^{(i)} = \argmin_x \lVert M^{(i)} x - b^{(i)} \rVert_p^p$$ and $$y^{(i)} = \argmin_y \med (SM^{(i)} y - Sb^{(i)}) / \med_p.$$ Then w.h.p we have $$\sum_i \lVert M^{(i)} y^{(i)} - b^{(i)} \rVert_p^p \leq (1 + O(\epsilon)) \sum_i \lVert M^{(i)} x^{(i)} - b^{(i)} \rVert_p^p $$
	\end{thm}

	\begin{proof}
		By Lemmas \ref{pfixedbound} and \ref{smallpfixedbound}, w.h.p. $$ \sum_i \frac{\med (S (M^{(i)} x^{(i)} - b^{(i)}))^p}{\med_p^p} \leq (1 + O(\epsilon)) \sum_i \lVert M^{(i)} x^{(i)} - b^{(i)} \rVert_p^p .$$

		Let $T$ be the set of all $i$ such that $$\frac{\med(S [M^{(i)} \ b^{(i)}] y)^p}{\med_p^p} \geq (1 - \Theta(\epsilon)) \lVert [M^{(i)} \ b^{(i)}] y \rVert_p^p$$ for all $y$. By Corollary \ref{pmedembed}, we know that for each $i$, the probability that $i \in T$ is at least $1 - \Theta(\epsilon)$.

		Thus $$ \ex{\sum_{i \notin T}  \lVert [M^{(i)} \ b^{(i)}] y \rVert_p^p} \leq \Theta(\epsilon) \sum_i \lVert [M^{(i)} \ b^{(i)}] y \rVert_p^p $$ so by Markov's inequality, w.h.p we have $$\sum_{i \notin T}  \lVert [M^{(i)} \ b^{(i)}] y \rVert_p^p  \leq \Theta(\epsilon) \sum_i \lVert [M^{(i)} \ b^{(i)}] y \rVert_p^p .$$

		Let $y$ be arbitrary. Since $$\sum_i \frac{\med(S [M^{(i)} \ b^{(i)}] y)^p}{\med_p^p} \geq \sum_{i \notin T} \frac{\med(S [M^{(i)} \ b^{(i)}] y)^p}{\med_p^p} \geq (1 - \Theta(\epsilon)) \sum_{i \notin T} \lVert [M^{(i)} \ b^{(i)}] y \rVert_p^p,  $$ it follows that for all $y$ we have $$\sum_i \frac{\med(S [M^{(i)} \ b^{(i)}] y)^p}{\med_p^p} \geq (1 - \Theta(\epsilon)) \sum_i \lVert [M^{(i)} \ b^{(i)}] y \rVert_p^p. $$

		Therefore w.h.p we have \begin{eqnarray*}
			&&(1 - \Theta(\epsilon)) \sum_i \lVert M^{(i)} y^{(i)} - b^{(i)} \rVert_p^p \leq \sum_i \frac{S(M^{(i)} y^{(i)} - b^{(i)})}{\med_p^p} \\
			&\leq& \sum_i \frac{S(M^{(i)} x^{(i)} - b^{(i)})}{\med_p^p} \leq (1 + O(\epsilon)) \sum_i \lVert M^{(i)} x^{(i)} - b^{(i)} \rVert_p^p
		\end{eqnarray*} because $0 < p < 2$. The result follows.

	\end{proof}

	\begin{thm} \label{thm:Wp02}
		Suppose $A$ and $W$ are $n \times d$ matrices with entries bounded by $\text{poly}(n)$,
		and $r=\mathrm{rank}(W)$.
		There is an algorithm that for any integer $k$, $p\in(0,2)$ and $\eps\in(0,1)$,
		outputs in time $n^{r \cdot \poly(k / \eps)}$ a $n \times k$ matrix $U^*$ and
		a $k \times d$ matrix $V^*$ such that
		\[
		\lVert W \circ (A - U^*V^*) \rVert_p^p \leq
		(1+ O(\eps))\min_{\text{rank-$k$ }A_k} \lVert W \circ (A - A_k) \rVert_p^p.
		\]
	\end{thm}

	\begin{proof}
		To achieve a relative-error low rank approximation $W \circ (UV - A)$, for each column $i$ we can guess sketches for $W_{:,i} \circ UV_i$ using a similar argument as in Theorem \ref{approxalg}. Indeed, we can apply Theorem \ref{multipleregression} with $M^{(i)} = W_{:,i} \circ U^* V^*_i$ and $b^{(i)} = W_{:,i} \circ A_{:,i}$. To do so, we need to be able to guess $S W_{:,i} \circ U^*$, a $\poly(\frac{k}{\eps}) \times k$ matrix, in $\text{poly} (n)$ tries. We will follow the same reasoning as in the proof of Theorem \ref{approxalg}. Since the entries of $W$ and $A$ are bounded by $\text{poly} (n)$, then we can bound the entries of $U^*$ by $\text{poly}(n)$ using a well-conditioned basis. Furthermore, we can round each entry of $U^*$ to the nearest multiple of $\text{poly} (n^{-1})$ while incurring an error factor of only $(1 + O(\eps))$. Thus, we need only $n^{\poly(k / \eps)}$ guesses.

		Of course, there $d$ columns so this is not enough to achieve a PTAS. However, we only need to guess sketches for $r$ values of $j$ because $W$ has rank $r$ so we can express any column of $W$ as a linear combination of those $r$ columns. That is, we choose a subset $S$ of the columns such that $|S| = r$ and guess the sketches of $W_{:,i} \circ UV_i$ for each $i \in S$ as described in the previous paragraph. Therefore, we require $n^{r \cdot \poly(k / \eps)}$ time in total for a $(1 + O(\eps))$ approximation algorithm. Since $k$ and $r$ are constants this results in a PTAS.
	\end{proof}

	\newpage
	\section*{Acknowledgements}
	We thank Uriel Feige for helpful discussions and for pointing us to the work of Alon and Sudakov~\cite{AlonS99}.
	We would also like to thank Luca Trevisan, Christos Papadimitriou, and Michael Mahoney for some useful discussions.


\end{document}